\pdfoutput=1
\documentclass[11pt]{article} 
\def\arXiv{1}  
\newcommand{\notarxiv}[1]{foo}
\newcommand{\arxiv}[1]{ba}
\ifdefined \arXiv
\renewcommand{\arxiv}[1]{#1}%
\renewcommand{\notarxiv}[1]{\ignorespaces}%
\else%
\renewcommand{\arxiv}[1]{\ignorespaces}%
\renewcommand{\notarxiv}[1]{#1}%
\fi%

\notarxiv{
  \usepackage[letterpaper]{geometry}
  \usepackage{ltexpprt}
  \usepackage{hyperref}
}

\usepackage{enumerate}
\arxiv{
\usepackage{amssymb}
\usepackage{amsbsy}
\usepackage{amsmath}
\usepackage{amsthm}
\usepackage{amsfonts}       %
}
\usepackage{amsfonts}
\usepackage{latexsym}
\usepackage{graphicx}
\usepackage{mathtools}
\usepackage{enumitem}
\usepackage{thmtools}
\usepackage{thm-restate}
\usepackage{graphicx}
\usepackage{color}
\usepackage{bbm}
\usepackage{xifthen}
\usepackage{xspace}
\usepackage{mathtools}
\usepackage{titlesec}
\usepackage{setspace}
\usepackage{etoolbox}
\usepackage{xfrac}
\usepackage{esvect}
\usepackage{nicefrac}
\usepackage{cases}
\usepackage{empheq}
\usepackage{booktabs}
\usepackage[table,xcdraw]{xcolor}
\usepackage{multirow}
\usepackage{xparse}
\usepackage{caption}
\usepackage{subcaption}
\usepackage{bbm}
\usepackage{placeins}
\usepackage{pifont}

\usepackage{nicefrac}       %
\usepackage{microtype}      %
\usepackage[numbers, square]{natbib}

\newcommand\blfootnote[1]{%
    \bgroup
    \renewcommand\thefootnote{\fnsymbol{footnote}}%
    \renewcommand\thempfootnote{\fnsymbol{mpfootnote}}%
    \footnotetext[0]{#1}%
    \egroup
}

\arxiv{
	\usepackage[top=1in, right=1in, left=1in, bottom=1in]{geometry}
	\usepackage{url}
	\usepackage[hidelinks]{hyperref}
	\hypersetup{
		colorlinks=true,
		linkcolor=blue!70!black,
		citecolor=blue!70!black,
		urlcolor=blue!70!black}
}

\usepackage[linesnumbered,ruled,vlined,algo2e]{algorithm2e}

\usepackage{cleveref}

\definecolor{darkblue}{rgb}{0,0,.75}

\renewcommand{\checkmark}{\ding{52}}%
\newcommand{\xmark}{\ding{55}}%

\newcommand{\mc}[1]{\mathcal{#1}}

\DeclarePairedDelimiter{\abs}{\lvert}{\rvert} %
\DeclarePairedDelimiter{\brk}{[}{]}
\DeclarePairedDelimiter{\crl}{\{}{\}}
\DeclarePairedDelimiter{\prn}{(}{)}

\DeclarePairedDelimiter{\norm}{\|}{\|}

\DeclarePairedDelimiter{\ceil}{\lceil}{\rceil}

\DeclareDocumentCommand\opnorm{ s o m }{%
	\IfBooleanTF{#1}{%
		\norm*{#3}_{\mathrm{op}}
	}{%
		\IfNoValueTF{#2}{%
			\norm{#3}_{\mathrm{op}}
		}{%
			\norm[#2]{#3}_{\mathrm{op}}
		}
	}
}

\newcommand{\inner}[2]{\left<#1,#2\right>}

\newcommand{\overle}[1]{\overset{#1}{\le}}

\NewDocumentCommand\Ex{s O{} m }{%
	\mathbb{E}%
	\begingroup
	\IfBooleanTF{#1}
	{\ExInn*{#3}}
	{\ExInn[#2]{#3}}%
	\endgroup
}

\DeclarePairedDelimiterX\ExInn[1]{[}{]}{%
	\activatebar
	#1%
}

\RenewDocumentCommand\Pr{sO{}r()}{%
	\mathbb{P}%
	\begingroup
	\IfBooleanTF{#1}
	{\PrInn*{#3}}
	{\PrInn[#2]{#3}}%
	\endgroup
}

\DeclarePairedDelimiterX\PrInn[1](){%
	\activatebar
	#1%
}

\newcommand{\activatebar}{%
	\begingroup\lccode`~=`|
	\lowercase{\endgroup\def~}{\;\delimsize\vert\;}%
	\mathcode`|=\string"8000
}

\newcommand\numberthis{\addtocounter{equation}{1}\tag{\theequation}}

\newcommand{\lone}[1]{\norm{#1}_1} %
\newcommand{\normbigg}[1]{\bigg\|{#1}\bigg\|} %

\newcommand{\defeq}{\coloneqq}

\newcommand{\half}{\frac{1}{2}}
\newcommand{\indic}[1]{\mathbbm{1}_{\!\left\{#1\right\}}} %

\newcommand{\pind}[1]{^{(#1)}}

\newcommand{\R}{\mathbb{R}}
\newcommand{\Z}{\mathbb{Z}}
\makeatletter
\long\def\@makecaption#1#2{
  \vskip 0.8ex
  \setbox\@tempboxa\hbox{\small {\bf #1:} #2}
  \parindent 1.5em  %
  \dimen0=\hsize
  \advance\dimen0 by -3em
  \ifdim \wd\@tempboxa >\dimen0
  \hbox to \hsize{
    \parindent 0em
    \hfil 
    \parbox{\dimen0}{\def\baselinestretch{0.96}\small
      {\bf #1.} #2
    } 
    \hfil}
  \else \hbox to \hsize{\hfil \box\@tempboxa \hfil}
  \fi
}
\makeatother

\definecolor{innerboxcolor}{rgb}{.9,.95,1}
\definecolor{outerlinecolor}{rgb}{.6,0,.2}

\newcommand{\E}{\mathbb{E}} %
\renewcommand{\P}{\mathbb{P}} %

\newcommand{\ball}{\mathbb{B}}

\providecommand{\argmax}{\mathop{\rm argmax}} %
\providecommand{\argmin}{\mathop{\rm argmin}}

\providecommand{\abs}{\mathop{\rm abs}}

\providecommand{\sign}{\mathop{\rm sign}}

\providecommand{\minimize}{\mathop{\rm minimize}}
\providecommand{\maximize}{\mathop{\rm maximize}}

\arxiv{
\newtheorem{theorem}{Theorem}
\newtheorem{lemma}{Lemma}
\newtheorem{corollary}[theorem]{Corollary}

\newtheorem{proposition}{Proposition}
\newtheorem{definition}{Definition}

\numberwithin{theorem}{section}
\numberwithin{lemma}{section}
\numberwithin{corollary}{section}
\numberwithin{proposition}{section}
\numberwithin{definition}{section}

\newcounter{example}

\newenvironment{example}[1][]{
  \refstepcounter{example}
  \ifthenelse{\isempty{#1}}{%
    \noindent \textbf{Example \theexample:}\hspace*{.05em}
  }{%
    \noindent \textbf{Example \theexample} ({#1})\textbf{:}\hspace*{.05em}
  }
}{%
  $\Diamond$ \medskip
}

\newenvironment{example*}[1][]{
  \ifthenelse{\isempty{#1}}{%
    \noindent \textbf{Example:}\hspace*{.05em}
  }{%
    \noindent \textbf{Example} ({#1})\textbf{:}\hspace*{.05em}
  }
}{%
  $\Diamond$ \bigskip
}
\numberwithin{equation}{section}
\numberwithin{example}{section}
}

\notarxiv{
  \crefname{@theorem}{theorem}{theorems}
  \Crefname{@theorem}{Theorem}{Theorems}
  \newtheorem{definition}{Definition}
  \numberwithin{definition}{section}
  \newtheorem{example}{Example}
  \numberwithin{example}{section}
}

\newcommand{\otilde}{\widetilde{O}}

\newcommand{\Otil}[1]{\widetilde{O}( #1 )}

\newcommand{\Thtil}[1]{\widetilde{\Theta}( #1 )}
\newcommand{\Otilb}[1]{\widetilde{O}\left( #1 \right)}
\newcommand{\poly}{\mathsf{poly}}
    
\newcommand{\eps}{\epsilon}
\newcommand{\grad}{\nabla}

\newcommand{\del}{\partial}
\global\long\def\runtime{\mathcal{T}}%
\newcommand{\xopt}{x_\star}
\newcommand{\Ropt}{R_\star}

\newcommand{\veps}{\varepsilon}

\SetCommentSty{mycommfont}
\SetKwInput{KwInput}{Input}                %
\SetKwInput{KwOutput}{Output}              %
\SetKwInput{KwReturn}{Return}              %
\SetKwInput{KwParameters}{Parameters}
\SetKwInput{KwInput}{Input} 
\SetKwInput{KwParameter}{Parameters}                %
\SetKwProg{Fn}{function}{}{}
\SetKwInput{KwOutput}{Output}              %
\SetKwInput{KwReturn}{Return}              %
\SetKwComment{Comment}{$\triangleright$\ }{}
\SetKwRepeat{Do}{do}{while}
\SetKwBlock{Repeat}{Repeat}{}
\SetCommentSty{mycommfont}

\makeatletter
\newcommand\Block[2]{%
	#1%
	\algocf@group{#2}%
}
\makeatother

\newcommand{\linesearch}{\lambda\textsc{-Bisection}}

\newcommand{\xset}{\mathcal{X}}

\newcommand{\hg}{\hat{g}}

\newcommand{\cmax}{c_{\max}}
\newcommand{\Kmax}{K_{\max}}
\newcommand{\Tmax}{T_{\mathrm{outer}}}
\newcommand{\Tinner}{T_{\mathrm{inner}}}

\newcommand{\halfsilon}{\frac{\epsilon}{2}}

\newcommand{\fmax}{f_{\max}}
\newcommand{\fsm}{f_{\mathrm{smax}}}

\newcommand{\gest}{\mathcal{G}}

\newcommand{\tw}{\tilde{w}}

\newcommand{\mprox}{\textup{\textsf{LI-MD}}}

\newcommand{\filt}[1][t]{\mathcal{F}_{#1}}

\newcommand{\vref}[1]{#1^{\mathrm{ref}}}

\newcommand{\sketchDist}{\mathcal{M}}

\newcommand{\inprod}[2]{\left\langle#1,#2\right\rangle}

\newcommand{\nnz}{\mathrm{nnz}}
\newcommand{\dgf}{\varphi}

\newcommand{\uprox}[1][\lambda]{u_{#1}^\star}
\newcommand{\Hprox}[1][\lambda]{H_{#1}}

\newcommand{\codeStyle}[1]{\textsc{#1}}
\newcommand{\dsStyle}[1]{\textsc{#1}}
\newcommand{\dsInit}{\dsStyle{init}}
\newcommand{\dsQuery}{\dsStyle{query}}
\newcommand{\decode}{\codeStyle{decode}}

\newcommand{\GG}{\Gamma}
\newcommand{\tef}{\tau}

\newcommand{\Lf}{L_f}
\newcommand{\Lg}{L_g}

\newcommand{\bigC}{C}
\newcommand{\cE}{\mathcal{E}}
\newcommand{\flag}{\mathsf{OutOfBound}}
\newcommand{\False}{\textup{\textsf{False}}}
\newcommand{\True}{\textup{\textsf{True}}}

\usepackage[textsize=scriptsize,textwidth=2cm]{todonotes}
\newcommand{\yairside}[1]{\todo[color=blue!10]{Yair: #1}}
\newcommand{\yair}[1]{{\bf \color{blue} Yair: #1}}
\newcommand{\yujia}[1]{{\bf \color{red} Yujia: #1}}
\newcommand{\sidford}[1]{{\bf \color{purple} Sidford: #1}}
\newcommand{\arun}[1]{{\bf \color{orange} Arun: #1}}

\renewcommand{\yairside}[1]{\ignorespaces}
\renewcommand{\yair}[1]{\ignorespaces}
\renewcommand{\yujia}[1]{\ignorespaces}
\renewcommand{\sidford}[1]{\ignorespaces}
\renewcommand{\arun}[1]{\ignorespaces}

\newcommand{\calF}{\mathcal{F}} 
\arxiv{
\titlespacing*{\paragraph}{0pt}{6pt plus 2pt minus 2pt}{6pt}  %
}

\notarxiv{
\let\oldsubsection\subsection

\renewcommand{\subsection}[1]{%
  \oldsubsection{#1.}%
}
}

\title{\notarxiv{\Large}A Whole New Ball Game: A Primal Accelerated Method for Matrix Games and Minimizing the Maximum of Smooth Functions}

\arxiv{
\author{%
	Yair Carmon\thanks{Tel Aviv University, \texttt{ycarmon@tauex.tau.ac.il}.} 
	~~
	Arun Jambulapati\thanks{University of Washington,  \texttt{jmblpati@uw.edu}.}~~
	Yujia Jin\thanks{Stanford University, \texttt{\{yujiajin,sidford\}@stanford.edu}.} ~~
	Aaron Sidford\footnotemark[3]
}
}

\notarxiv{
    \author{%
	Yair Carmon\thanks{Tel Aviv University, \texttt{ycarmon@tauex.tau.ac.il}.} 
	\and
	Arun Jambulapati\thanks{University of Washington,  \texttt{jmblpati@uw.edu}.} \and
	Yujia Jin\thanks{Stanford University, \texttt{yujiajin@stanford.edu}.
    } \and
	Aaron Sidford\thanks{Stanford University, \texttt{sidford@stanford.edu}.}
}
}

\date{}

\notarxiv{%
\fancyfoot[R]{\scriptsize{Copyright \textcopyright\ 2024 by SIAM\\
Unauthorized reproduction of this article is prohibited}}
}

\begin{document}

\maketitle

\arxiv{\thispagestyle{empty}}

\begin{abstract}
\notarxiv{\small\baselineskip=9pt} %
We design algorithms for minimizing $\max_{i\in[n]} f_i(x)$ over a $d$-dimensional Euclidean or simplex domain. When each $f_i$ is $1$-Lipschitz and $1$-smooth, our method computes an $\epsilon$-approximate solution using $\widetilde{O}(n \epsilon^{-1/3} + \epsilon^{-2})$ gradient and function evaluations, and $\widetilde{O}(n \epsilon^{-4/3})$ additional runtime.
For large $n$, our evaluation complexity is optimal up to polylogarithmic factors. 
In the special case 
where each $f_i$ is linear---which corresponds to finding a near-optimal primal strategy in a matrix game---our method finds an $\epsilon$-approximate solution in runtime $\widetilde{O}(n (d/\epsilon)^{2/3} + nd + d\epsilon^{-2})$. For $n>d$ and $\epsilon=1/\sqrt{n}$ this improves over all existing first-order methods. When additionally $d = \omega(n^{8/11})$ our runtime also improves over 
all known interior point methods.

Our algorithm combines three novel primitives: (1) A dynamic data structure which enables efficient stochastic gradient estimation in small $\ell_2$ or $\ell_1$ balls. (2) A mirror descent algorithm tailored to our data structure
implementing an oracle which minimizes the objective over these balls. (3) A simple ball oracle acceleration framework suitable for non-Euclidean geometry.
\blfootnote{The ``ball'' in the title refers to ball oracle acceleration~\cite{carmon2020acceleration} at the heart of our results; no balls are placed into bins in this paper.}
\end{abstract}

\arxiv{
\newpage

\tableofcontents
\thispagestyle{empty}

\newpage
\pagenumbering{arabic} 
}

\section{Introduction}\label{sec:intro}

Consider the optimization problem
\begin{equation}
\label{eq:problem-bilinear}
\minimize_{x\in\xset} \crl*{
    \max_{i \in [n]} a_i^\top x = 
\max_{y \in \Delta^n} x^\top A y
},
\end{equation}
where $\xset \subset \R^d$ is closed and convex, $\Delta^n$ is the probability simplex in $n$ dimensions, and $A\in\R^{d\times n}$ has columns $a_1, \ldots, a_n$. 
We consider two settings of $\xset$: (1) the $\ell_1$ setting where $\xset \subseteq \Delta^n$ and we measure distance with the 1-norm, and (2) the $\ell_2$ setting where $\xset$ is a subset of the unit Euclidean ball and we measure distance with the Euclidean norm. The first setting encompasses finding an optimal strategy for one side of a matrix game, which is sufficient for linear programming~\cite{dantzig1953linear,adler2013equivalence}. 
The second setting includes important problems in machine learning and computational geometry: hard-margin support vector machines~\cite{minsky1988perceptrons} and minimum enclosing and maximum inscribed ball~\cite{clarkson2012sublinear}.

Due to its fundamental nature, many algorithms have been developed to solve \eqref{eq:problem-bilinear}. 
The frontier of the best performing algorithms comprises efficiently-implemented second-order interior point methods~\cite{cohen2021solving,brand2021minimum} and stochastic first-order methods~\cite{grigoriadis1995sublinear,clarkson2012sublinear,carmon2019variance}.
We are interested in methods of the second type, which currently obtain preferable runtimes as we fix the solution accuracy and let the problem dimensions $n$ and $d$ grow.
The best existing methods of this type jointly evolve the primal $x$ and dual $y$ variables via stochastic mirror descent; it not clear if additional runtime improvements are possible with this approach.

In this work we adopt a different approach, and design a primal stochastic first-order method that evolves the variable $x$ by directly sampling from an (approximate) best-response distribution $y$ at each step.\footnote{It is not clear whether our method can efficiently extract the solution to the dual problem $\maximize_{y\in\Delta^n} \min_{x\in\xset} x^\top A y$ without simply swapping the role of $y$ and $x$. Nevertheless, in many applications finding an approximately-optimal $x$ suffices.}
Our method solves the more general problem
\begin{equation}
    \label{eq:problem-gen}
    \minimize_{x\in\xset} \crl*{ \fmax(x) \defeq 
    \max_{i\in[n]} f_i(x)
    =
        \max_{y\in\Delta^n}  \sum_{i\in[n]} y_i f_i(x)
    },
\end{equation}
where $f_1, \ldots, f_n$ are convex, $\Lf$-Lipschitz, and $\Lg$-smooth with respect to the norm of interest.  The problem~\eqref{eq:problem-bilinear} corresponds to $f_i(x) = a_i^\top x$ and $\Lg=0$. 

\arxiv{
\newcommand{\runtimeColLen}{3.2cm}
\newcommand{\simplexColLen}{2cm}
}
\notarxiv{
\newcommand{\runtimeColLen}{3cm}
\newcommand{\simplexColLen}{1.25cm}
}
\begin{table}[h]
    \centering
    \begin{tabular}{@{}lp{4cm}p{\runtimeColLen}p{\simplexColLen}@{}}
    \toprule
    Method                                                              & $f_i, \grad f_i$ evaluation \phantom{blah} complexity                                                      & Additional runtime 
    & Simplex guarantees? \\ \midrule
    Subgradient method                                                  & $n \prn*{\frac{\Lf}{\epsilon}}^2$                                                  & -                                                           & \checkmark             \\
    AGD on softmax \cite{nesterov2005smooth}                            & $n \prn*{\frac{\Lf}{\epsilon}}$                                                    & -                                                           & \checkmark                         \\
    ``Thinking inside the ball'' \cite{carmon2021thinking}              & $n \prn*{\frac{\Lf}{\epsilon}}^{2/3} + \sqrt{n} \prn*{\frac{\Lf}{\epsilon}}$       & -                                                           & \xmark                  \\
    AGD on linearization \cite{nesterov2018lectures,carmon2021thinking} & $n \prn*{\frac{\Lg}{\epsilon}}^{1/2}$                                              & $\sqrt{nd(n+d)} \frac{\Lf\sqrt{\Lg}}{\epsilon^{3/2}}$       & \xmark                  \\
    \rowcolor[HTML]{EFEFEF} 
    Proposed method                                                     & $n \prn*{\frac{\Lg}{\epsilon}}^{1/3} + \prn*{\frac{\Lf}{\epsilon}}^2$              & 
    $n\frac{\Lf^2}{\Lg^{2/3} \epsilon^{4/3}}$
         & \checkmark              \\ \midrule
    Lower bound \cite{carmon2021thinking}                               & $n \prn*{\frac{\Lg}{\epsilon}}^{1/3} + \sqrt{n} \prn*{\frac{\Lg}{\epsilon}}^{1/2}$ & N/A                                                         & \xmark                  \\ \bottomrule
    \end{tabular}
    \caption{\label{table:runtimes-gen}
    Complexity guarantees for solving the problem~\eqref{eq:problem-gen} to $\epsilon$ accuracy. Parameters $n$ and $d$ denote the number of functions and domain dimensions, respectively, while $\Lf$ and $\Lg$ are the respective Lipschitz constants of $f_i$ and $\grad f_i$. Expressions in the table omit constant and polylogarithmic factors. We assume that each $f_i$ and $\grad f_i$ evaluation takes time $\Omega(d)$ so that the ``additional runtime'' column only includes terms that are not dominated by $d$ times the evaluation complexity. For simplicity, we also assume that $\epsilon \le \Lg \le \Lf^2 / \epsilon$. The final column indicates whether the method has proven guarantees for the $\ell_1$/simplex setting.
    }
\end{table}

\begin{table}[h!]
    \centering
    \begin{tabular}{@{}lll@{}}
        \toprule
        Method                                                                            & \begin{tabular}[c]{@{}l@{}}Runtime for \\ general parameters\end{tabular} & \begin{tabular}[c]{@{}l@{}}Runtime for \\ $n>d$ and $\epsilon=\frac{1}{\sqrt{n}}$\end{tabular} \\ \midrule
        Stochastic primal-dual \cite{grigoriadis1995sublinear,clarkson2012sublinear}      & $(n+d)\epsilon^{-2}$                                                      & $n^2$                                                                                          \\
        Exact gradient primal-dual \cite{nemirovski2004prox,nesterov2007dual}             & $nd\epsilon^{-1}$                                                         & $n^{3/2} d$                                                                                    \\
        Variance-reduced primal-dual \cite{carmon2019variance}                            & $nd + \sqrt{nd(n+d)}\epsilon^{-1}$                                        & $n^{3/2} d^{1/2}$                                                                              \\
        \rowcolor[HTML]{EFEFEF} 
        Proposed method                                                                   & $nd + n(d/\epsilon)^{2/3} + d\epsilon^{-2}$                               & $n^{4/3} d^{2/3}$                                                                              \\
                                                                                          & $\max\crl{n,d}^\omega$                                                        & $n^\omega$                                                                                     \\
        \multirow{-2}{*}{Interior point \cite[resp.,][]{cohen2021solving,brand2021minimum} $^\dagger$} & $nd + \min\crl{n,d}^{5/2}$                                                    & $nd + d^{5/2}$                                                                                 \\ \bottomrule
        \end{tabular}
    \caption{\label{table:runtimes-bilinear}
    Runtime bounds for solving the problem~\eqref{eq:problem-bilinear} to $\epsilon$ accuracy, omitting constant and polylogarithmic factors. The bounds assume a unit Lipschitz constant, i.e., $\norm{a_i}_* \le 1$ for all $i$, where the dual norm $\norm{\cdot}_*$ is the $\infty$-norm in the $\ell_1$ setting and the 2-norm ins the $\ell_2$ setting. $^\dagger$To our knowledge the runtime bound $nd + \min\crl{n,d}^{5/2}$ is proven only in the $\ell_1$ setting.
    }
\end{table}

Our methods builds upon previous work~\cite{carmon2021thinking,asi2021stochastic,carmon2022distributionally} that develop \emph{ball oracles} which approximately minimize $\fmax$ in a small ball around a reference point, and then apply \emph{ball oracle acceleration}~\cite{carmon2020acceleration,carmon2022optimal} to globally minimize the objective in a small number of ball oracle calls. 
These methods have two key shortcomings that prevent them from providing better runtimes for matrix games: (1) the ball oracles they implement have too small ball radii and (2) they do not apply to $\ell_1$ geometry.
This work overcomes the first shortcoming by designing data structures that, using sketching and sampling techniques, maintain linear approximations of the functions $\{f_i\}$ which facilitate efficient gradient estimation at larger distance from the reference point.
To overcome the second challenge we redesign the ball oracle acceleration framework using a novel accelerated proximal point method formulation, and implement an approximate non-Euclidean ball oracle using a careful mirror descent scheme that provides a fine-grained control of the amount of iterate movement which our data structures require.

\Cref{table:runtimes-gen,table:runtimes-bilinear} summarize the complexity guarantees of our method and compare them to prior work.
We measure complexity as either the runtime or the number of evaluations of $f_i(x)$ and $\grad f_i(x)$ (for some $x\in\xset$ and $i\in[n]$) required to produce $x$ such that $\E \fmax(x) - \min_{\xopt\in\xset}\fmax(\xopt) \le \epsilon.$
For the problem \eqref{eq:problem-gen} in the regime where $n$ is large (e.g., $n \ge \Lf^2\epsilon^{-2}$) we obtain the optimal evaluation complexity with a modest additional computational cost due to our data structures, which becomes negligible as $d$ grows.
For problem \eqref{eq:problem-bilinear}, in the regime $d < \min\{\epsilon^{-2}, n\epsilon^{2}\}$ (which implies $d<n$) our bounds improve on all previous first-order methods. This regime includes $\epsilon=\frac{1}{\sqrt{n}}$, which is standard for empirical risk minimization problems, where statistical errors are typically also of order $\frac{1}{\sqrt{n}}$. Since we consider maximum (rather that mean) risk minimization, statistical errors (if they exist) will likely be higher for the problems we study. Additionally, our runtime improves over all known methods (including interior point methods) when $n$ is not too much larger than $d$ and $\epsilon$ lies in some range, namely $d < n < d^{3/2}$, $\epsilon \leq \frac{1}{\sqrt{d}}$, and $\epsilon>\max\{\frac{1}{d^{3/4}}, \frac{n^{3/2}}{d^{11/4}}, \frac{d^{1/2}}{n}\}$. For $\epsilon=\frac{1}{\sqrt{n}}$ we improve over all methods when $n^{8/11} < d < n$.

Our results also directly lead to an algorithm for finding the minimum Euclidean ball enclosing points $a_1,\ldots,a_n\in\R^d$, a fundamental problem in computational geometry~\cite{sylvester1857question,clarkson2012sublinear}. Our algorithm finds an $\epsilon$-accurate solution in $\Otil{nd + d/\eps + n d^{2/3} \eps^{-1/3}}$ time, while the previous best known runtime obtained by first-order methods is $\Otil{nd+nd^{1/2}\eps^{-1/2}}$~\cite{allen2016optimization,carmon2020coordinate}. This is an improvement for a range of parameter values including $n > 1/\epsilon > d$.

\paragraph{Paper organization (see also \Cref{fig:diagram}).} In~\Cref{ssec:related-work} we discuss related work. In~\Cref{sec:overview} we provide a detailed overview of our key technical contributions. \Cref{sec:prelim} introduces the general notation and conventions of the paper. In~\Cref{sec:outerloop}, we describe the main acceleration framework building on a ball-restricted proximal oracle, followed by the implementation of restricted  oracle in~\Cref{sec:innerloop}. In~\Cref{sec:data_structure}, we build the main data structure used for $\ell_p$-matrix-vector maintenance, which we then use to build an efficient stochastic gradient estimator in~\Cref{sec:grad-est}. In~\Cref{sec:runtimes}, we combine our developments and obtain guarantees for solving problem \eqref{eq:problem-gen} and, as special cases, problem~\eqref{eq:problem-bilinear} and minimum enclosing ball.

\begin{figure}
    \centering
    \includegraphics[width=0.95\textwidth]{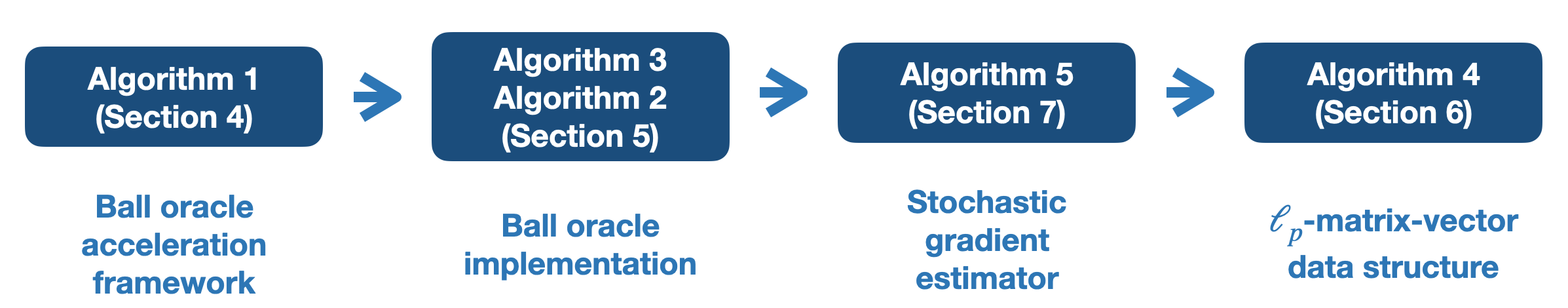}
    \caption{\label{fig:diagram} A diagram of the main components of our algorithm and their location in the paper.}
    \end{figure}

\subsection{Related work}\label{ssec:related-work}
We now review several additional closely related lines of research.

\paragraph{Minimizing the maximum of linear functions.}
Research on algorithms for solving problems of the form~\eqref{eq:problem-bilinear}, particularly in the context of linear programming, has a long and celebrated history in computer science~\cite{dantzig1953linear}. The best existing methods fall on a spectrum of trade-offs between per-iteration cost and number of iterations. 
At one end of the spectrum lie second-order, interior-point methods~\cite{karmarkar84new,renegar1988polynomial}, whose iterations are expensive (usually requiring a linear system solution) but the number of iterations depends only logarithmically on the desired accuracy $\epsilon^{-1}$; recent years saw much progress at making the iterations of these methods more efficient~\cite[e.g.,][]{lee2015efficient,cohen2021solving,Brand20,brand2020solving,brand2021minimum,Brand21,Jiang0WZ21}. 
Next come first-order methods that use exact gradients~\cite[e.g.,][]{nesterov2005smooth,nemirovski2004prox,nesterov2007dual} whose per-iteration cost is linear in the problem size, but whose iteration complexity typically scales as $\epsilon^{-1}$. 
Finally, at the other end of the spectrum are stochastic first-order methods~\cite[e.g.,][]{grigoriadis1995sublinear,clarkson2012sublinear} whose per iteration cost is sublinear in the problem size---and sometimes even near-constant~\cite{carmon2020coordinate,wang2020randomized}---but whose iteration complexity typically scales as $\epsilon^{-2}$.
In addition, variance reduction techniques~\cite[e.g.,][]{balamurugan2016stochastic,carmon2019variance,carmon2020coordinate,song2021variance,song2022coordinate} use a mix of exact and stochastic gradient computation to obtain a faster rate of convergence in terms of $\epsilon$ while maintaining a sublinear per-iteration cost. 

It is possible to view our ball oracle approach as a hybrid of stochastic and exact gradient queries, though the way we leverage the exact gradient queries is quite different from variance reduction: we query exact gradients to increase the efficiency of nearby stochastic gradient estimates, while variance reduction methods seek to make them more accurate. \citet{carmon2021thinking} (discussed at length in the following section) combine a ball oracle and variance reduction for minimizing the maximum of Lipschitz, slightly smooth functions. However, to do so they rely on an ``exponentiated softmax'' technique that is not compatible with the larger balls we consider in this paper. Enhancing our method using variance reduction is a promising direction for future work.

\paragraph{Minimizing the maximum of general convex functions.}
The general problem~\eqref{eq:problem-gen} has seen less research than the matrix games problem~\eqref{eq:problem-bilinear}. The exact-gradient first order methods mentioned above~\cite[e.g.,][]{nesterov2005smooth,nemirovski2004prox,nesterov2007dual} also apply in the general case,
and \citet[Section 2.3.1]{nesterov2018lectures} shows how to reduce the general cases to a sequence of matrix games. 
However, stochastic gradient methods typically exploit the matrix structure in~\eqref{eq:problem-bilinear} and do not extend to the general case. 
Indeed, stochastic methods for the problem~\eqref{eq:problem-gen} typically have high variance gradient estimators, leading to an iteration count that depends on the number of functions $n$~\cite{namkoong2016stochastic,shalev2016minimizing,nemirovski2009robust,carmon2022distributionally}. The work~\cite{carmon2021thinking} made significant progress in reducing the number of full-data passes required to solve the problem~\eqref{eq:problem-gen}, and we improve it further to obtain (for large $n$) the optimal number of data passes for smooth problems.

\paragraph{Accelerated approximate proximal point methods.}
The accelerated proximal point method \cite{guler1992new,salzo2012inexact} is a powerful and versatile building block for convex optimization algorithms, owing to the fact that the proximal point operation admits several approximate solution criteria that preserve the accelerated rate of convergence~\cite{frostig2015unregularizing,lin2015universal,carmon2022recapp}.
In particular, the approximate solution notion due to \citet{monteiro2013accelerated} has led to a plethora of accelerated optimization methods~\cite[e.g.,][]{gasnikov2019optimal,bubeck2019optimal,jiang2019optimal,bullins2020highly,bubeck2019complexity,kovalev2022first} including the ball oracle acceleration framework~\cite{carmon2020acceleration,carmon2021thinking,asi2021stochastic,carmon2022distributionally} at the core of our algorithm. We contribute to this line of research by designing a new approximate accelerate proximal point method that is suitable for non-Euclidean geometry and allows efficient oracle implementation using stochastic gradient methods; our technique also borrows the momentum damping technique from~\cite{carmon2022optimal} for improving the simplicity and efficiency of the Monteiro-Svaiter method.

Our acceleration scheme also bears a strong resemblance to gradient sliding~\cite{lan2016gradient,thekumparampil2020projection,lan22accelerated}: both techniques efficiently approximate the accelerated proximal point method by making use of both the averaged and final iterates of stochastic gradient descent. Since our method is based on a simple approximation condition for an exact proximal point problem, it provides insight into the efficacy of this approach.

\paragraph{Data structures for optimization.}
Optimization algorithms often rely on data structures for leveraging iterate sparsity and efficiently computing projections~\cite{lee2013efficient,sidford2018coordinate,duchi2008efficient,carmon2020coordinate}.
However, randomized data structures---such as the matrix-vector maintainer we employ---are notoriously difficult to use in the context of optimization, since the iterative nature of the algorithm could make the sequence of data structure queries non-oblivious, thus invalidating the data structure's guarantees.
We address this difficulty using rejection sampling, which ensures that the distribution of consecutive queries is the same regardless of the data structure's random state. 

Our matrix-vector maintenance data structure is closely related to data structures designed in recent works on efficient interior point methods for linear programming \cite{brand2020solving,brand2020bipartite,brand2021minimum}, e.g., the ``vector maintenance data structure'' in \cite{brand2020solving}. The interior point methods using these data structures also take care to ensure that their queries remain oblivious, though not always via rejection sampling. Similar to our data structure, the ones in  \cite{brand2020solving,brand2020bipartite,brand2021minimum} also maintain an approximation to the products of a sequence of query vector with a given matrix, and they use a linear sketch similar to the one we use for the Euclidean case (but not the $\ell_1$ case). 
Our data structure differs in the type of approximation maintained, the norms considered, and the assumptions on the query sequence. 
Moreover, our technique of supporting a long query sequence by instantiating multiple simpler data structures at different scales is well known~\cite[see, e.g.,][]{axelrod2020near}.
\sidford{post deadline see if more references should be added and if this has a name, e.g., diadic / temporal bucketing?}

\section{Technical overview}\label{sec:overview}
In this section we provide a detailed overview of our technical contribution. \Cref{subsec:overview-prelims} describes the initial setup proposed in~\cite{carmon2021thinking,asi2021stochastic}. In \Cref{subsec:overview-ds} we explain how we use linear approximations and data structures to increase the size of the ball for which we can implement an optimization oracle. Then, in \Cref{subsec:overview-accel} we explain how to extend ball oracle acceleration to non-Euclidean geometry, in \Cref{subsec:overview-oracle-impl} we describe the ball oracle implementation, and in \Cref{subsec:overview-putting-together} we put the components of our algorithm together and derive its complexity bounds.

\subsection{Preliminaries}\label{subsec:overview-prelims}
To begin the technical exposition, we first explain the key components of the ``thinking inside the ball'' approach~\cite{carmon2021thinking,asi2021stochastic} to solving the problem~\eqref{eq:problem-gen}, which we build upon to obtain our results.
The first step at tackling the problem is the standard ``softmax'' trick of smoothing the maximum operation by considering
\begin{equation}
    \fsm(x) \defeq \max_{y\in\Delta^n}\crl*{\sum_{i\in[n]} \brk*{ y_i f_i(x) - \epsilon' y_i \log y_i}} = \epsilon' \log\prn*{\sum_{i\in[n]} e^{f_i(x)/\epsilon'}},~\mbox{with}~\epsilon'=\frac{\epsilon}{2\log n},
\end{equation}
which is a uniform $\halfsilon$-approximation to $\fmax$ and therefore minimizing it to accuracy $\halfsilon$ solves the problem~\eqref{eq:problem-gen} to accuracy $\epsilon$. 

Next, we design an oracle that approximately minimizes $\fsm$ in a ball of radius $r$ around a query point $y\in\xset$. 
Roughly speaking, the implementation consists of stochastic gradient descent (SGD) with an unbiased estimator for $\grad \fsm(x)=\E_{i \sim e^{f_i(x)/\epsilon'}} \grad f_i(x)$.
Naively computing the distribution proportional to $e^{f_i(x)/\epsilon'}$ requires $n$ function/gradient evaluations, which is as expensive as computing $\grad \fsm$ exactly.
Instead, the estimator proposed in~\cite{asi2021stochastic} uses rejection sampling to efficiently draw $i \sim e^{f_i(x)/\epsilon'}$, and then returns $\grad f_i(x)$. 
Given a query point $x$ and a reference point $y$, the rejection sampling operates by drawing $i \sim e^{\tilde{f}_i(x;y)/\epsilon'}$, where $\tilde{f}_i(x;y)$ is an approximation of $f_i(x)$ for $x$ close to $y$, and then accepting with probability $\exp\prn[\big]{(f_i(x) - \tilde{f}_i(x;y)-C)/\epsilon'}$ for $C$ such that $\abs[\big]{f_i(x) - \tilde{f}_i(x;y)} \le C$ for all $\norm{x-y}\le r$. 
For an approximation $\tilde{f}_i$ with $C=O(\epsilon')$, this rejection sampling routine returns a valid sample from $e^{f(x)/\epsilon'}$ using an expected $O(1)$ draws from $e^{\tilde{f}_i(x;y)/\epsilon'}$. 
\citet{asi2021stochastic} simply perform $n$ evaluations to precompute $f_1(y),\ldots,f_n(y)$ and then take $\tilde{f}_i(x;y) = f_i(y)$, for which $C=\Lf r$ by the Lipschitz continuity of the $f_i$.  
Taking $r = \epsilon' / \Lf$ ensures that each $\grad \fsm$ estimation takes $O(1)$ expected additional evaluations.
Thus, the overall expected evaluation complexity of minimizing $\fsm$ inside a ball of radius $r=O(\epsilon/\Lf)$ is $n+O(T)$, where the SGD iteration number $T$ is sublinear in $n$.

Finally, we make efficient use of the ball oracle to globally minimize $\fsm$. To this end, we rely on the ball oracle acceleration technique proposed by~\citet{carmon2020acceleration} and refined in~\cite{carmon2021thinking,asi2021stochastic,carmon2022distributionally,carmon2022optimal}, which we further improve in this work. The technique, a type of accelerated proximal point method~\cite{guler1992new,salzo2012inexact,monteiro2012iteration} finds an $\epsilon$-accurate minimizer in $O(r^{-2/3}\log(1/\epsilon))$ ball oracle calls. 
Combining these ingredients yields a gradient evaluation complexity bound whose leading term in $n$ is $\Otil{nr^{-2/3}} = \Otil{n(\Lf/\epsilon)^{2/3}}$.

\subsection{Increasing the ball size by linear approximation data structures}\label{subsec:overview-ds}

\paragraph{Exact linear approximation.}
The main limitation of the softmax gradient estimation procedure described above is that it only works for fairly small balls of radius $\Otil{\epsilon/\Lf}$. 
To increase the ball size, we leverage smoothness to build better function value approximations $\tilde{f}_i(x;y)$. 
As a starting point, consider the linear approximation 
\[\tilde{f}^{\mathrm{lin}}_i(x;y) \defeq f_i(y) + \inner{\grad f_i(y)}{x-y}.\]
When each $f_i$ is $\Lg$-smooth (i.e., $\grad f_i$ is $\Lg$-Lipschitz) then $\abs[\big]{f(x) - \tilde{f}^{\mathrm{lin}}_i(x;y)} \le \half \Lg \norm{x-y}^2$ for all $x$ and $y$. Therefore, we may increase the ball radius $r$ from $\epsilon'/\Lf$ to $\sqrt{\epsilon' / \Lg}$. Since computing $\tilde{f}^{\mathrm{lin}}_1(\cdot;y),\ldots,\tilde{f}^{\mathrm{lin}}_n(\cdot;y)$ requires only $n$ function and gradient evaluations, substituting this improved approximation into the acceleration framework described above yields a leading order evaluation complexity term of $\Otil{nr^{-2/3}} = \Otil{n(\Lg/\epsilon)^{1/3}}$.

However, sampling $i\sim e^{\tilde{f}^{\mathrm{lin}}_i(x;y)/\epsilon'}$ is computationally expensive, since exactly computing the inner products $\inner{\grad f_1(y)}{x-y}, \ldots, \inner{\grad f_n(y)}{x-y}$ takes $\Theta(nd)$ time. In some cases, including bilinear problems~\eqref{eq:problem-bilinear}, this is as expensive as calculating $\grad \fsm$ exactly, undoing the efficiency gains of rejection sampling using $\tilde{f}^{\mathrm{lin}}$.

\paragraph{Matrix-vector estimation data structure.}
We address this challenge by replacing $\tilde{f}^{\mathrm{lin}}_i(x;y)$ with an efficient randomized approximation, denoted $\tilde{f}^{\mathrm{est}}_i(x;y)$, such that $\abs[\big]{\tilde{f}^{\mathrm{est}}_i(x;y)-\tilde{f}^{\mathrm{lin}}_i(x;y)} \le \epsilon'$ with high probability. We construct \emph{matrix-vector estimation data structures} that, after $O(nd)$ preprocessing time, for query $x$ and reference $y$, compute $\{\tilde{f}^{\mathrm{est}}_i(x;y)\}_{i\in[n]}$ in time $\Otilb{n \prn*{{\Lf \norm{x-y}}/{\epsilon'}}^2}$: in the $\ell_2$ setting we achieve this using CountSketch~\cite{charikar2002finding,larsen2021count}, while in the $\ell_1$ setting we simply approximate $\inner{\grad f_i(y)}{x-y}$ by sampling entries of $\grad f_i(y)$ from a distribution proportional to $|x-y|$, a technique similar to ``sampling from the difference'' used for variance reduction in matrix games~\cite{carmon2019variance}. 

\paragraph{From matrix-vector estimation to maintenance.}
If we were to implement the ball oracle using the estimate described above, the additional runtime cost would be $\Otil{n \prn*{{\Lf r}/{\epsilon}}^2 T}$, where $T$ is the SGD iteration count. 
While independent of $d$, such runtime would have a large dependence on the desired accuracy $\epsilon$, again rendering the approach unhelpful for matrix games. To further improve efficiency, we design \emph{matrix-vector maintenance data structures} that allow evaluating $\tilde{f}^{\mathrm{est}}$ at a series of query points $x_1, \ldots, x_T$ with additional runtime \[\Otilb{nd + n \prn*{\frac{\Lf \sum_{i\in[T]}\norm{x_i-x_{i-1}}}{\epsilon'}}^2}.\] As we explain in more detail below, we design a careful stochastic gradient method for which the queries satisfy $\sum_{i\in[T]}\norm{x_i-x_{i-1}} = \Otil{r}$, leading to the $\Otilb{r^{-2/3} n \prn*{{\Lf r}/{\epsilon}}^2} = \Otilb{n {\Lf^2}/{\Lg^{2/3} \epsilon^{4/3}}}$ additional runtime shown in \Cref{table:runtimes-gen}.

Our matrix-vector maintenance data structure solves the more general problem of approximately maintaining the value of $A x$ for a suitably bounded  matrix $A \in \R^{n \times d}$ and a changing $x$ that is guaranteed not to move too much. Specialized to our applications, this data structure essentially maintains approximations to $\inner{\grad f_i(y)}{v}$ for vectors $v$ of different (exponentially spaced) distances from the current point. Given a new query $x$, the data structure simply updates the $v$ vectors and the approximations to $\inner{\grad f_i(y)}{v}$ to preserve the exponentially space distance invariant. This update is made efficient by estimating the values of $\inner{\grad f_i(y)}{v}$ for the update $v$ in terms of their difference in value from the closest (to the query point) non-updated $v$ and using a matrix-vector estimator. Finally, the data structure outputs the approximation to $\inner{\grad f_i(y)}{v}$ for the closest $v$. 

By carefully choosing and reusing approximations to the $\inner{\grad f_i(y)}{v_j}$ over time, we are able to guarantee the claimed runtime bound. Essentially, we obtain a runtime for maintaining approximation over a whole sequence of queries in essentially the same complexity a matrix-vector estimation data structure would naturally use for answering one query whose distance form $y$ is the sum of all query movements. We design our matrix-vector maintenance data structure via a reduction to matrix-vector estimation, which; this more general framework could be of utility in other geometries.

\paragraph{A note on obliviousness.}
Our use of efficient randomized data structures hinges on a subtle yet crucial property of our method: our data structure query sequences do not depend on its random state, and hence the probabilistic approximation guarantees remain valid throughout. At first glance this might appear to be false, since we use the output of the data structure to draw random indices that define the stochastic gradient estimate and hence influence the next SGD iterate and data structure query point. However, due to rejection sampling, the \emph{distribution} of the rejection sampling output is proportional to $e^{f(x)/\epsilon'}$, without any dependence on the random bits of the data structure.\footnote{More precisely, the distribution of the next iterate is the same for all possible random bits, except for a low-probability set of random bits for which the approximation condition $\abs[\big]{\tilde{f}^{\mathrm{est}}_i(x;y)-\tilde{f}^{\mathrm{lin}}_i(x;y)} \le \epsilon'$ fails for some $i\in[n]$.}

\subsection{Accelerating entropy ball oracles}\label{subsec:overview-accel}

We now shift our focus to the ball oracle acceleration algorithm that takes in an (approximate) radius-$r$ ball oracle and returns an approximate minimizer in $\Otil{r^{-2/3}}$ oracle calls. Here, the main challenge is extending the algorithm to support a non-Euclidean domain geometry. Specifically, the difficulty lies in coming up with an approximate oracle notion that supports efficient implementation via stochastic gradient methods while still allowing acceleration. 

\paragraph{Prior idealized scheme.}
To explain our developments, it is instructive to first consider idealized acceleration schemes using exact ball oracles, and contrast the idealized scheme of prior work to the one proposed here. Previous ball acceleration methods~\cite{carmon2020acceleration,carmon2021thinking,asi2021stochastic,carmon2022optimal}\footnote{In order to ensure correctness, these ball acceleration methods must either choose $a_t$ such that $x_{t+1}$ has $\norm{x_{t+1} -  \Phi_t(v_t)} \in [r/2, r)$ (which necessitates a bisection to solve an implicit equation) or modify their iterates through a momentum damping scheme~\cite{carmon2022optimal}. We ignore this point throughout the overview, and use momentum damping in our full method.} maintain a parameter sequence $a_1, \ldots, a_T$ and its running sums $A_t = \sum_{i \le t} a_i$, and construct an iterate sequence $(x_t,v_t)$ according to
\begin{flalign}
x_{t+1} &= \argmin_{x\in\xset: \norm{x-\Phi_t(v_t)}_2\le r} \crl*{f(x) + \frac{A_{t+1}}{2a_{t+1}^2} \norm{x-\Phi_t(v_t)}_2^2}~~\mbox{where}~\Phi_t(z) \defeq \frac{A_t}{A_{t+1}} x_t  + \frac{a_{t+1}}{A_{t+1}} z
\label{eq:exact-apm-prox-step}
\\
v_{t+1} &= \argmin_{v\in\xset}\crl*{\inner{\grad f(x_{t+1})}{v} + \frac{1}{2a_{t+1}}\norm{v-v_{t}}^2_2}.
\label{eq:exact-apm-mirror-step}
\end{flalign}
The step~\eqref{eq:exact-apm-prox-step} calls a radius-$r$ ball oracle with center point $\Phi_t(v_t)$, while the step~\eqref{eq:exact-apm-mirror-step} executes a mirror descent iteration using the gradient of $f$ at the output of the ball oracle. Proper setting of $a_t$ ensures that for all $t$ we have $f(x_t) - f(\xopt) \le \frac{\norm{x_0-\xopt}_2^2 - S_t}{A_t}$ for some $S_t \geq 0$, and that after $T = O\prn*{\prn*{\frac{\norm{x-\xopt}_2}{r}}^{2/3}\log \frac{f(x_0)-f(\xopt)}{\epsilon}}$ iterations either $A_T \ge \frac{\norm{x_0-\xopt}_2^2}{\epsilon}$ or $S_T \geq \norm{x_0 - \xopt}_2^2$. 

To move to general norms, we use the standard technique of introducing a Bregman divergence $V_a(b)$ induced by a 1-strongly-convex distance generating function, so that $V_a(b) \ge \half\norm{b-a}^2$; in the Euclidean we simply have $V_a(b) = \half\norm{a-b}_2^2$, while for the simplex setting we use the KL divergence $V_a(b) = \sum_{i\in[d]} b_i \log \frac{b_i}{a_i}$ (see \Cref{sec:prelim} for more details). 

A straightforward generalization of the idealized method above exists, but is not conducive to approximation. Such generalization consists of replacing $\norm{\cdot}_2$ in step~\eqref{eq:exact-apm-prox-step} with a general norm $\norm{\cdot}$, and replacing $\half\norm{v-v_k}^2_2$ with $V_{v_k}(v)$ in step~\eqref{eq:exact-apm-mirror-step}. It can be shown that $f(x_t) - f(\xopt) \le \epsilon$ after $O\prn*{\prn*{\frac{V_{x_0}(\xopt)}{r}}^{2/3}\log \frac{f(x_0)-f(\xopt)}{\epsilon}}$ iterations. %
However, it is not clear how to efficiently approximate the non-Euclidean ball oracle computation in this method. In particular, in order to approximate the step~\eqref{eq:exact-apm-mirror-step}, \citet{asi2021stochastic} design a multilevel Monte Carlo (MLMC) estimator that is nearly unbiased for the exact ball oracle output~\eqref{eq:exact-apm-prox-step}, and the analysis of this technique appears to strongly rely on properties that are unique to the Euclidean norm. 

\paragraph{New idealized scheme.} To address this challenge, we redesign the acceleration method with Bregman divergences and efficient approximation in mind. Our new idealized method is 
\begin{flalign}
    v_{t+1} &= \argmin_{v\in\xset: V_{v_t}(v)\le \half \rho_{t+1}^2} \crl*{A_{t+1} f(\Phi_t(v)) + V_{v_t}(v)}~~\mbox{where}~\rho_{t+1} \defeq
        \frac{A_{t+1}}{a_{t+1}} r
    \label{eq:new-exact-apm-prox-step}
    \\
    x_{t+1} &= \Phi_t(v_{t+1})
    \label{eq:new-exact-apm-substitution}
\end{flalign}
and $\Phi_t$ is as defined in~\eqref{eq:exact-apm-prox-step}. In the unconstrained Euclidean case (i.e., when $\xset=\R^d$ and the $a_t$ sequence is such that $\norm{x_{t+1}-\Phi_t(v_t)}<r$ for all $t$), straightforward algebra shows that the old and new idealized schemes  are exactly equivalent. However, outside that setting---and particularly in the non-Euclidean case---the two methods produce different iterates. Nonetheless, both methods enjoy the same $O\prn[\big]{\prn[\big]{\frac{V_{x_0}(\xopt)}{r}}^{2/3}\log \frac{f(x_0)-f(\xopt)}{\epsilon}}$ iteration complexity guarantee. Moreover, the constraint $V_{v_t}(v) \le \half \rho_{t+1}^2$ and the definition of $\rho_{t+1}$ implies that every feasible point $v$ in step~\eqref{eq:new-exact-apm-prox-step} satisfies $\norm{\Phi_t(v) - \Phi_t(v_t)}\le \frac{a_{t+1}}{A_{t+1}} \norm{v-v_t} \le \frac{a_{t+1}}{A_{t+1}}\rho_{t+1} = r$. This justifies considering~\eqref{eq:new-exact-apm-prox-step} a call to a radius-$r$ optimization oracle centered at $\Phi_t(v_t)$.

\newcommand{\vopt}{v_\star}
\paragraph*{Defining the approximate ball oracle.}
We now briefly derive our approximation condition for step~\eqref{eq:new-exact-apm-prox-step}. To lighten notation, let $y\defeq v_t$, let $\rho\defeq \rho_{t+1}$, let $h(v) \defeq A_{t+1}f(\Phi_t(v))$ and $\vopt \defeq v_{t+1}$ (i.e., the exact ball oracle output). Note that $\vopt$ is the global minimizer of $H(v) \defeq h_t(v) + c V_{y}(v)$, for some $c\ge 1$ which enforces the constraint $V_{y}(v) \le \half \rho^2$. Therefore, by convexity we have $H(\vopt) - H(u) \le -c V_{\vopt}(u)$ for all $u\in\xset$. Substituting the definition of $H$ and dividing through by $c$ gives
\begin{equation}\label{eq:exact-oracle-condition}
    \frac{h(\vopt)-h(u)}{c} \le V_y(u) - V_{\vopt}(u) - V_{y}(\vopt)
    \le V_y(u) - V_{\vopt}(u) - \half \rho^2 \indic{c>1}
    ~\mbox{for all $u\in\xset$},
\end{equation}
where the final inequality holds since $V_y(\vopt)=\half\rho^2$ when $c>1$ due to complementary slackness.

To further relax the condition~\eqref{eq:exact-oracle-condition}, we allow the approximate ball oracle to \emph{return two points} $z,w\in\xset$ such that $h(z)$ replaces $h(\vopt)$ and $V_{w}(u)$ replaces $V_{\vopt}(u)$. We further replace $\half \rho^2 \indic{c>1}$ with $\gamma \rho^2 \indic{c>2}$ for some $\gamma \le \half$, and we allow $\gamma\rho^2$ additive error for $c\le 2$. Finally, we allow randomization by requiring that bound holds only in expectation. The resulting relaxed output condition is
\begin{equation}\label{eq:inexact-oracle-condition}
   \E \frac{h(z)-h(u)}{c} 
    \le \E\brk*{V_y(u) - V_{w}(u)} - \gamma \rho^2 \E\brk*{\indic{c>2} - \indic{c\le2}}
    ~\mbox{for all $u\in\xset$}.
\end{equation}
In the acceleration framework, we approximate $v_{t+1}$ with $w$, and $x_{t+1}$ with $\Phi_t(z)$, and show that the resulting sequence still satisfies (up to constants) the same error bound as the exact proximal method. The key advantage of the two-point approximation condition~\eqref{eq:inexact-oracle-condition} is that SGD naturally achieves it, with $z$ and $w$ being the average and final SGD iterates respectively. This ``two outputs'' property of SGD has been leveraged before in the literature on gradient sliding methods in structured convex optimization \cite{lan2016gradient,thekumparampil2020projection, lan22accelerated}. It allows us to sidestep the need for Multilevel Monte-Carlo~\cite{blanchet2015unbiased,asi2021stochastic}, which appears challenging to use in the non-Euclidean setting.

\subsection{Implementing entropy-ball oracles}\label{subsec:overview-oracle-impl}
We now explain the key components in constructing an approximate ball oracle meeting the condition~\eqref{eq:inexact-oracle-condition} using our data structure-based gradient estimator. There are two main challenges in designing this oracle. First, the inequality~\eqref{eq:inexact-oracle-condition} needs to hold for all $u\in\xset$ rather than just in a ball of radius $\rho$ around $y$; this prevents us from using standard constrained optimization techniques. Second, our matrix-vector maintenance data structure requires that the total movement in the SGD iterates sum to $\Otil{\rho}$, a guarantee which standard SGD does not provide. We explain our solution to each challenge in turn.

\paragraph{Implicitly-constrained SGD.}
To obtain a guarantee valid for any comparator point $u\in\xset$, we approximately find the Lagrange multiplier for the constraint $V_y(v)\le \half \rho^2$ and apply unconstrained SGD, taking careful care to show that its iterates nevertheless stay close to the reference point $y$. First, we perform bisection to find a Lagrange multiplier $\lambda\ge 1$ such that $v_\lambda = \argmin_{v\in\xset} \crl*{h(v) + \lambda V_y(v)}$ satisfies $V_y(v_\lambda)\in [\frac{\alpha}{2} \rho^2, \frac{\beta}{2}\rho^2]$ for some $\alpha,\beta=\widetilde{\Theta}(1)$, where we use SGD to approximate $V_y(v_\lambda)$. 
Second, having found a suitable $\lambda$, we apply (unconstrained) SGD once more to obtain the global guarantee~\eqref{eq:inexact-oracle-condition} with $c\approx \lambda$. 
However, removing the explicit ball constraint introduces another difficulty: SGD could potentially query iterates outside the ball, where our gradient estimator is inefficient.
To address this concern we use techniques introduced in~\cite{carmon2022making,ivgi2023dog} to show that, with high probability, SGD never leaves a ball of radius $O(\norm{v_\lambda - y})$ around $y$.
Since $\lambda$ satisfies  $V_y(v_\lambda)=O(\rho^2)$, the SGD iterates remain (with high probability) in the region where our gradient estimator is efficient. 

\paragraph{A relaxed triangle inequality of KL divergence.}
Before proceeding to the next challenge we highlight a technical point of potential broader interest. 
To establish the correctness of the procedures described above, we need to assume that the Bregman divergence satisfies a relaxed triangle inequality of the form 
\[
    V_a(b) + V_b(a) \le \tau\prn*{\tilde{V}(a,c) + \tilde{V}(c,b)}
    ~\mbox{where}~\tilde{V}(x,y) = \min\{V_x(y),V_y(x)\}
\]
for all $a,b,c\in\xset$.
In the Euclidean case where $V_a(b) = \half\norm{a-b}_2^2$, this holds for  $\tau = 4$. However, when $\xset$ is the simplex and $V$ is the KL divergence, this inequality is false for any $\tau$. Nevertheless we show that for a \emph{truncated simplex} $\Delta_{\nu}^n = \{ p \in \Delta^n \mid p_i \ge \nu~\mbox{for all }i\}$, the relaxed triangle inequality holds with $\tau=O(\log \frac{1}{\nu})$. This observation is new to the best of our knowledge, and potentially of independent interest. 
The Lipschitz continuity of our objective functions means that its optimal value in $\xset$ and $\xset \cap \Delta_{\nu}^n$ differ by at most $O(\Lf \nu)$. 
Therefore, truncating the simplex with $\nu = \poly(\epsilon/\Lf)$ allows us to use the relaxed triangle inequality with $\tau=\Otil{1}$ without significantly changing the solution quality.

\paragraph{Controlling the sum of query movement sizes.}
Next, we address the challenge introduced by our matrix-vector maintenance data structure. This data structure enables us to generate stochastic gradients for SGD at a computational cost proportional to the sum of distances between consecutive SGD queries. For standard SGD using $T$ iterations, this sum is $\Omega(\sqrt{T})$, resulting in a bad complexity bound. To address this, we employ a variant of SGD  due to \citet{cutkosky2019anytime} which enables much tighter control over total query movement. This variant applies mirror descent updates on the gradient estimated on the running average of its iterates, computing
\[
w_{t+1} = \argmin_{w\in\xset}\crl*{\inner{\gest(x_t)}{w} + \frac{1}{\eta} V_{w_t}(w)},
\]
where $\eta$ is a step size, $\gest$ is the gradient estimator, and $x_t = \frac{1}{t}\sum_{i\le t} w_i = \frac{t-1}{t}x_{t-1} + \frac{1}{t}w_t$. Therefore, we have $\sum_{t\le T} \norm{x_{t} - x_{t-1}} = \sum_{t\le T} \frac{1}{t} \norm{w_t - x_{t-1}}$. 
Since we guarantee that $\norm{w_t - w_0} = O(\rho)$ for all $t\le T$ with high probability, we have $\norm{w_t-x_t} = O(\rho)$ as well. This implies the movement bound $\sum_{t\le T}\frac{1}{t}\norm{w_t - x_t} = O(\rho\log T)$ that is sufficient for our purposes.

\subsection{Putting it all together}\label{subsec:overview-putting-together}

Having described our main algorithmic ingredients, we now briefly derive the runtime bounds shown in \Cref{table:runtimes-gen,table:runtimes-bilinear}. 

\paragraph{Acceleration framework setup.} We begin by considering our accelerated proximal point method applied on the function $\fsm$. We stop the method at the first time $T$ in which $A_T = \Omega(\epsilon^{-1})$, where its potential analysis guarantees $\E \fsm(x_T) - \fsm(\xopt) = O(\epsilon)$.
Roughly speaking, our algorithm sets the $a_t$ sequence such that $\frac{a_{t+1}}{A_t} = \widetilde{\Theta}(r^{2/3})$ is constant for all iterations. We show that with an appropriate damping scheme, our algorithm will either grow $A_{t+1}$ by a multiplicative $1 + \widetilde{\Theta}(r^{2/3})$ factor or decrease a nonnegative potential function with initial value $1$ by $\widetilde{\Theta}(r^{2/3})$: this implies that $A_T$ exceeds the stopping threshold in $T = \Otil{r^{-2/3}}$ steps.

Our setting of $a_t$ means that
\[\rho_t = \frac{A_t}{a_t} r =  \widetilde{\Theta}(r^{1/3})\] is also constant for all the iterations. 
At step $t$ we apply our approximate ball oracle on $h_t(v) = A_{t+1} \fsm(\Phi_t(v))$. 
Noting that the Jacobian of $\Phi_t$ is $\frac{a_{t+1}}{A_{t+1}}I$, we have $\grad h_t(v) = a_{t+1} \grad \fsm(\Phi_t(v))$ by the chain rule. Therefore, to estimate $\grad h_t(v)$ we simply apply our estimator for $\grad \fsm$ at the point $\Phi_t(v)$ and multiply the resulting vector by $a_{t+1}$. Since our estimates for $\grad \fsm$ are always of the form $\grad f_i(\Phi_t(v))$ for some $i\in[N]$, they are bounded by $\Lf$. The gradients estimates for $h_t$ are therefore bounded by \[\Gamma = a_{t+1} \Lf = \Otilb{r^{2/3} A_{t+1} \Lf} = \Otilb{\frac{r^{2/3}}{\epsilon}\Lf},\] where the last transition holds since $A_t = O(\epsilon^{-1})$ for all iterations before stopping. 

\paragraph{Iteration and evaluation complexity.} Next, we bound the iteration count of all ball oracle calls and the total gradient evaluation complexity. For a function $h$ with stochastic gradients bounded by $\Gamma$ and target movement $\rho$, the approximate ball oracle requires $\Otil{\Gamma^2 /\rho^2}$ iterations; the complexity of finding a point that is $\Otil{\rho}$ away from the optimum of a $1$-strongly-convex function using stochastic gradients bounded by $\Gamma$. Substituting the above bounds for $\Gamma$ and $\rho$, the iteration complexity per oracle call is $\Otilb{{r^{2/3}\Lf^2}\epsilon^{-2}}$. For each ball oracle call we require $n$ individual function and gradient evaluations to set up the data structure, and (for $r = \Otil{\sqrt{\epsilon / \Lg}}$) an additional $\Otil{1}$ gradient evaluations per step with high probability, giving $\Otil{n + r^{2/3} \Lf^2 \epsilon^{-2}}$ evaluations overall. 
Since the expected number of ball oracle calls is $\Otil{r^{-2/3}}$,  with constant probability the total evaluation complexity is $\Otilb{n r^{-2/3} + \Lf^2 \epsilon^{-2}}.$ 

\paragraph{Runtime complexity.}
To account for the runtime complexity of our method, we make the simplifying assumption that each function/gradient evaluation takes $\Omega(d)$ time.
 In this case, the only term not subsumed by the function/gradient evaluation cost comes from the matrix-vector maintenance'' data structure. Our oracle implementation makes sequences of queries to our $\grad \fsm$ estimator, whose total movement is $\Otil{r}$. Therefore, the additional runtime of a single oracle call is $\Otilb{n {\Lf^2 r^2}/{\epsilon^2}}$, and for the whole algorithm the  cost is $\Otilb{n {\Lf^2 r^{4/3}}/{\epsilon^2}}$. 

 \paragraph{Choosing the ball radius $r$.}
\newcommand{\cost}{\mathcal{T}}
Finally, we discuss the optimal choice of the parameter $r$. 
For general problems~\eqref{eq:problem-gen} with $\Lg > 0$, a simple strategy is to choose the highest value of $r$ for which the linear approximation is sufficiently accurate, i.e, $\widetilde{\Theta}(\sqrt{\epsilon/\Lg})$. This yields the complexity bounds in \Cref{table:runtimes-gen}. 
However, when $\Lg$ is very small it is more computationally efficient to choose a smaller value of $r$. Letting $\cost=\Omega(d)$ denote the runtime of an individual function/gradient evaluation, the value of $r$ that minimizes the runtime terms $n\cost r^{-2/3} + n r^{4/3} {\Lf^2}/{\epsilon^2}$ is $r={\epsilon \sqrt{\cost}}/{\Lf}$, and the minimal value is $n(\cost \Lf /\epsilon)^{2/3}$. For $\Lg < {\Lf^2}/{\cost \epsilon}$ this optimal $r$ is permissible (i.e., smaller than $\sqrt{\epsilon/\Lg}$), and the total runtime of the method is $\Otilb{ n \prn*{\frac{\cost \Lf}{\epsilon}}^{2/3} +\cost \prn*{\frac{ \Lf}{\epsilon}}^{2} }$. In particular, for matrix games (where $\cost = \Theta(d)$ and $\Lg=0$) we obtain the runtimes listed in \Cref{table:runtimes-bilinear}. 

\section{Notation and conventions}\label{sec:prelim}

\paragraph{General.} We use $\xset$ to denote a general closed convex set. We use  $\Delta^d\defeq \{x\in\R^d, x\ge0,\sum_i x_i=1\}$ to denote the simplex, $\Delta^d_\nu\defeq \{x\in\Delta^d, x\ge \nu\mathbf{1}\}$ to denote the truncated simplex, and $\ball^d \defeq \{x\in\R^d, \|x\|_2\le 1\}$ to denote the unit Euclidean ball. We denote the binary indicator of event $\mathfrak{E}$ by $\indic{\mathfrak{E}}$. 

\paragraph{Vector, matrix and norm.} 
We use $\|\cdot\|$ to denote a general norm on $\xset$ and $\|\cdot\|_*=\sup_{\norm{x}\le 1} \inner{x}{\cdot}$ to denote its dual norm on the dual space $\xset^*$. For any vector $v\in\R^d$ and $p\ge 1$ we denote the $\ell_p$ norm by $\|v\|_p \defeq \prn*{\sum_{i\in[d]}|v_i|^p}^{1/p}$ with  $\|v\|_\infty = \max_{i\in[d]}|v_i|$. 
For any $p\ge 1$ we let $p^* = (1-\frac{1}{p})^{-1}$ be such that $\norm{\cdot}_{p^*}$ is dual to $\norm{\cdot}_p$.
For any matrix $A\in\R^{n\times d}$, we write $A_{ij}$ for the $(i,j)$ entry,  $A_{i:}$ for the $i$-th row as a row vector, and $A_{:j}$ for the $j$-th column as a column vector. Given  $p,q\ge 1$, we write the matrix norm $\|A\|_{p\rightarrow q} \defeq \max_{v\in\R^d, v\neq 0}\frac{\|Av\|_q}{\|v\|_p}$.

\paragraph{Functions.}
We work with convex, differentiable functions $f$ on domain $\xset$ throughout the paper. We say a function $f$ is $\Lf$-Lipschitz with respect to $\|\cdot\|$ if and only if $\abs{f(x)-f(y)}\le \Lf\|x-y\|$ for all $x,y\in\xset$. A function $f$ is $\Lg$-smooth with respect to $\|\cdot\|$ if and only if $\|\nabla f(x)-\nabla f(y)\|_*\le \Lg\|x-y\|$ for all $x,y\in\xset$. A convex function $f$ is $\mu$-strongly convex with respect to $\|\cdot\|$ if and only if for any $x,y\in\xset$, $f(x)-f(y)-\inner{\nabla f(y)}{x-y}\ge \frac{\mu}{2}\|x-y\|^2$. We call a random point $x$ an $\eps$-optimal minimizer of $f$ in expectation if $\E f(x)-\min_{x'\in\xset}f(x')\le \eps$.

\paragraph{Bregman divergences.}
Given a distance-generating function (dgf) $\dgf:\xset\rightarrow \R$, we define its induced \emph{Bregman divergence} $V^r_{x}(y) \defeq \dgf(y)-\dgf(x)-\langle\nabla \dgf(x), y-x\rangle$, and drop the superscript $\dgf$ when clear from context. Within the paper, for Euclidean space equipped with $\|\cdot\|_2$, we use $\dgf(x) = \frac{1}{2}\|x\|^2$ and its induced Bregman divergence is $V_x(y) = \frac{1}{2}\|x-y\|_2^2$, which is $1$-strongly convex in $\|\cdot\|_2$. For the simplex (or a closed convex subset thereof) equipped with $\|\cdot\|_1$, we use $\dgf(x) = \sum_i x_i\log x_i$ and its induced Bregman divergence is the Kullback–Leibler (KL) divergence $V_x(y) = \sum_{i}y_i\log(y_i/x_i)$, which is $1$-strongly convex in $\|\cdot \|_1$ by Pinsker's inequality.

\paragraph{Runtime.}
To simplify the presentation of our runtime bounds we use the following conventions throughout. We assume that the number of non-zero elements in matrix $A \in \R^{d \times n}$, denoted $\nnz(A)$, satisfies $\nnz(A) = \Omega(d+n)$. This holds for any matrix without empty rows or columns. In similar vein, we assume that the number of non-zero elements in any vector $x$ satisfies $\nnz(x)=\Omega(1)$.

We also assume that we are working in a computational model in which can pre-process any vector $v\in\R^n$ in $O(n)$ time and then be able to sample index $i$ with probability proportional to $|v_i|$ in $O(1)$ time, e.g., as in \cite{vose91sample}. If these costs are larger by multiplicative polylogarithmic factors then our final runtimes similarly grow by multiplicative polylogarithmic factors.

Throughout the paper, we use $\widetilde{O}$, $\widetilde{\Omega}$ and $\widetilde{\Theta}$ to hide poly-logarithmic factors in problem parameters, e.g. dimension, smoothness, Lipschitz constant, domain size, and desired accuracy $\eps$ and probability factor $1/\delta$.
\section{Non-Euclidean ball oracle acceleration}\label{sec:outerloop}

In this section, we describe our main acceleration framework leveraging a non-Euclidean ball oracle. The main result proved in this section is the following.

\begin{theorem}\label{thm:outerloop}
Let $f : \xset \to \R$ be a convex function which supports a gradient oracle $\mathcal{G}$ with $\norm{\mathcal{G}(x)}_* \leq G$ for all $x \in \xset$. 
For some $\mathcal{E}_0,R>0$, let $x_0, v_0\in\xset$ satisfy $f(x_0) - f(\xopt) \leq \mathcal{E}_0$ and $V_{v_0}(\xopt) \leq R^2$, where $\xopt$ is a minimizer of $f$. For any ball radius $r \le R$, oracle approximation parameter $\gamma < 1/2$, and error tolerance $\epsilon>0$, \Cref{alg:framework} has the following guarantees:
\begin{itemize}
    \item The algorithm outputs a point $x_T$ such that $\mathbb{E} \left[ f(x_T) \right] - f(\xopt) \leq \epsilon$. 
    \item The algorithm terminates after $O(\gamma^{-1/3} R^{2/3} r^{-2/3} \log(\mathcal{E}_0/\epsilon))$ iterations in expectation.
    \item Each iteration of the algorithm performs $O(1)$ arithmetic operations on elements of $\xset$ and makes a single call to a ball-restricted proximal oracle (\Cref{def:balloracle} below) with parameter $\rho =  \Theta\left( \gamma^{-1/3} R^{2/3} r^{1/3} \right)$ for a convex function $h_t$ that supports a gradient estimator $\mathcal{G}_t$ with $\norm{\mathcal{G}_t(x)}_* = O \left( \frac{\gamma^{1/3}r^{2/3} R^{4/3}}{\epsilon} G \log(\mathcal{E}_0/\epsilon) \right)$.
\end{itemize}
\end{theorem}

Our result in this section follows the outline in \Cref{subsec:overview-accel}. \Cref{alg:framework} chooses parameter sequences $A'_t, a_t$ and in each iteration calls an oracle that attempts to solve the optimization problem 
\begin{equation}
\label{eq:conc_step}
\minimize_{V_{v_t}(z) \leq \rho^2}\,\crl*{ h_t(z) + V_{v_t}(v)},~~\mbox{where}~~h_{t}(z) = A'_{t+1} f(\Phi_t(z))~~\mbox{and}~~\Phi_t(z) = \frac{A_t}{A'_{t+1}} x_t + \frac{a_{t+1}}{A'_{t+1}} z.
\end{equation}
We consider an approximate oracle that relaxes the exact solution to~\eqref{eq:conc_step} in three critical ways:
\begin{itemize}
\item We allow the oracle to return a parameter $c$, which corresponds to the Lagrange multiplier on the domain constraint,
\item We let the oracle return two points---each used for a different purpose in our final algorithm,
\item We allow the oracle's output guarantee to hold in expectation and to tolerate some additive error.
\end{itemize}

Formally, we define this relaxed oracle as follows.

\begin{definition}[Ball-restricted proximal oracle]
\label{def:balloracle}
Let $h : \xset \to \R$ be a convex function with gradient estimator $\mathcal{G}$. A \emph{$(\rho, \gamma, \cmax)$-restricted proximal oracle} takes as input $\mathcal{G}$, center point $y \in \xset$ and points $z, w \in \xset$ and a scalar $c \in [1, \cmax]$ satisfying 
\begin{equation}\label{eq:inexact-oracle-def}
\mathbb{E}\left[ \frac{ h(z) - h(u)}{c} \right] \leq \mathbb{E} \left[V_y(u) - V_{w}(u) \right] -  \gamma \mathbb{E} \left[  \indic{c \ge 2} - \indic{c < 2}  \right] \rho^2.
\end{equation}

\end{definition}
\noindent
We note that our analysis only needs the oracle parameter $\cmax$ to be finite, since its only use is verifying a condition of the optional stopping theorem. We therefore omit it for \Cref{thm:outerloop} and the subsequent lemmas used to prove it, and argue that it is indeed finite for our oracle implementations.

We now describe a final component of \Cref{alg:framework} that is omitted from the outline in \Cref{subsec:overview-accel}: \emph{momentum damping}. This mechanism, introduced in recent work on optimal methods for Monteiro-Svaiter acceleration \cite{carmon2022optimal}, handles the fact the sequence $c_t$  of regularization terms varies over time which introduces subtlety to the selection of a suitable sequence $a_t$. Given the outputs $z_{t+1}, w_{t+1}, c_{t+1}$ from the oracle in an iteration, we set $v_{t+1} = w_{t+1}$ and $x'_{t+1} = \Phi_t(z_{t+1})$. However, instead of returning $x_{t+1}=x'_{t+1}$ for the next iteration, we actually set 
\[
x_{t+1} = \frac{1}{c_{t+1}} x'_{t+1} + \frac{c_{t+1}-1}{c_{t+1}} x_t \quad \text{and} \quad A_{t+1} = A_t + \frac{a_{t+1}}{c_{t+1}}.
\]
To provide intuition for this, we consider the cases where $c_{t+1} \approx 1$ and $c_{t+1} \gg 1$. In the former case, the Lagrange multiplier on the domain constraint of \cref{eq:conc_step} is nearly inactive: thus our output $x_{t+1} \approx x'_{t+1}$ makes good progress.
On the other hand, if $c_{t+1} \gg 1$ then the ball constraint on the proximal step is extremely active. In this case, we are unable to conclude that $x'_{t+1}$ has good function error: we set $x_{t+1} \approx x_t$ and $A_{t+1} \approx A_t$ to prevent $x'_{t+1}$ from destabilizing the algorithm. However, we show that $c_{t+1}$ being very large implies that a natural potential function significantly decreases: this `'`win-win situation'' enables us to guarantee progress regardless of the actual range of $c_{t+1}$. 

\begin{algorithm2e}
	\DontPrintSemicolon
	\caption{Generalized ball acceleration framework}
	\label{alg:framework}
	\SetKwProg{Fn}{function}{}{}
	\KwInput{Convex function $f$ with gradient estimator $\mathcal{G}$}
    \KwInput{$\mathcal{O}$, a $(\rho,\gamma,\cmax)$-ball restricted proximal oracle}
    \KwInput{Parameters $r, R, \mathcal{E}_0, \epsilon > 0$}
    \KwInput{Input points $x_0, v_0 \in \R^n$ satisfying $f(x_0) - f(\xopt) \leq \mathcal{E}_0$, $V_{v_0}(\xopt) \leq R^2$}
	$A_0 = \frac{R^2}{\mathcal{E}_0}$\;
	\While{$A_t < \frac{40 R^2 \log(80 \mathcal{E}_0/\epsilon)}{\epsilon}$}{
    $a_{t+1} = \prn*{ {\sqrt{\gamma} r}/{R} }^{2/3} A_t $ and $A'_{t+1} = A_t + a_{t+1}$\;
    $\Phi_t(z) = \frac{A_t}{A'_{t+1}} x_t + \frac{a_{t+1}}{A'_{t+1}} z$ and $\rho = \frac{A'_{t+1}}{a_{t+1}} r = \left(1 +  \left( \frac{R}{\sqrt{\gamma} r} \right)^{2/3} \right) r$ \Comment{$\rho$ is constant across iterations}
    $h_t(z) \defeq A'_{t+1} f\left( \Phi_t(z) \right)$ \Comment{$\norm{x-y} \leq \rho$ implies $\norm{\Phi_t(x) - \Phi_t(y)}\leq \frac{a_{t+1}}{A'_{t+1}} \rho = r$}\label{line:r-check}
    $\mathcal{G}_t(z) = a_{t+1} \mathcal{G} (\Phi_t(z) )$ \Comment{$\mathcal{G}_t$ is a stochastic gradient estimator for $h_t$} 
    $z_{t+1} , v_{t+1}, c_{t+1} =  \mathcal{O}(\mathcal{G}_t, v_t, \rho)$ \label{line:oraclecall}\;
    $x_{t+1}= \frac{1}{c_{t+1}} \Phi_t(z_{t+1}) + \frac{c_{t+1}-1}{c_{t+1}} x_t$ \;
    $A_{t+1} = \frac{c_{t+1}-1}{c_{t+1}} A_t + \frac{1}{c_{t+1}} A'_{t+1} = A_{t} + \frac{a_{t+1}}{c_{t+1}}$ \label{line:A_def}\;
    $t = t+1$\;
    }
    \KwReturn{$x_{t}$}
\end{algorithm2e}

We begin the analysis by proving a potential decrease bound.
\begin{lemma}[Potential decrease]
\label{lem:potential}
Consider an execution of \Cref{alg:framework}. Let $\xopt \in \xset$ be a minimizer of $f$ and for each iteration $t$ let 
\[
E_t \defeq f(x_t) -  f(\xopt) \quad \text{ and } \quad D_t \defeq V_{v_t}(\xopt)\,.
\]
Let $P_t \defeq A_t E_t + D_t$ where $A_t$ is defined on \cref{line:A_def}. Then for any $t \geq 0$
\[
\mathbb{E} \left[ P_{t+1} \right] \leq P_t - \gamma \mathbb{E} \left[ \indic{c_{t+1} \ge 2} - \indic{c_{t+1} < 2} \right] \rho^2.
\]
where the expectation is taken over the choice of randomness in a single iteration.
\end{lemma}

\begin{proof}
By the guarantee of the ball-restricted proximal oracle $\mathcal{O}$, we have
\begin{align*} 
\mathbb{E} \left[ \frac{h_{t}(z_{t+1}) - h_{t}(u)}{c_{t+1}} \right] \leq \mathbb{E} \left[ V_{v_t}(u) - V_{v_{t+1}}(u) \right] - \gamma \mathbb{E} \left[ \indic{c_{t+1} \ge 2} - \indic{c_{t+1} < 2} \right] \rho^2
\end{align*}
We will bound the left-hand side of this inequality. First, observe that for any choice of $c_{t+1}, z_{t+1}$, 
\begin{align*}
\frac{h_t(z_{t+1}) -h_t(u)}{c_{t+1}} &\geq \frac{A'_{t+1} f(\Phi_t(z_{t+1})) - A_t f(x_t) - a_{t+1} f(u)}{c_{t+1}} \\
&= \left( \frac{A'_{t+1}}{c_{t+1}} f(\Phi_t(z_{t+1})) + A_t \frac{c_{t+1} - 1}{c_{t+1}} f(x_t) \right) - A_t f(x_t) -  \frac{a_{t+1}}{c_{t+1}} f(u) \\
&\geq A_{t+1} f(x_{t+1} ) - A_t f(x_t) -  \frac{a_{t+1}}{c_{t+1}} f(u)\\
&= A_{t+1} E_{t+1} - A_t E_t. 
\end{align*}
where the inequalities follow from the convexity of $f$. Substituting this in yields 
\[
\mathbb{E} \left[ A_{t+1} E_{t+1} - A_t E_t  \right] \leq \mathbb{E} \left[ D_t - D_{t+1} \right] - \gamma \mathbb{E} \left[ \indic{c_{t+1} \ge 2} -  \indic{c_{t+1} < 2} \right] \rho^2.
\]
and rearranging gives the claim. 
\end{proof}

Iterating this potential decrease lemma gives a full complexity bound. 

\begin{lemma}
\label{lem:error_and_iterations}
Let $x_T$ be the output of the above algorithm. We have
\[
\mathbb{E} \left[ f(x_T) -  f(\xopt) \right] \leq \epsilon.
\]
In addition, the algorithm performs at most 
\[
18 \left( \frac{R}{\sqrt{\gamma} r} \right)^{2/3} \log\left( \frac{80 \mathcal{E}_0}{\epsilon} \right) 
\] 
iterations in expectation, where each iteration calls a ball-restricted proximal oracle (\Cref{def:balloracle}).
\end{lemma}
\begin{proof}
Let $T$ denote the (random) iteration where the algorithm returns $x_T$, i.e., the first $T$ for which $A_T \ge \frac{40 R^2 \log(80 \mathcal{E}_0/\epsilon)}{\epsilon}$. Define the random process 
\[
Q_t = P_t + \sum_{i=1}^t \gamma  \left( \indic{c_{i} \ge 2} - \indic{c_{i} < 2} \right) \rho^2 .
\]
We recall that 
\begin{align}
\label{eq:A_grows}
A_{t+1} = A_t + \frac{a_{t+1}}{c_{t+1}} = A_t \left( 1 + \frac{1}{c_{t+1}} \left( \frac{\sqrt{\gamma} r}{R} \right)^{2/3} \right) \implies A_{t+1} \geq A_{t} \exp \left( \frac{1}{2 c_{t+1}} \left( \frac{\sqrt{\gamma}  r}{R} \right)^{2/3} \right).
\end{align}
As $c_{t+1} \leq \cmax$ by the definition of the ball-restricted proximal oracle, we observe that with probability $1$
\[
A_t \geq A_0 \exp \left( \frac{1}{2 \cmax} \left( \frac{\sqrt{\gamma} r}{R} \right)^{2/3} t \right).
\]
As we terminate when $A_T \ge \frac{40 R^2 \log(80 \mathcal{E}_0/\epsilon)}{\eps}$, this implies that $T$ is finite with probability $1$. \Cref{lem:potential} implies that $Q_t$ is a supermartingale and therefore, by the optional stopping theorem, we have
\begin{equation}
\label{eq:stopping}
\mathbb{E} \left[ Q_T \right] \leq Q_0 \leq 2 R^2.
\end{equation}
Now define 
\[
T_1 = \sum_{i=1}^T \indic{c_{i} \ge 2} \quad \text{and} \quad T_2 = \sum_{i=1}^T \indic{c_{i} < 2}.
\]
By definition, we have $T = T_1 + T_2$ and
\[
Q_T = P_T +  \gamma \rho^2 ( T_1 - T_2).
\]
Now for any iteration with $c_{t+1} < 2$, \cref{eq:A_grows} implies
\[
A_{t+1} \geq A_{t} \exp \left( \frac{1}{4} \left( \frac{\sqrt{\gamma} r}{R} \right)^{2/3} \right) \implies A_T \geq A_0 \exp \left( \frac{T_2 }{4} \left( \frac{\sqrt{\gamma} r}{R} \right)^{2/3} \right).
\]
As $A_{T-1} < \frac{40 R^2 \log(80 \mathcal{E}_0/\epsilon)}{\eps}$ and $A_T \leq \left(1 +  (\sqrt{\gamma}r/R)^{2/3} \right) A_{T-1} < \frac{80 R^2 \log(80 \mathcal{E}_0/\epsilon)}{\eps}$, this implies that with probability $1$
\begin{align*}
\frac{80 R^2 \log(80 \mathcal{E}_0/\epsilon)}{\eps} \geq \frac{R^2}{\mathcal{E}_0} \exp \left( \frac{1}{4} \left( \frac{\sqrt{\gamma} r}{R} \right)^{2/3} T_2 \right) \implies T_2 &\leq 4 \left( \frac{R}{\sqrt{\gamma} r} \right)^{2/3} \log \left( \frac{80 \mathcal{E}_0 \log(80 \mathcal{E}_0/\epsilon)}{\epsilon} \right) \\
&\leq 8 \left( \frac{R}{\sqrt{\gamma} r} \right)^{2/3} \log \left( \frac{80 \mathcal{E}_0}{\epsilon} \right)
\end{align*}
where the last inequality follows from $\alpha \log \alpha \leq \alpha^2$ for any $\alpha > 0$. This implies
\begin{equation}
\label{eq:stopping2}
\mathbb{E} \left[ P_T + \gamma \rho^2 T_1  \right] = \mathbb{E} [Q_T + \gamma \rho^2 T_2] \leq 2 R^2 +  8 \gamma \rho^2 \left( \frac{R}{\sqrt{\gamma} r} \right)^{2/3} \log \left( \frac{80 \mathcal{E}_0}{\epsilon} \right)
\end{equation}
where the inequality follows from the above bound on $T_2$ and \cref{eq:stopping}. Now, note that 
\[
\rho = \left(1 + \left(\frac{R}{\sqrt{\gamma} r}\right)^{2/3} \right) r < \frac{2 R^{2/3} r^{1/3}}{\gamma^{1/3}}
\]
as $r < R$. 
Substituting this into \cref{eq:stopping2}, we obtain
\begin{align*}
& \frac{40 R^2 \log(80 \mathcal{E}_0/\epsilon)}{\epsilon} \mathbb{E} \left[ f(x_T) - f(\xopt) \right] \leq \mathbb{E} \left[ A_T E_T \right] \\
& \hspace{5em} \leq \mathbb{E} \left[ P_T + \gamma \rho^2 T_1 \right] \leq  2 R^2 +  8 \gamma \rho^2 \left( \frac{R}{\sqrt{\gamma} r} \right)^{2/3}  \log \left( \frac{80 \mathcal{E}_0}{\epsilon} \right) \\
& \hspace{5em} \leq 2 R^2 + 32 R^2  \log \left( \frac{80 \mathcal{E}_0}{\epsilon} \right) \leq  34 R^2  \log \left( \frac{80 \mathcal{E}_0}{\epsilon} \right)
\end{align*}
and therefore $\mathbb{E} \left[ f(x_T) - f(\xopt) \right] \leq \epsilon$. In addition, \cref{eq:stopping2} also yields
\begin{align*}
\mathbb{E} [ T_1] &\leq \frac{2 R^2}{\gamma \rho^2}  + 4 \left( \frac{R}{\sqrt{\gamma} r} \right)^{2/3}  \log \left( \frac{80 \mathcal{E}_0 \log(80 \mathcal{E}_0/\epsilon)}{\epsilon} \right) \\
&\leq \frac{2 \gamma^{2/3} R^2}{\gamma R^{4/3} r^{2/3}} + 8 \left( \frac{R}{\sqrt{\gamma} r} \right)^{2/3} \log \left( \frac{80 \mathcal{E}_0}{\epsilon} \right) \leq 10 \left( \frac{R}{\sqrt{\gamma} r} \right)^{2/3} \log \left( \frac{80 \mathcal{E}_0}{\epsilon} \right).
\end{align*}
Thus, the expected number of iterations of the method satisfies
\[
\mathbb{E} [T] = \mathbb{E} [T_2] + \mathbb{E} [T_1] \leq  18 \left( \frac{R}{\sqrt{\gamma} r} \right)^{2/3} \log \left( \frac{80 \mathcal{E}_0}{\epsilon} \right).
\]
\end{proof}

We combine these facts to prove \Cref{thm:outerloop}.

\begin{proof}[{Proof of \Cref{thm:outerloop}}]
\Cref{lem:error_and_iterations} implies the first two items in \Cref{thm:outerloop}. For the third item, we observe that the only nontrivial step of the while loop is on \cref{line:oraclecall}, which performs a single call to $\mathcal{O}$ with gradient estimator $\mathcal{G}_t$ and parameter $\rho = O(r^{1/3} R^{2/3})$. Any $t$ prior to terminating has $A_t = O\prn*{\frac{R^2 \log(\mathcal{E}_0/\epsilon)}{\epsilon}}$. Thus, for any $x \in \xset$ 
\begin{align*}
\norm{\mathcal{G}_t(x)}_{*} &= a_{t+1} \norm{\mathcal{G}(\Phi_t(x))}_{*} \leq \left( \frac{\sqrt{\gamma} r}{R} \right)^{2/3} A_t G \\
& = O \left( \gamma^{1/3}\frac{r^{2/3}}{R^{2/3}} \cdot \frac{R^2 \log(\mathcal{E}_0/\epsilon)}{\epsilon} G \right)  = O \left( \frac{\gamma^{1/3}R^{4/3} r^{2/3} \log(\mathcal{E}_0/\epsilon)}{\epsilon} G \right).
\end{align*}
\end{proof}

\section{Ball oracle implementation}\label{sec:innerloop}

In this section, we develop \Cref{alg:ball-oralce} which implements a ball-restricted proximal introduced in \Cref{def:balloracle} in the previous section. The algorithm combines last-iterate proximal mirror descent (\mprox,  \Cref{alg:mirror-descent}) with a careful bisection procedure ($\linesearch$ in \Cref{alg:ball-oralce}). For high-level description of the algorithm, see \Cref{subsec:overview-oracle-impl}. 

Let us briefly describe \Cref{alg:ball-oralce,alg:mirror-descent}. 
The $\linesearch$ procedure tries (with high probability) $\widetilde{O}(1)$ values of $\lambda$ and finds one for which $\uprox = \arg\min_{x\in\xset} h(x) + \lambda V_y(x)$ satisfies $V_y(\uprox)=\widetilde{\Theta}(\rho^2)$. Using this $\lambda$ we call \Cref{alg:mirror-descent} once again to obtain random outputs $z,w$ independent of the random bits that produce $\lambda$. By properly choosing the step sizes and number of iterations, we argue that the results satisfy the restricted proximal ball oracle condition~\eqref{eq:inexact-oracle-def}.

Our algorithm has additional properties that enable efficient gradient estimation for the problems we study. First, all iterations stay within a radius-$\rho$ norm ball centered at $y$, as $\mprox$ aborts whenever going outside the radius. This enables efficiently sampling from softmax distribution using linear approximation and rejection sampling, see ~\Cref{sec:grad-est}. Second, due to the ``last-iterate'' mechanism (which performs iterate averaging before the stochastic gradient queries), the total movement of iterates throughout~\Cref{alg:ball-oralce} is also bounded by $\widetilde{O}(\rho)$. This movement bound is used for bounding the runtime when querying the data structure that we designed in~\Cref{alg:matrix_vector_maintenance} for constructing $\gest$. 

\begin{algorithm2e}[p]
	\DontPrintSemicolon
	\caption{Projection-free ball oracle implementation $\mathcal{O}(\mathcal{G}, y, \rho)$}\label{alg:ball-oralce}
	\KwInput{Objective $h:\xset\to\R$ with gradient estimator $\gest$, center point $y\in\xset$, radius $\rho$}
	\KwParameters{Gradient bound $\GG$,
    1-strongly-convex dgf $\dgf$ and associated Bregman divergence $V$, triangle inequality factor $\tau$}
	\KwParameters{constant $\bigC=66\cdot 2^{12}$, error probability $\delta \le \frac{\rho^2}{2^{14}(\sqrt{2}R\GG+R^2) \tau^5}$}
	$\lambda \gets \linesearch(\gest, y, \rho)$, $\delta_0\gets\delta$, $\eta \gets \frac{\rho^2\lambda}{\bigC\cdot \log(16/\delta_0)\cdot \tau^5\GG^2}$, $T \gets \frac{4\tau}{\eta\lambda}$\;
	$(z,w,\flag)\gets \mprox(\gest,\dgf, y,\rho, \lambda,\eta, T)$\;\label{line:main-mprox-step}
	\Return{} $z$, $w$ and $c=\lambda + \frac{1}{\eta T}$
\BlankLine
\Fn{$\linesearch(\gest,y,\rho)$
	}{
		Set $\lambda_{\max} = \frac{16\tau\GG}{\rho}$, $\lambda_0 = \lambda_{\min} = 1$, $\eta_0 = \frac{\rho^2\lambda_{\min}}{\bigC\cdot\log(16/\delta)\cdot \tau^5\GG^2}$, $T_0 = \frac{4\tau}{\eta_0\lambda_{\min}}$ and $\Kmax = \ceil*{\log\frac{9600\tau^3\GG^3}{\rho^3}}+1$\;
	$(z^{(0)}, w^{(0)}, \flag^{(0)})\gets \mprox(\gest,\dgf,y,\rho,\lambda_0,\eta_0, T_0)$\;
	\lIf{$V_y(z^{(0)})<\frac{\rho^2}{64\tau}$}{\label{line:initial-check}\Return{} $\lambda_{\min}$}
	
	\For{$k=1,\ldots,\Kmax$}{
	$\lambda_k = \frac{1}{2}\prn*{\lambda_{\max}+\lambda_{\min}}$, 	$\delta_k =\frac{\delta}{8k^2}$, an $\eta_k =\frac{\rho^2\lambda_k}{\bigC\cdot \log(16/\delta_k)\cdot \tau^5\GG^2}$~~\Comment{satisfying $\frac{\bigC\log\frac{16}{\delta_k}\eta_k\GG^2}{\lambda_k} = \frac{\rho^2}{\tau^5}$}
	$T_k=\frac{4\tau}{\eta_k\lambda_k}$~~\Comment{satisfying $\frac{2}{\lambda_k\eta_k T_k}=\frac{1}{2\tau}$}
	$(z^{(k)},w^{(k)}, \flag^{(k)})\gets \mprox(\gest, \dgf, y,\rho, \lambda_k,\eta_k,T_k)$\;
		\lIf{$\flag^{(k)}$ \textup{or} $V_{y}(z^{(k)})>\frac{\rho^2}{64\tau}$}{
				$\lambda_{\min}=\lambda_k $
		}
		\lElseIf{$V_{y}(z^{(k)})<\frac{\rho^2}{256\tau^3}$}{
		$\lambda_{\max}=\lambda_k$}
		\lElse{
		\Return{$\lambda_k$} 
		}
		}
	\Return{$\lambda_{\Kmax}$}\Comment{the probability of reaching this line is less than $\delta/2$}
	}
\end{algorithm2e}

\begin{algorithm2e}[p]
		\DontPrintSemicolon
		\caption{Last-iterate proximal mirror descent $\mprox(\gest, \dgf, y, \rho, \lambda, \eta, T)$}\label{alg:mirror-descent}
		\KwInput{Objective function $h:\xset\to\R$ with gradient estimator $\gest$, 1-strongly-convex dgf $\dgf$ (and associated Bregman divergence $V$), center point $y\in\xset$, radius $\rho$, regularization parameter $\lambda \ge0$, step size $\eta$, iteration budget $T$}
		
		Set $w_0 = x_0 = y$\;
		Set $\flag = \False$ ~~\Comment{monitor if iterations go out of $\rho$-radius from center $y$}
		\For{$t=1,\ldots,T$}{
			$x_t = \frac{1}{t} \sum_{i=0}^{t-1} w_i = \frac{t-1}{t} x_{t-1} + \frac{1}{t} w_{t-1}$\;
			\lIf{$\|x_t-y\|\ge \rho$}{
			$\flag=\True$~~\textbf{break}}
			$\hg_t = \gest(x_t)$ ~~\Comment{$g_t \defeq \Ex*{\hg_t |x_t}\in\del h(x_t)$}
			$w_t = \argmin_{w\in\xset}\crl*{ \eta\brk*{\inner{\hg_t}{w} + \lambda V_{y}(w) } + V_{w_{t-1}}(w)}$
			}\label{line:mirror-step}
			$\tw_T = 
			\argmax_{v\in\xset} \crl*{\inner{\frac{\sum_{t=1}^T \grad \dgf(w_t) + \frac{1}{\lambda \eta}\grad \dgf(w_T)}{T + \frac{1}{\lambda \eta}}}{v}-\dgf(v)}
			=
			\grad \dgf^*\prn*{ \frac{\sum_{t=1}^T \grad \dgf(w_t) + \frac{1}{\lambda \eta}\grad \dgf(w_T)}{T + \frac{1}{\lambda \eta}}}$\;\label{line:define-wtilde}
			\lIf{\textup{$\flag = \False$}}{%
			\Return{} $x_T$, $\tw_T$, $\flag=\False$. 
            }%
			$z = y + \rho \frac{x_t - y}{\norm{x_t-y}}$\Comment*[h]{Arbitrarily selecting a point with distance $\rho$ from $y$}\;
			{\Return{} $z$, $z$, $\flag=\True$~~\Comment{return arbitrary point if outside radius $\rho$}}\label{line:arbitrary}
	\end{algorithm2e} 
The formal guarantees of our algorithm require the following notion of \emph{$\tef$-triangle inequality} for Bregman divergences.
\begin{definition}[$\tau$-triangle inequality]\label{def:tau-triangle} For any $\tau \ge 1$, a domain $\xset$ and Bregman divergence $V$ satisfy a $\tef$-triangle inequality, for all $x,y,z\in\xset$,
	\begin{equation}
		V_{x}(z) + V_{z}(x) \le \tef \prn*{
		\min\{V_{x}(y), V_{y}(x)\} + \min\{V_{y}(z), V_{z}(y)\}}.
	\end{equation}
\end{definition}

\noindent With this definition in hand, we state the main guarantees of~\Cref{alg:ball-oralce}.

\begin{theorem}\label{thm:oracle-impl}
	Let $\xset$ be a closed convex set, let $h:\xset \to \R$ be a convex function with gradient estimator $\gest$ that satisfies $\norm{\gest(x)}_* \le \GG$ with probability 1, and let $\xset$ and $V$ satisfy a $\tef\ge 4$ triangle inequality (\Cref{def:tau-triangle}) as well as $\max_{x,y\in\xset}V_x(y)\le R^2$. 
	 Let $\hat{T}_k$ to be the number of iterations in the $k$'th call to $\mprox$. Then, for any radius $\rho>0$, center point $y\in\xset$, for error probability $\delta\le \frac{\rho^2}{2^{14}(\sqrt{2}R\GG+3R^2)\tau^5}$, the following holds:
	\begin{enumerate}
		\item \Cref{alg:ball-oralce} implements a $(\rho,\gamma,\cmax)$ restricted proximal oracle for function $h$, with $\gamma = \frac{1}{2^{13} \tef^5}$ and $\cmax = \frac{32\tau\GG}{\rho}$. That is, the outputs $z,w$ and $c$ of \Cref{alg:ball-oralce} satisfy 
		\begin{equation}\label{eq:oracle-impl-error-bound}
			\E \frac{h(z) - h(u)}{c} \le \E \brk*{ V_y(u) - V_{w}(u) } - \frac{1}{2^{13} \tef^5}\E \prn*{\indic{c\ge 2}-\indic{c<2}}\rho^2 ~~\mbox{for all}~u\in\xset
		\end{equation}
		and $c \le \frac{32\tau\GG}{\rho}$ with probability 1.
		\item With probability $1$, the queries $x_1\pind{k}, \ldots, x_{\hat{T}_k}\pind{k}$ that \Cref{alg:mirror-descent} makes to $\gest$ when called in the $k$'th iteration of \Cref{alg:ball-oralce} satisfy
		\begin{equation*}
			\norm{x_t\pind{k} - y} \le \rho 
			~~\mbox{for all}~t\le \hat{T}_k~\mbox{and}~k\le K.
		\end{equation*}
		\item With probability $1$, the sequences $\{x_1\pind{k}, \ldots, x_{\hat{T}_k}\pind{k}\}_{k=0}^K$ defined above satisfy
		\begin{equation*}
			\sum_{k=0}^K \sum_{t=1}^{\hat{T}_k} \norm{ x_{t}\pind{k} - x_{t-1}\pind{k} } \le \rho \cdot 2 \Kmax\log\frac{4\bigC\log(16\Kmax^2/\delta)\tau^6\GG^2}{\rho^2},
			\end{equation*}
where $\Kmax = \lceil\log\frac{9600\tau^3\GG^3}{\rho^3}\rceil+1$.
	\item  %
		\Cref{alg:ball-oralce} 
		makes at most $O\prn*{\frac{ \GG^2}{\rho^2}\cdot \tau^6\prn*{\log\frac{1}{\delta}+\log\log\frac{\tau\GG}{\rho}}\cdot \log\frac{\tau\GG}{\rho}}$ calls to $\gest$ and the same number of mirror-descent steps.
	\end{enumerate}
\end{theorem}

The remainder of this section is organized as follows. First, in~\Cref{ssec:tef} we give two examples of Bregman divergences satisfying the $\tef$-triangle inequality and calculate the particular values of $\tef$ in difference cases. Then we analyze~\Cref{alg:mirror-descent} and the~$\linesearch$ procedure in~\Cref{ssec:mirror-descent,ssec:bisection}, respectively. Finally, in~\Cref{ssec:oracle-impl} we combine those results to prove the main proposition of the section.

\subsection{Divergences satisfying $\tef$-triangle inequality}\label{ssec:tef}

Throughout the paper we mainly consider two divergences, $V_x(y)=\frac{1}{2}\|x-y\|^2$ for the ball setup and $V_{x}(y) = \sum_{i\in[d]}y_i\log(y_i/x_i)$ for the simplex setup. In this section we show both divergences satisfy $\tef$-triangle inequality with $\tef = \widetilde{\Theta}(1)$.

\begin{example}[Euclidean setup with $\ell_2$-norm-squared]\label{example:euclidean-inequality}
Any $\xset\subseteq \R^d$ and $V_x(y)=\frac{1}{2}\|x-y\|^2_2$, satisfy a $\tef$-triangle inequality with $\tef = 4$.
\end{example}

\begin{example}[Truncated simplex setup with KL-divergence]
\label{example:truncated-entropy-inequality}
For any $\nu\in(0,1/4]$, the simplex $\Delta^d_\nu\defeq \{x\in\Delta^d, x\ge \nu\mathbf{1}\}$ and KL-divergence $V_x(y) = \sum_{i\in[d]}y_i\log\frac{y_i}{x_i}$ satisfy a $\tef$-triangle inequality with $\tef = 6\log(\nu^{-1})$.
\end{example}

\Cref{example:euclidean-inequality} is an immediate consequence of the standard triangle inequality, but \Cref{example:truncated-entropy-inequality} is less obvious and stems from the following connection between KL-divergence to squared Hellinger distance. \yujia{TODO for arXiv: show tight up to constant}
\begin{lemma}\label{lem:truncated-entropy-inequality}
	Given $\nu\in(0,1/4]$ and any  $x,y\in\Delta^d_\nu$, consider the KL-divergence $V_{x}(y) = \sum_{i\in[d]}y_i\log\frac{y_i}{x_i}$ and squared Hellinger distance $H^2(x,y) = \frac{1}{2}\|\sqrt{x}-\sqrt{y}\|^2$. Then 
	\[
		V_{x}(y)+V_{y}(x)\le \prn[\big]{6\log \tfrac{1}{\nu} }H^2(x,y).
	\] 
\end{lemma}
\begin{proof}%
We have
\begin{align*}
V_{x}(y)+V_{y}(x) & = \sum_{i\in[d]}(y_i-x_i)\log\frac{y_i}{x_i} = \frac{1}{2}\sum_{i\in[d]}f\left(\frac{y_i}{x_i}\right)\left(\sqrt{y_i}-\sqrt{x_i}\right)^2,~~
\text{where}~~f(t) = \frac{2(t-1)\log t}{(\sqrt{t}-1)^2}.
\end{align*}
The lemma follows from noting that $\max_{t\in[\nu, 1/\nu]}f(t)  = \frac{2(\frac{1}{\nu}-1)\log\frac{1}{\nu}}{\left(\sqrt{\frac{1}{\nu}}-1\right)^2}\le 6\log\frac{1}{\nu}$.
\end{proof}

\begin{proof}[Proof of~\Cref{example:truncated-entropy-inequality}]
We first use the AM-GM inequality of Hellinger distance, which gives $2H^2(x,z)+2H^2(z,y)\ge H^2(x,y)$, and consequently we have
\[
\min(V_x(z), V_z(x))+\min(V_y(z),V_z(y))\stackrel{(i)}{\ge} 2H^2(x,z)+2H^2(z,y)\ge H^2(x,y)\stackrel{(ii)}\ge \frac{1}{6\log(\nu^{-1})}(V_x(y)+V_y(x)).
\]
Here we use $(i)$ the well-established inequality $V_x(z)\ge 2H^2(x,z)$ (see, e.g. \citet{reiss2012approximate}), and $(ii)$ the inequality shown in~\Cref{lem:truncated-entropy-inequality}. This proves the desired claim.
\end{proof}

We remark that any divergence $V$ satisfying the $\tef$-triangle inequality on $\xset$  is also symmetric in its arguments up to factor $\tef$, formally stated as follows.
\begin{corollary}
	For any closed convex set $\xset$ and some Bregman divergence $V$ on $\xset$ satisfying a $\tef$-triangle inequality, then $\frac{1}{\tef}V_x(y)\le V_y(x)\le \tef V_x(y)$.
\end{corollary}
\begin{proof}
We can apply the definition of $\tef$-triangle inequality with $z=y$ to get that $\min\{V_x(y), V_y(x)\}+\min\{V_y(y), V_y(y)\}\ge \frac{1}{\tau}\left(V_x(y)+V_y(x)\right)\ge \frac{1}{\tau}V_x(y)$, which implies the first inequality. The second inequality follows by symmetry.
\end{proof}

\subsection{Analysis of \Cref{alg:mirror-descent}}\label{ssec:mirror-descent}

In this section, we provide the main analysis and guarantees for $\mprox$~(\Cref{alg:mirror-descent}). The first lemma is a deterministic error bound for last-iterate proximal mirror descent. Throughout the analysis we use 
\[
	\Hprox(x)\defeq h(x) + \lambda V_y(x)~~\mbox{and}~~\uprox \defeq \arg\min_{x\in\xset}\Hprox(x)
\]
for the regularized objective function and its minimizer, respectively.

\begin{lemma}\label{lem:mirror-descent-prob1}
	Let $\xset$ be a closed convex set, and $h:\xset \to \R$ be a convex function with gradient estimator $\gest$ that satisfies $\norm{\gest(x)}_* \le \GG$ with probability 1, and let $y\in\xset$ and $\lambda \ge 0$. 
	The iterates of \Cref{alg:mirror-descent} satisfy, for all $u\in\xset$,
	\begin{equation}\label{eqn:progress-realize}
		\Hprox(x_{T})-\Hprox(u)\le-\frac{\lambda}{T}\sum_{t=1}^{T}V_{w_{t}}\left(u\right)+\frac{V_{y}\left(u\right)-V_{w_{T}}\left(u\right)}{\eta T}+\frac{\eta}{2}\GG^{2}+\frac{1}{T}\sum_{t=1}^{T}\inner{g_{t}-\hat{g_{t}}}{w_{t-1}-u}.
	\end{equation}
\end{lemma}

\begin{proof}
At iteration $t\in[T]$, the optimality condition for each iteration of~\Cref{line:mirror-step} gives, for any $u\in\xset$, 
\[\inner{\eta \hat{g}_t+\eta\lambda \nabla V_{y}(w_t)+\nabla V_{w_{t-1}}(w_t)}{w_t-u}\le 0,\]
which by rearranging terms implies 
\begin{align}\label{eqn:optimality-both}
\inner{\hat{g}_t+\lambda\nabla V_{y}(w_t)}{w_t-u}&~\le \frac{1}{\eta}\inner{-\nabla V_{w_{t-1}}(w_t)}{w_t-u} =\frac{1}{\eta}\left(V_{w_{t-1}}(u)-V_{w_t}(u)-V_{w_{t-1}}(w_t)\right),
\end{align}
where we use the three-point equality following the definition of Bregman divergence for the last equality.

Now, for the terms on the LHS of~\eqref{eqn:optimality-both}, by applying three-point equality again, \begin{align}\label{eqn:optimality-1}
\lambda\langle\nabla V_y(w_t), w_t-u\rangle = \lambda\left(V_{w_t}(u)+V_y(w_t)-V_y(u)\right).
\end{align}
By rearranging terms
\begin{equation}
\begin{aligned}\label{eqn:optimality-2}
	& \inner{\hat{g}_t}{w_t-u} = \inner{\hat{g}_t}{w_{t}-w_{t-1}}+\inner{g_t}{w_{t-1}-u}+ \inner{\hat{g}_t-g_t}{w_{tc-1}-u}\\
	& \hspace{1em} \stackrel{(i)}{=} \inner{\hat{g}_t}{w_t-w_{t-1}}+\inner{g_t}{tx_t-(t-1)x_{t-1}-u}+\inner{\hat{g}_t-g_t}{w_{t-1}-u}\\
	& \hspace{1em} \stackrel{(ii)}{\ge} -\frac{\eta}{2}\|\hat{g}_t\|_*^2-\frac{1}{2\eta}\|w_t-w_{t-1}\|^2+(t-1) \left(h(x_t)-h(x_{t-1})\right)+\left(h(x_t)-h(u)\right)+\inner{\hat{g}_t-g_t}{w_{t-1}-u}\\
	& \hspace{1em} \stackrel{(iii)}{\ge} -\frac{\eta}{2}\|\hat{g}_t\|_*^2-\frac{1}{\eta}V_{w_{t-1}}(w_{t})+(t-1) \left(h(x_t)-h(x_{t-1})\right)+\left(h(x_t)-h(u)\right)+\inner{\hat{g}_t-g_t}{w_{t-1}-u}.
\end{aligned}
\end{equation}
Here we use $(i)$ the relation that $tx_t = (t-1)x_{t-1}+w_{t-1}$, $(ii)$ the AM-GM inequality and convexity of $h$, and $(iii)$ the 1-strong-convexity of the distance generating function.

Plugging~\Cref{eqn:optimality-1,eqn:optimality-2} back into~\Cref{eqn:optimality-both} and rearranging terms,
\begin{align*}
& t(h(x_t)-h(u))-(t-1)	(h(x_{t-1})-h(u))\\
& \hspace{5em} \le \lambda \left(V_y(u)-V_{w_t}(u)-V_y(w_t)\right)+\frac{1}{\eta} \left(V_{w_{t-1}}(u)-V_{w_t}(u)-V_{w_{t-1}}(w_t)\right)\\
& \hspace{8em} +\frac{\eta}{2}\|\hat{g}_t\|_*^2+\frac{1}{\eta}V_{w_{t-1}}(w_{t})+\inner{g_t-\hat{g}_t}{w_{t-1}-u}\\
& \hspace{5em} \le \lambda \left(V_y(u)-V_{w_t}(u)-V_y(w_t)\right)+\frac{1}{\eta} \left(V_{w_{t-1}}(u)-V_{w_t}(u)\right)+\frac{\eta}{2}\GG^2+\inner{g_t-\hat{g}_t}{w_{t-1}-u}.
\end{align*}
Here for the last inequality we use $\|\hat{g}_t\|_*\le \GG$ by definition of the gradient estimator.

Averaging over $t\in[T]$, we have for any $u\in\xset$,
\begin{align*}
h(x_T)-h(u)	\le & \lambda V_y(u)-\frac{\lambda}{T}\sum_{t\in[T]}V_{w_t}(u)-\frac{\lambda}{T}\sum_{t\in[T]}V_{y}(w_t)\\
& \hspace{3em}+\frac{1}{\eta T} \left(V_y(u)-V_{w_T}(u)\right)+\frac{\eta}{2}\GG^2+\frac{1}{T}\sum_{t\in[T]}\inner{g_t-\hat{g}_t}{w_{t-1}-u},\\
	\le & \lambda V_y(u)-\frac{\lambda}{T}\sum_{t\in[T]}V_{w_t}(u)-\lambda V_{y}(x_T)\\
& \hspace{3em}+\frac{1}{\eta T} \left(V_y(u)-V_{w_T}(u)\right)+\frac{\eta}{2}\GG^2+\frac{1}{T}\sum_{t\in[T]}\inner{g_t-\hat{g}_t}{w_{t-1}-u},
\end{align*}
where the last inequality is due to the convexity of $V_y(\cdot)$, $V_y(w_0)=0$ and $x_T = \frac{1}{T}\sum_{t=0}^{T-1}w_t$ so that $\frac{1}{T}\sum_{t\in[T]}V_{y}(w_t)\ge \frac{1}{T}\sum_{t=0}^{T-1}V_{y}(w_t)\ge V_y(x_T)$. Rearranging terms concludes the proof.
\end{proof}

Combining with $\tau$-triangle inequality of divergence $V$, we can get the following in-expectation progress guarantee~\eqref{eqn:progress-expect}.
 
\begin{corollary}\label{lem:mirror-descent}
	In the setting of \Cref{thm:oracle-impl}, the outputs $z,w$ and $\flag$ of~\Cref{alg:mirror-descent} satisfy
	\begin{equation}\label{eqn:progress-expect}
	\begin{aligned}
		\E h(z) - h(u) & \le \prn*{\lambda + \frac{1}{\eta T}}  \E \brk*{V_y(u) - V_{w}(u)} + \eta\GG^2 - \prn*{\frac{\lambda}{\tef} - \frac{1}{\eta T}} \E  V_y(\uprox)\\
		& \hspace{5em} +
		\left(\sqrt{2}R\GG+\left(\lambda+\frac{1}{\eta T}\right)R^2\right)\Pr(\flag=\True).
	\end{aligned}
	\end{equation}
\end{corollary}

\begin{proof}
We first consider an alternative ``imaginary'' algorithm which continues even if $\flag$ becomes $\True$ (i.e., we go outside of radius-$\rho$ ball) and deterministically terminate after $T$ iterations, outputting $x_T,\tw_T$. For such an ``imaginary'' algorithm we have $\E[\inner{g_t-\hat{g}_t}{w_{t-1}-u}|w_{t-1}, x_{t-1}]=0$, thus by taking expectation on~\Cref{lem:mirror-descent-prob1},
\begin{align}\label{eqn:optimality-second}
\E h(x_T)-h(u)	&\le \E\bigg[\bigg(\lambda+\frac{1}{\eta T} \bigg)V_y(u)-\bigg(\frac{\lambda}{T}\sum_{t\in[T]}V_{w_t}(u)+\frac{1}{\eta T} V_{w_T}(u)\bigg)+\frac{\eta}{2}\GG^2-\lambda V_y(x_T)\bigg].
\end{align}

Standard tools from convex analysis imply that $\dgf^*$ (the dual function of $\dgf$), and its induced Bregman divergence $V^{\dgf^*}_a(a') = \dgf^*(a')-\dgf^*(a)-\langle\nabla \dgf^*(a), a'-a\rangle$ satisfy
\[ 
V_a(b) = V^{\dgf^*}_{\grad \dgf(b)}(\grad \dgf(a))
\]
for any $a,a'\in\xset^*$ \cite{rockafellar1997convex}. Now,
\begin{align*}
	\frac{\lambda}{T}\sum_{t\in[T]}V_{w_t}(u)+\frac{1}{\eta T} V_{w_T}(u) & \stackrel{(i)}{=} \frac{\lambda}{T}\sum_{t\in[T]}V^{\dgf^*}_{\nabla \dgf(u)}(\nabla \dgf(w_t))+\frac{1}{\eta T} V^{\dgf^*}_{\nabla \dgf(u)}(\nabla \dgf(w_T))\\
	& \stackrel{(ii)}{\ge} \left(\lambda+\frac{1}{\eta T}\right)\cdot V^{\dgf^*}_{\nabla \dgf(u)}\left(\frac{\frac{\lambda }{T}\sum_{t\in[T]}\nabla \dgf(w_t)+\frac{1}{\eta T}\nabla \dgf(w_T)}{\lambda+\frac{1}{\eta T}}\right)\\
	& \stackrel{(iii)}{=} \left(\lambda+\frac{1}{\eta T}\right)V^{\dgf^*}_{\nabla \dgf(u)}(\nabla \dgf(\tilde{w}_T))=\left(\lambda+\frac{1}{\eta T}\right) V_{\tilde{w}_T}(u).
\end{align*}
Here we use $(i)$ the equality $V_a(b) = V^{\dgf^*}_{\grad \dgf(b)}(\grad \dgf(a))$, $(ii)$ the convexity of $V_x(\cdot)$ and $(iii)$ the definition of $\tilde{w}_T$ as in~\cref{line:define-wtilde} of~\Cref{alg:mirror-descent}. Plugging this back into~\Cref{eqn:optimality-second} proves the expected guarantee for the ``imaginary'' algorithm:
\[
\E h(x_T) - h(u) \le \prn*{\lambda + \frac{1}{\eta T}}  \E \brk*{V_y(u) - V_{\tw_T}(u)} + \frac{\eta}{2}\GG^2 - \lambda \E  V_y(x_T).
\]

Further, applying strong convexity of $H_\lambda$, we have $H_\lambda(x_T)-H_\lambda(\uprox)\ge \lambda V_{\uprox}(x_T)$. Combining it with~\Cref{eqn:progress-realize} (where we choose $u = \uprox$), we have 
\[
\lambda V_{\uprox}(x_T)\le \frac{V_y(\uprox)}{\eta T}+\frac{\eta}{2}\GG^2+\frac{1}{T}\sum_{t\in[T]}\inner{g_t-\hat{g}_t}{w_{t-1}-u}.
\]
Taking expectation yields, 
\begin{align}\label{eqn:triangle-1}
\lambda\E V_{\uprox}(x_T)\le \frac{V_y(\uprox)}{\eta T}+\frac{\eta}{2}\GG^2.
\end{align}

By the $\tef$-triangle inequality, we also have 
\begin{align}\label{eqn:triangle-2}
\frac{\lambda}{\tau}V_y(\uprox)\le \lambda V_{y}(x_T)+\lambda V_{\uprox}(x_T).
\end{align}

Combining~\Cref{eqn:triangle-1,eqn:triangle-2} we have
\begin{align*}\label{eqn:triangle}
\frac{\lambda}{\tau} V_y(\uprox)\le \lambda \E V_{y}(x_T)+\frac{V_y(\uprox)}{\eta T}+\frac{\eta}{2}\GG^2 ~\implies~ -\lambda \E V_{y}(x_T)\le -\left(\frac{\lambda}{\tau}-\frac{1}{\eta T}\right) V_y(\uprox)+\frac{\eta}{2}\GG^2.  
\end{align*}

Plugging this back into~\Cref{eqn:progress-expect} proves the following guarantee for the output $x_T, \tw_T$ of the ``imaginary'' algorithm.
\begin{equation}\label{eqn:progress-expect-imaginary}
		\E h(x_T) - h(u) \le \prn*{\lambda + \frac{1}{\eta T}}  \E \brk*{V_y(u) - V_{\tw_T}(u)} + \eta\GG^2 - \prn*{\frac{\lambda}{\tef} - \frac{1}{\eta T}} \E  V_y(\uprox).
\end{equation}

Now considering the original algorithm, the iterates will behave exactly the same when $\flag=\False$ for all iterations. When $\flag=\True$ the actual algorithm returns an arbitrary point $x_t$ incurs a loss bounded by $h(x_{t})-h(x_T)+(\lambda+\frac{1}{\eta T})V_{x_t}(u) \le \sqrt{2}R\GG + (\lambda+\frac{1}{\eta T})R^2$. Thus, we have for the claimed bound for the actual algorithm's output iterates $z,w$, i.e.,
\begin{align*}
\E h(z) - h(u) \le & \prn*{\lambda + \frac{1}{\eta T}}  \E \brk*{V_y(u) - V_{w}(u)} + \eta\GG^2-\prn*{\frac{\lambda}{\tef} - \frac{1}{\eta T}} \E  V_y(\uprox) - \lambda \E  V_y(x_T)\\
& \hspace{3em} +\left(\sqrt{2}R\GG+\prn*{\lambda+\frac{1}{\eta T}}R^2\right)\Pr(\flag=\True).
\end{align*}
\end{proof}

Next, we bound the term $\sum_{t=1}^{T}\inner{g_{t}-\hat{g_{t}}}{w_{t-1}-u}$ on the RHS of~\Cref{eqn:progress-realize} using concentration of measure. This is formally stated in the next lemma; we defer its proof to the end of this subsection.

\begin{lemma}\label{lem:helper-concentration}
	In the setting of \Cref{lem:mirror-descent-prob1}, for any $\delta,\veps\in(0,1)$ and $u\in\xset$, we have 
	\begin{equation}\label{eq:concentration-event}
		\Pr*(\cE(\delta) \defeq \crl*{\max_{1\le t\le T}\abs*{\sum_{i=1}^{t}\inner{\hg_i-g_{i}}{w_{i-1}-u}}\le \GG\max_{0\le i < T}\norm*{ w_{i}-u} \sqrt{32T\log\frac{2}{\delta}}}) \ge 1 -\delta.\end{equation}
\end{lemma}

Combining \Cref{eqn:progress-realize} in \Cref{lem:mirror-descent-prob1} with the concentration guarantees in~\Cref{lem:helper-concentration}, we show the iteration $\{w_t\}_{t\in[T]}$ and $x_T$ stay relatively close to the true optimizer $\uprox$ in the following.

\begin{lemma}\label{lem:mirror-descent-highprob}
	In the setting of \Cref{lem:mirror-descent-prob1} and~\Cref{lem:helper-concentration}, let $\uprox \defeq \argmin_{x\in\xset} \Hprox(x)$. For any $\delta\in(0,1)$ and $T\ge 1$, when event $\cE(\delta)$ happens, 
	\begin{equation*}
		\max_{0\le t\le T} V_{w_t}(\uprox) \le 2 V_y(\uprox) + \prn*{65 \log\frac{2}{\delta}} \eta^2 \GG^2 T
	\end{equation*}
	and
	\begin{equation*}
		 \lambda V_{\uprox}\left(x_{T}\right)\le\frac{2V_{y}(\uprox)}{\eta T}+\left(66\log\frac{2}{\delta}\right)\eta \GG^{2}.
	\end{equation*}
\end{lemma}

\begin{proof}
For the first inequality, we follow~\Cref{eqn:progress-realize}, due to $H_\lambda(x_T)-H_\lambda(\uprox)\ge0$ and the non-negativity of Bregman divergences, we have\[
V_{w_T}(\uprox)\le V_y(\uprox)+\frac{\eta^2}{2}\GG^2T+\eta\sum_{t\in[T]}\inner{g_t-\hat{g}_t}{w_{t-1}-\uprox}.
\]
Applying the same argument for all $t\in[T]$ gives
\[
V_{w_t}(\uprox)\le V_y(\uprox)+\frac{\eta^2}{2}\GG^2t+\eta\sum_{i\in[t]}\inner{g_i-\hat{g}_i}{w_{i-1}-\uprox},~~\text{for all}~t\in[T].
\]
Applying~\Cref{lem:helper-concentration} with $u=\uprox$, we have under the event $\cE(\delta)$,
\begin{equation}\label{eqn:progress-max}
\begin{aligned}
\max_{t\in[T]}V_{w_t}(\uprox) & \le V_y(\uprox)+\frac{\eta^2}{2}\GG^2T+\eta\max_{t\in[T]}\left|\sum_{i\in[t]}\inner{g_i-\hat{g}_i}{w_{i-1}-\uprox}\right|\\
& \le V_y(\uprox)+\frac{\eta^2}{2}\GG^2T+\eta\GG\max_{0\le t\le T}\|w_t-\uprox\|\sqrt{8T\log\frac{2}{\delta}}\\
& \stackrel{(i)}{\le} V_y(\uprox)+\frac{\eta^2}{2}\GG^2T+\eta^2\GG^2\cdot(32T\log\frac{2}{\delta})+\max_{0\le t\le T}\frac{1}{4}\|w_t-\uprox\|^2\\
& \stackrel{(ii)}{\le} V_y(\uprox)+\left(\frac{65}{2}\log\frac{2}{\delta}\right)\eta^2\GG^2T+\max_{0\le t\le T}\frac{1}{2}V_{w_t}(\uprox).\end{aligned}
\end{equation}
Here we use $(i)$ the AM-GM inequality and $(ii)$ the strong convexity of Bregman divergence by definition. Note the RHS in~\Cref{eqn:progress-max} also upper bounds $V_{w_0}(\uprox)$ since $w_0=y$ in the initialization of~\Cref{alg:mirror-descent}. Combining these together and rearranging terms, 
\begin{align*}
\max_{0\le t\le T}V_{w_t}(\uprox) & \le 2V_y(\uprox)+\left(65\log\frac{2}{\delta}\right)\eta^2\GG^2T,\end{align*}
thus proving the first inequality.

For the second inequality, we note by strong convexity, $H_\lambda(x_T)-H_\lambda(\uprox)\ge \lambda V_{\uprox}(x_T)$, plugging this back into~\Cref{eqn:progress-realize} and again using non-negativity of Bregman divergences and similar arguments following~\Cref{lem:helper-concentration}, we have when event $\cE(\delta)$ happens,
\begin{align*}
\lambda V_{\uprox}(x_T) & \le \frac{V_{y}(\uprox)}{\eta T}+\frac{\eta}{2}\GG^2+\frac{1}{T}\sum_{t\in[T]}\inner{g_t-\hat{g}_t}{w_{t-1}-\uprox}\\
& \le \frac{V_{y}(\uprox)}{\eta T}+\frac{\eta}{2}\GG^2+\frac{1}{T}\GG\max_{0\le i\le T}\|w_i-\uprox\|\sqrt{32T\log\frac{2}{\delta}}\\
& \stackrel{(i)}{\le} \frac{V_{y}(\uprox)}{\eta T}+\frac{\eta}{2}\GG^2+32\left(\log\frac{2}{\delta}\right)\eta\GG^2+\frac{1}{2\eta T}\max_{0\le i\le T}V_{w_i}(\uprox)\\
& \stackrel{(ii)}{\le} \frac{2V_{y}(\uprox)}{\eta T}+65\left(\log\frac{2}{\delta}\right)\eta\GG^2.\end{align*}
Here we use the Cauchy-Schwarz inequality for $(i)$ and the first inequality proven for $(ii)$. This concludes the proof for the second inequality.
\end{proof}

We use \Cref{lem:mirror-descent-highprob} to control the possibility of $\mprox$ going out of bounds when $\uprox$ is not too far from the center point $y$.
\begin{lemma}\label{lem:not-out-of-bounds}
	In the setting of \Cref{lem:mirror-descent-highprob}, if $V_{y}(\uprox) \le \frac{\rho^2}{16}$ and for some $\delta \in (0,1)$ we have $\prn*{\log \frac{2}{\delta}}\eta^2\GG^2 T \le \frac{\rho^2}{65\cdot 16}$ then the event $\mc{E}(\delta)$ implies that $\flag=\False$. 
\end{lemma}
\begin{proof}
	We have $\flag=\False$ if and only if $\max_{t\le T}\norm{x_t - y} \le \rho$. To derive a sufficient condition for this inequality we upper bound $\max_{t\le T}\norm{x_t - y}$ as follows: 
	\begin{equation*}
		\max_{t\le T}\norm{x_t - y} \overle{(i)} \max_{t< T}\norm{w_t - y}
		\le \norm{\uprox - y} + \max_{t< T}\norm{w_t - \uprox}
		\overle{(ii)} \sqrt{2V_y(\uprox)} + \sqrt{2\max_{t<T} V_{w_t}(\uprox)},
	\end{equation*}
	where $(i)$ follows by convexity and the definition of $x_t$ as the averaging of $w_0, \ldots, w_{t-1}$, and $(ii)$ follows from the 1-strong-convexity of the distance generating function. 

	Next, we apply \Cref{lem:mirror-descent-highprob} and the assumption $\prn*{\log \frac{2}{\delta}}\eta^2\GG^2 T \le \frac{\rho^2}{65\cdot 16}$ to obtain that $\mc{E}(\delta)$ implies
	\begin{equation*}
		\max_{t<T} V_{w_t}(\uprox) \le 2V_y(\uprox) + \frac{\rho^2}{16}.
	\end{equation*}
	Substituting $V_y(\uprox)\le \frac{\rho^2}{16}$ and combining the above displays yields
	\begin{equation*}
		\max_{t\le T}\,\norm{x_t - y} \le \frac{\rho}{\sqrt{8}} + \sqrt{\frac{\rho^2}{4} + \frac{\rho^2}{8}} \le \rho
	\end{equation*}
	as required.
\end{proof}

Finally, the next lemma bounds the total movement of iterations $\{x_t\}_{t\in[T]}$:
\begin{lemma}\label{lem:movement}
	In the setting of \Cref{lem:mirror-descent-prob1}, we have, for any $u\in\xset$, 
	\begin{equation}
		\sum_{t=1}^T \norm{x_t - x_{t-1}} \le 
		2(\log T+1) \max_{0\le t \le T} \norm{w_t - u}.
	\end{equation}
\end{lemma}
\begin{proof}
By definition of $x_t$, we have $x_{t}-x_{t-1} = \frac{1}{t}(w_{t-1}-x_{t-1})$, consequently by triangle inequality we have
\begin{align*}
	\sum_{t=1}^T \norm{x_t - x_{t-1}} =\sum_{t\in[T]}\frac{1}{t}\norm{w_{t-1}-x_{t-1}} \le \sum_{t\in[T]}\frac{1}{t}\norm{w_{t-1}-u}+\sum_{t\in[T-1]}\frac{1}{t}\norm{x_{t-1}-u}.
\end{align*}
We proceed to bound the two terms on the RHS respectively. For the first term, 
\begin{align*}
	\sum_{t\in[T]}\frac{1}{t}\norm{w_{t-1}-u}\le \prn*{\sum_{t\in[T]}\frac{1}{t}} \max_{0\le t\le T-1}\|w_t-u\|\le (\log T+1) \max_{0\le t \le T} \norm{w_t - u}.
\end{align*}
For the second term, 
\begin{align*}
	\sum_{t\in[T]}\frac{1}{t}\norm{x_{t-1}-u}\le \prn*{\sum_{t\in[T]}\frac{1}{t}} \max_{0\le t\le T-1}\|x_t-u\|\stackrel{(\star)}{\le} (\log T+1) \max_{0\le t \le T} \norm{w_t - u},
\end{align*}
where we also use convexity of the norm function $\|\cdot\|$ and the fact that $x_{t-1}=\frac{1}{t-1}\sum_{i=0}^{t-2}w_{i}$ for $(\star)$.
Summing the two terms proves the claimed bound.
\end{proof}

\begin{proof}[Proof of~\Cref{lem:helper-concentration}]
We consider the random variable $X_i \defeq \frac{1}{2\GG\max_{0\le j\le i-1}\|w_{j}-u\|}\inner{g_i-\hat{g}_i}{w_{i-1}-u}$ and the filtration $\calF_{i-1} \defeq \sigma(x_0,w_0, x_1, w_1,\cdots, w_{i-1},x_{i})$. Note we have $\E[X_i|\calF_{i-1}]=0$ and additionally $|X_i|\le \frac{\|g_i-\hat{g}_i\|_*}{2\GG}\le 1$ with probability $1$. Thus, applying Blackwell's inequality~(cf. \citet{blackwell1997large} Theorem 1), we have for any $a,b>0$,
\begin{align*}
	\mathbb{P}\left(\exists~t\in[T],\bigg|\sum_{i\in[t]}X_i\bigg|\le a+bt\right)\le 2e^{-2ab}.
\end{align*}

Replacing $a=\sqrt{T\log(2/\delta)/2}$, $b=\sqrt{\log(2/\delta)/2T}$, with probability $1-\delta$, we have for all $t\in[T]$,
\begin{align*}
\bigg|\sum_{i\in[t]}X_i\bigg|\le \sqrt{T\log(2/\delta)/2}+\sqrt{\log(2/\delta)/2T}\cdot t \le  \sqrt{2T\log(2/\delta)}.
\end{align*}
Now applying Lemma 5 of~\citet{ivgi2023dog} with $a_i = 2\GG\max_{0\le j\le i-1}\|w_j-u\|$ and $b_i = X_i$, we have 
\begin{align*}
	\bigg|\sum_{i\in[t]}\inner{g_i-\hat{g}_i}{w_{i-1}-u}\bigg|& \le4\GG\max_{0\le i\le t-1}\|w_i-u\|\cdot \max_{1\le i\le t}\bigg|
	\sum_{j\in[i]}X_j\bigg|\\
	& \le \GG\max_{0\le i\le T-1}\|w_{i}-u\|\sqrt{32T\log\frac{2}{\delta}}~~\text{for all}~t\in[T].
\end{align*}
Taking maximum over all $t\in[T]$ gives the desired claim.
\end{proof}

\subsection{Analysis of $\linesearch$}\label{ssec:bisection}
In this section, we prove the correctness and bound the number of iterations for $\linesearch$ in~\Cref{alg:ball-oralce}. We use $\cE_k(\delta)$ to denote the probabilistic event described in~\Cref{lem:helper-concentration} with parameter $\delta$ when calling~$\mprox(\gest,y,\rho,\lambda_k,\eta_k,T_k)$, which according to that lemma happens with probability at least $1-\delta$. In the next lemma, we first show that if the stopping criterion of the binary search holds for some $k$ and $\lambda_k$, then with high probability the value of $V_y\big(\uprox[\lambda_k]\big)$  is $\Theta(\rho^2/\mathrm{poly}(\tau))$. 

\begin{lemma}\label{lem:stopping-criteria}
	Assume $\xset$ and $V$ satisfy a $\tau$-triangle inequality. 
	For $\delta_k\in(0,1)$, under the event $\cE_k(\delta_k/8)$, at iteration $k$ of~\Cref{alg:ball-oralce} the call to $\mprox$ outputs $z^{(k)}$ such that if 
	$V_y(z^{(k)}) \le \frac{\rho^2}{64\tau}$ then $V_y(\uprox[\lambda_k]) \le \frac{\rho^2}{16}$ and if $V_y(z^{(k)}) \ge \frac{\rho^2}{256\tau^3}$ then $V_y(\uprox[\lambda_k]) \ge \frac{\rho^2}{1024\tau^4}$.
\end{lemma}
\begin{proof}
	We begin by noting that $V_y(z\pind{k}) \le \frac{\rho^2}{64\tau}$ implies that $\norm{z\pind{k} -y} < \rho$ and therefore that $\flag\pind{k} = \False$ and $z\pind{k} = x\pind{k}_{T_k}$, i.e., the last iterate of $\mprox$. This allows us to apply \Cref{lem:mirror-descent-highprob} to bound, in the event $\cE_k(\delta_k/8)$,
	\[
	V_{\uprox[\lambda_k]}(z^{(k)})\le \frac{2V_y(\uprox[\lambda_k])}{\lambda_k\eta_kT_k}+\left(66\log\frac{16}{\delta_k}\right)\frac{\eta_k}{\lambda_k}\GG^2\le \frac{1}{2\tau}V_y(\uprox[\lambda_k])+\frac{\rho^2}{1024\tau^4}.
	\]

	To upper bound $V_y(\uprox[\lambda_k])$ we use the $\tef$-triangle inequality, $V_y(z^{(k)}) \le \frac{\rho^2}{64\tau}$ and the bound on $V_{\uprox[\lambda_k]}(z^{(k)})$ to write
	\[
	V_y(\uprox[\lambda_k])\le \tau\left(V_y(z^{(k)})+V_{\uprox[\lambda_k]}(z^{(k)})\right) \le \frac{\rho^2}{36} + \frac{1}{2} V_y(\uprox[\lambda_k]) + \frac{\rho^2}{1024}.
	\]
	Rearranging yields $V_y(\uprox[\lambda_k])\le \frac{\rho^2}{16}$ as required.
	
	To lower bound $V_y(\uprox[\lambda_k])$ we combine the $\tau$-triangle with the assumed lower bound on $V_y(z^{(k)})$, 
	\begin{flalign*}
	V_y(\uprox[\lambda_k])
	&\ge \frac{1}{\tau}V_y(z^{(k)})-V_{\uprox[\lambda_k]}(z^{(k)})
	\ge \frac{1}{\tau}V_y(z^{(k)})-\frac{\rho^2}{1024\tau^4}-\frac{1}{2\tau} V_y(\uprox[\lambda_k]) \ge \frac{\rho^2}{512\tau^4} - V_y(\uprox[\lambda_k])
	\\ &  \implies
	V_y(\uprox[\lambda_k])
	\ge \frac{\rho^2}{1024\tau^4}.
	\end{flalign*}
\end{proof}

The next lemma shows that there exists a nontrivial range of $\lambda$ values for which the binary search will terminate with high probability. Here by overloading notations we let $\cE_\lambda(\delta)$ to denote the probablistic event in~\Cref{lem:helper-concentration} with parameter $\delta$ when calling~$\mprox(\gest,y,\rho,\lambda,\eta,T)$ with $\eta$ and $T$ chosen as in~$\linesearch$.
 
\begin{lemma}\label{lem:stopping-criteria-2}
	Assume $\xset$ and $V$ satisfy a $\tau$-triangle inequality with $\tau\ge 4$. For  $\delta\in(0,1)$ let $\eta\le \frac{\rho^2\lambda}{66\cdot 1024\cdot\log(16/\delta)\tau^5\GG^2}$ and   $T=\frac{4\tau}{\eta\lambda}$. Then under event $\cE_\lambda(\delta/8)$ 
	the output $z$ of $\mprox(\gest, y, \rho, \lambda,\eta,T)$ satisfies if $V_y(\uprox) \le \frac{\rho^2}{100\tau^2}$ then $V_y(z)\le \frac{\rho^2}{64\tau}$ and if $V_y(\uprox) \ge \frac{\rho^2}{120\tau^2}$ then $V_y(z)\ge \frac{\rho^2}{256\tau^3}$.
\end{lemma}
\begin{proof}
	We begin by noting that by \Cref{lem:not-out-of-bounds} the assumption $V_y(\uprox) \le \frac{\rho^2}{48\tau^2}$, the event $\mc{E}_{\lambda}(\delta/8)$, and the choice of $\eta$ implies that $\flag=\False$. Therefore, as in the proof of \Cref{lem:stopping-criteria} above, we may use \Cref{lem:mirror-descent-highprob} and conclude that
	\[
	V_{\uprox}(z)\le \frac{2V_y(\uprox)}{\lambda\eta T}+\left(66\log\frac{2}{\delta}\right)\frac{\eta}{\lambda}\GG^2\le \frac{1}{2\tau}V_y(\uprox)+\frac{\rho^2}{1024\tau^4}.
	\]

By the $\tau$-triangle inequality,
\[
	V_y(z)\le \tau\left(V_y(\uprox)+V_{\uprox}(z)\right)\le \frac{3}{2}\tau V_y(\uprox)+\frac{\rho^2}{1024\tau^3}
	\le \rho^2 \prn*{\frac{3\tau}{2} \cdot \frac{1}{100\tau^2} + \frac{1}{1024\tau^4}} 
	\le 
	\frac{\rho^2}{64\tau}.
	\]
Applying the $\tau$-triangle inequality in the other direction gives
\[
V_y(z)\ge \frac{1}{\tau}V_y(\uprox)-V_{\uprox}(z)\ge \frac{1}{2\tau} V_y(\uprox)-\frac{\rho^2}{1024\tau^4}
\ge \rho^2 \prn*{\frac{1}{2\tau} \cdot \frac{1}{120\tau^2} -\frac{1}{1024\cdot 4\tau^3}} 
\ge \frac{\rho^2}{256\tau^3}.
\]
\end{proof}

The next lemma justifies the choice of the upper bisection limit $\lambda_{\max}$.

\begin{lemma}[Upper bisection limit]\label{lem:line-search-upper}
	Let $h:\xset\to\R$ be convex and, for some $y\in\xset$, let $\Hprox(x) \defeq h(x) + \lambda V_y(x)$ with 1-strongly-convex $V_y(\cdot)$ and $\uprox\defeq \argmin_{x\in\xset} \Hprox(x)$. If $h$ is $\GG$-Lipschitz then for any $\lambda \ge 0$, 
	\begin{equation}\label{eq:bregman-abs-stability}
		V_y(\uprox[\lambda]) \le \frac{\GG^2}{2\lambda^2}.
	\end{equation}
	Consequently, for $\lambda_{\max} = \frac{16\tau\GG}{\rho}$ we have $V_y(\uprox[\lambda_{\max}]) <  \frac{\rho^2}{100\tau^2}$.
\end{lemma}
\begin{proof}
	We may assume that $\uprox$ is in the interior of $\xset$, since otherwise $V_y(\uprox) = V_y(\uprox[\lambda'])$ for some $\lambda' \ge \lambda$ such that for all $\lambda''>\lambda'$ the point $\uprox[\lambda'']$ is in the interior of $\xset$, and we may apply the following considerations to $\lambda''\downarrow\lambda'$ instead. We further assume without loss of generality that $h$ and $\dgf$ are differentiable, as otherwise we may unifromly approximate them with convex differentiable functions via Moreau envelopes. 

	These assumptions imply that 
	\[
		0 = \grad \Hprox(\uprox) = \grad h(\uprox) + \lambda \grad V_y(\uprox).
	\] 
	Hence, the fact that $h$ is $\GG$-Lipschitz implies that
	\[
		\norm*{\grad V_y(\uprox)}_* = \frac{1}{\lambda} \norm*{ \grad h(\uprox) }_* \le \frac{\GG}{\lambda}.
	\]
	Finally, the 1-strong-convexity of $x\mapsto V_y(x)$ and the fact that its minimal value of 0 is obtained at $y$ implies that 
	\[
		V_y(\uprox) = V_y(\uprox) - V_y(y) \le \half \norm*{\grad V_y(\uprox)}_*^2  \le \frac{\GG^2}{2\lambda^2}
	\]
	as required.
\end{proof}

The next lemma justifies the lower bisection limit $\lambda_{\min}$.

\begin{lemma}[Lower bisection limit]\label{lem:line-search-lower}
	Let $\lambda_{\min}=\lambda_0 =1$ and $\uprox[\lambda_{\min}] \defeq \argmin_{x\in\xset} \Hprox[\lambda_{\min}](x)$ and assume that $\xset$ and $V$ satisfy a  $\tau$-triangle inequality with $\tau\ge 4$. 
	Under the event $\cE_0(\delta/8)$, if $V_y(z\pind{0}) \le \frac{\rho^2}{64\tau}$ then $V_y(\uprox[\lambda_{\min}]) \le \frac{\rho^2}{16}$ and if $V_y(z\pind{0}) \ge \frac{\rho^2}{64\tau}$ then $V_y(\uprox[\lambda_{\min}]) \ge \frac{\rho^2}{120\tau^2}$.
\end{lemma}
\begin{proof}
Immediate from~\Cref{lem:stopping-criteria,lem:stopping-criteria-2}.
\end{proof}

Finally, we bound the Lipschitz constant of $\lambda \mapsto V_y(\uprox)$ and apply the above lemmas to conclude that $\linesearch$ returns a valid points within $\Otil{1}$ iterations.

\begin{proposition}\label{prop:binary-search-iterations}
In the setting of \Cref{thm:oracle-impl}, under the event $\cap_{k=0}^{\Kmax}\cE_k(\delta_k/8)$, with $\Kmax = \lceil\log_2\frac{9600\tau^2\GG^3}{\rho^3}\rceil+1$, which happens with probabiltiy at least $1-\frac{\delta}{2}$, the $\linesearch$ procedure in~\Cref{alg:ball-oralce} successfully returns at iteration $K< \Kmax$ a value $\lambda_K$ such that $V_y(\uprox[\lambda_K]) \le \frac{\rho^2}{16}$ and, if $K\ge 1$, also $V_y(\uprox[\lambda_K]) \ge \frac{\rho^2}{1024\tau^4}$.
\end{proposition}
\begin{proof}
We begin by noting that $\Pr*(\cap_{k=0}^{\Kmax}\cE_k(\delta_k/8)) \ge 1-\frac{\delta}{2}$ by \Cref{lem:helper-concentration} and the union bound.

Next, \Cref{lem:line-search-lower} establishes the claims of the proposition in the edge case we return with $K=0$.

Moving on to the main case we return with $K\ge 1$. If also $K<\Kmax$ then \Cref{lem:stopping-criteria} guarantees that the claim $V_y(\uprox[\lambda_K]) \in \brk[\big]{\frac{\rho^2}{1024\tau^4},\frac{\rho^2}{16}}$ holds. It therefore remains to argue that the bisection does indeed terminate in less than $\Kmax$ steps. 
Let $\lambda',\lambda''\in (\lambda_{\min},\lambda_{\max}]$ satisfy $V_y(\uprox[\lambda']) = \frac{\rho^2}{100\tau^2}$ and $V_y(\uprox[\lambda'']) = \frac{\rho^2}{120\tau^2}$. By \Cref{lem:stopping-criteria-2,lem:line-search-upper,lem:line-search-lower}, when $\cap_{k=0}^{\Kmax}\cE_k(\delta_k/8)$ holds then $[\lambda',\lambda''] \subseteq [\lambda_{\min}, \lambda_{\max}]$ is an invariant of the bisection and moreover the bisection terminates if we query $\lambda_K \in [\lambda',\lambda'']$. Since the bisection the search interval at every step, it must return in $\log_2\frac{\lambda_{\max}-\lambda_{min}}{\lambda''-\lambda'}$ steps. We have $\lambda_{\max}-\lambda_{\min} \le \frac{16\tau\Gamma}{\rho}$, so to conclude the proof we need only lower bound $\lambda''-\lambda'$.

To do so, we write
\begin{align*}
\frac{\rho^2}{600\tau^2} = V_y(\uprox[\lambda'])-V_y(\uprox[\lambda'']) & 
= 
\int_{\lambda = \lambda''}^{\lambda'} (V_y(\uprox))' d\lambda
=
\int_{\lambda = \lambda''}^{\lambda'}\inner{\grad V_y(\uprox)}{\nabla_\lambda \uprox}d\lambda\\
& \stackrel{(i)}{=} -\int_{\lambda = \lambda''}^{\lambda'}\inner{\grad V_y(\uprox)}{\prn*{\nabla^2 h(\uprox)+\lambda\nabla^2 V_y(\uprox)}^{-1}\nabla V_y(\uprox)}d\lambda\\
& \stackrel{(ii)}{\le} (\lambda''-\lambda')\frac{\GG^2}{(\lambda')^3}\le (\lambda''-\lambda')\GG^2.
\end{align*}
Here for $(i)$ we use 
$\nabla h(\uprox)+\lambda \nabla V_y(\uprox)=0$ for all $\lambda\in [\lambda',\lambda'']$, which implies $\grad_\lambda\uprox = -(\nabla^2 h(\uprox)+\nabla^2 V_y(\uprox))^{-1}\nabla V_y(\uprox)$ by taking derivatives with respect to $\lambda$ and rearranging terms (we assume here that $h$ and $r$ are twice differentiable; this is again without loss of generality due to smoothing arguments).  For $(ii)$ we reuse $\norm{\grad V_y(\uprox)} \le \frac{\GG}{\lambda}$ from the proof of~\Cref{lem:line-search-upper}.
The above display implies that $\lambda''-\lambda'\ge \frac{\rho^2}{600\tau^2\GG^2}$ and therefore our choice of $\Kmax$ guarantees that $\log_2\frac{\lambda_{\max}-\lambda_{min}}{\lambda''-\lambda'} < \Kmax$, concluding the proof.
\end{proof}

\subsection{Proof of~\Cref{thm:oracle-impl}}\label{ssec:oracle-impl}

\begin{proof}
We prove each part of the proposition in turn.

For part $1$, let $K\le \Kmax$ be the final iteration of $\linesearch$ and let $\lambda_K$ be its output. Recall from~\Cref{lem:mirror-descent} that, for  $\lambda = \lambda_K$, input parameters $\eta = \frac{\rho^2\lambda}{\bigC\cdot\log(16/\delta)\tau^5\GG^2}$, $T = \frac{4\tau}{\eta\lambda}$, the outputs of $\mprox$ satisfy
\begin{align*}
	\E_{\lambda} h(z) - h(u)  \le  &\prn*{\lambda + \frac{1}{\eta T}}  {\lambda} \brk*{V_y(u) - V_{w}(u)} + \eta\GG^2 - \prn*{\frac{\lambda}{\tef} - \frac{1}{\eta T}}  V_y(\uprox)\\
& +\prn*{\sqrt{2}R\GG+ \prn*{\lambda+\frac{1}{\eta T}} R^2}\P_{\lambda}(\flag=\True),
\end{align*}
where $\E_{\lambda}$ and $\P_{\lambda}$ denote conditional expectation over random variable $\lambda=\lambda_K$. 

Dividing both sides by $c = \lambda+(\eta T)^{-1}\ge 1$ and taking total expectation we have
\begin{align*}
&
\E \frac{h(z) - h(u)}{c} 
\\ & \quad \le  \E \brk*{V_y(u) - V_{w}(u)} + \E\brk*{ \frac{\eta}{\lambda} \GG^2 - \frac{\frac{\lambda}{\tef} - \frac{\lambda}{4\tef}}{\lambda + \frac{\lambda}{4\tef}}   V_y(\uprox)+\prn*{\frac{\sqrt{2}R\GG}{\lambda + (\eta T)^{-1}}+R^2}\P_\lambda(\flag=\True)}	
\\& \quad \le  \E \brk*{V_y(u) - V_{w}(u)} + \frac{\rho^2}{\bigC\log(16/\delta)\tau^5} - \frac{3}{5\tau} \E V_y(\uprox)+(\sqrt{2}R\GG+R^2)\Pr(\flag=\True). \label{eq:expectation-bound-intermediate}\numberthis
\end{align*}

\Cref{prop:binary-search-iterations} implies that $V_y(\uprox) \ge \frac{\rho^2}{1024\tau^4}\indic{\lambda \ne \lambda_{\min}}$ holds with probability at least $1-\frac{\delta}{2}$. Therefore, since Bregman divergences are nonnegative, 
\[
	\E V_y(\uprox) \ge \prn*{1-\frac{\delta}{2}} \frac{\rho^2}{1024\tau^4} \indic{\lambda \ne \lambda_{\min}}\ge \frac{\rho^2}{2^{11} \tau^4}\indic{\lambda \ne \lambda_{\min}}.
\]
By our choices of $\eta$ and $T$, \Cref{lem:not-out-of-bounds} tells us that $V_y(\uprox) \le \frac{\rho^2}{16}$ and $\mc{E}_\lambda(\delta/8)$ imply that $\flag = \False$. Therefore,
\begin{equation*}
	\P_{\lambda}(\flag=\True) \le \P_{\lambda}(\neg\mc{E}_\lambda(\delta/8)) + \indic{V_y(\uprox) > \frac{\rho^2}{16}} \le \frac{\delta}{8} + \indic{V_y(\uprox) > \frac{\rho^2}{16}},
\end{equation*}
where the final inequality used \Cref{lem:helper-concentration}. Taking expectation and invoking \Cref{prop:binary-search-iterations} again, we find that
\begin{equation*}
	\P(\flag=\True) \le \frac{\delta}{8} + \Pr*(V_y(\uprox) > \frac{\rho^2}{16})
	\le \frac{\delta}{8} + \frac{\delta}{2} \le \delta.
\end{equation*}	
Substituting the bounds on $\E V_y(\uprox)$ and $\P(\flag=\True)$ into \Cref{eq:expectation-bound-intermediate}, noting that $c\ge 2$ only when $\lambda\ne \lambda_{\min}$ and recalling that $\delta \le  \frac{\rho^2}{2^{14}(\sqrt{2}R\GG+R^2)\tau^5}$,  we obtain the required ball-restricted proximal oracle bound~\eqref{eq:oracle-impl-error-bound}. Additionally, we note for all possible choices of returned $c$ it satisfies $c \le 2\lambda_{\max} = \frac{32\tau\GG}{\rho}$ with probability 1, giving the claimed value of $\cmax$.

Part 2 of the proposition is immediate from the definition of $\mprox$, which always outputs points with distane at most $\rho$ from $y$. 

For part 3 we use~\Cref{lem:movement} with $u=\uprox$, which gives for all $k\le K$ that  $\sum_{t\in[\hat{T}_k]}\|x_t^{(k)}-x_{t-1}^{(k)}\|\le 2\log(2T_k)\rho$. Summing these bounds gives
 \[
	\sum_{k=0}^K\sum_{t\in[\hat{T}_k]}\|x_t^{(k)}-x_{t-1}^{(k)}\|\le 2\sum_{k=0}^K\log(2\hat{T}_k)\rho \le  2\rho \Kmax\log\frac{4\bigC\log(16\Kmax^2/\delta)\tau^6\GG^2}{\rho^2}.
\]

Finally, part 4 follows from the setting of the $T_k$ and $\Kmax$, since the total number of gradient queries and mirror descent steps is at most $\left(\sum_{k=0}^{\Kmax} T_k\right)$. 
\end{proof}

\newcommand{\timeInit}{\runtime_{\dsStyle{init}}}
\newcommand{\timeQuery}{\runtime_{\dsStyle{query}}}
\newcommand{\samp}{\codeStyle{sample}}
\newcommand{\MVMp}{\mathrm{MVM}_p}
\newcommand{\MVEp}{\mathrm{MVE}_p}
\newcommand{\MVEone}{\mathrm{MVE}_1}
\newcommand{\MVEtwo}{\mathrm{MVE}_2}

\section{Matrix-vector maintenance data structures}
\label{sec:data_structure}
In this section we formally define an \emph{$\ell_p$-matrix-vector maintenance data structures (abbreviated $\MVMp$)} and provide efficient algorithms for them for $p \in \{1, 2\}$. 
An $\MVMp$ approximates the sequence $\{A x_t\}$ to additive $\epsilon$ error in $\ell_\infty$, as long as the sum of the $\ell_p$ norm of the movements $\Delta_t=x_{t+1}-x_t$ does not exceed a given bound $R$. The data structure is formally defined below in \Cref{def:mv_maintenance}; for a  brief description of these data structures and how they fit into our overall method, see \Cref{subsec:overview-ds}. In the definition of an $\MVMp$, and throughout this section, for any $p \geq 1$ we let $p^* \geq 1$ be such that $\frac{1}{p} + \frac{1}{p^*} = 1$; if $p = 1$ then $p^* = \infty$. Furthermore, for any matrix $A \in \R^{n \times d}$ with rows $a_1, \ldots, a_n \in \R^d$ and $p \geq 1$ we let
\[
\norm{A}_{p \rightarrow \infty}
\defeq \sup_{x \in \R^{n}, \norm{x}_p = 1} \norm{A x}_\infty
= \max_{i \in [n]} \norm{a_i}_{p^*}\,.
\]

\begin{definition}[Matrix-vector maintenance]
	\label{def:mv_maintenance}
	We call a data structure an \emph{$\ell_p$-matrix-vector maintenance data structure ($\MVMp$)} if it supports the following operations: 
	\begin{itemize}
		\item $\dsInit(A \in \R^{n \times d}, x_0 \in \R^d, R \in \R_{> 0}, \epsilon \in \R_{> 0})$: initializes the data structure with a matrix $A$ with $\norm{A}_{p \rightarrow \infty} \leq 1$, initial point $x_0$, movement range $R$, and accuracy $\epsilon \leq R/2$.\footnote{This can always be obtained by initializing the algorithm with a smaller value of $\epsilon$ or a larger value of $R$. 
		} Sets $t \gets 0$.
		\item $\dsQuery(\Delta_t \in \R^d)$: sets $x_{t +1} \gets x_{t} + \Delta_t$, and $t \gets t +1$ and then outputs $y_t \in \R^n$ (or the coordinates which changed from the previous output if that is cheaper) with $\norm{y_t - A x_t}_\infty \leq \epsilon$ provided that $\sum_{i \in [t]} \norm{\Delta_i}_p \leq R$. 
	\end{itemize}
\end{definition}

Our main results for designing $\MVMp$'s are encapsulated in the following theorem.

\begin{theorem}[Matrix-vector maintenance]
	\label{thm:matvec_maintenance}	
	For both $p=1$ and $p=2$ and any $\delta > 0$, there is a $\MVMp$ (\Cref{def:mv_maintenance}) that implements $\dsInit$ and $T$ $\dsQuery$ operations with probability $1 - \delta$ (against an oblivious adversary) in total time
	\[
	O\left(
	\sum_{t \in [T]} \nnz(\Delta_t) 
	+
	\left(
	\nnz(A) \log^{p - 1}\left(\frac{R}{\epsilon}\right) + d \cdot \frac{R}{\epsilon} 
	\right)
	\log^{p -1} \left( \frac{n R}{\epsilon \delta}\right)
	+
	n \left(\frac{R}{\epsilon}\right)^2 \log \left(\frac{n R}{\epsilon \delta}\right) 
	\right).
	\]
\end{theorem}

The runtime in \Cref{thm:matvec_maintenance} is nearly linear in input size $\nnz(A) + \sum_{t\in[T]}\nnz(\Delta_t)$ with an additive $\otilde(d (R / \epsilon))$ and  $\otilde(n(R/\epsilon)^2)$ terms.
When $A$ is dense and $R$ does not depend on $T$, this runtime considerably improves on the $\Omega(ndT)$ cost of naively implementing the data structure by computing $A x_t$ exactly for each $t\in[T]$.

Our data structures have similar runtime complexity for both $p=1$ and $p=2$ (up to additional logarithmic factors for $p=2$), but potentially much smaller memory complexity for $p=2$. As developed in the rest of this section, our data structure for $p = 1$ needs to store the entire input matrix $A$. In contrast, our data structure when $p = 2$ requires $\otilde(d + n (R/\epsilon)^2)$  space after initialization, which can be sublinear in $\nnz(A)$.

\paragraph{Approach and section organization.}
We prove \Cref{thm:matvec_maintenance} in two steps. First, in \Cref{sec:sub:matvec_est}, we consider the simpler problem of designing a data structure which supports preprocessing $A$ and then outputting $\ell_\infty$ estimates for $A x$ for a single query $x$, under no movement bound assumptions. We call such a data structures an \emph{$\ell_p$-matrix-vector estimation data structure (abbreviated $\MVEp$)}, and provide efficient implementations for $p \in \{1, 2\}$. Our $\MVEp$ when $p = 2$ is then based on linear sketching and our data structure when $p = 1$ is based on random sampling.

Second, in \Cref{sec:sub:est_to_maint} we provide a general reduction from designing a $\MVMp$ to designing $\MVEp$'s. In particular, we provide an  $\MVMp$ which carefully uses $O(\log(R/\epsilon))$ copies of an $\MVEp$ with different accuracy parameters. We use these $\MVEp$'s approximately maintain $A \vref{x}_1,\ldots,A \vref{x}_k$ for $k=O(\log(R/\epsilon))$ reference points $\vref{x}_1, \ldots, \vref{x}_k$. By carefully updating these reference points when the movement is sufficient and using our $\MVEp$'s, we prove \Cref{thm:matvec_maintenance}. 
 
 Our runtimes for $\MVMp$'s for $p \in \{1,2\}$, i.e., \Cref{thm:matvec_maintenance}, are ultimately the same as the cost of initializing our  $\MVEp$'s and performing a single query for a vector that has $\ell_p$-norm at most $R$ (up to logarithmic factors). In other words, even though an $\MVMp$ needs to answer many queries, the computational cost we obtain is comparable to answering a single query to a vector that has $\ell_p$ distance $R$ from the initial point. 

\subsection{Matrix-vector estimation}
\label{sec:sub:matvec_est}

We now formally define an $\MVMp$ data structure (\Cref{def:ds_mv_est}) and efficiently implement it for $p \in \{1, 2\}$. 

\begin{definition}[Matrix-vector estimation]
	\label{def:ds_mv_est}
	We call a data structure an \emph{$\ell_p$-matrix-vector estimation data structure ($\MVEp$)} if it supports the following operations (against an oblivious adversary):
	\begin{itemize}
		\item $\dsInit(A \in \R^{n \times d}, \epsilon \in \R_{> 0}, \delta > 0)$: initialize the data structure with matrix $A$, accuracy parameter $ \epsilon$, and  failure probability $\delta > 0$.
		\item $\dsQuery(x \in \R^d)$: outputs $y \in \R^n$ such that $\norm{y - A x}_\infty \leq \epsilon \norm{A}_{p \rightarrow \infty} \norm{x}_p$ holds with probability at least $1 - \delta$ (for just this query).
	\end{itemize}
\end{definition}

\begin{theorem}[$\ell_2$-matrix-vector estimation]
	\label{thm:mat-vec-ell2}
	There is a $\MVEtwo$ (\Cref{def:ds_mv_est}) that implements $\dsInit(A, \epsilon, \delta)$ for in time $O((\nnz(A) + d) \log(n / \delta))$ and subsequent $\dsQuery(x)$ operations in time $O((\nnz(x) + n \epsilon^{-2}) \log (n/\delta))$.
\end{theorem}

\begin{proof}
	Our data structure is a natural application of CountSketch matrices \cite{charikar2002finding}. We use that, from the literature on CountSketch matrices (see e.g., \cite{charikar2002finding,larsen2021count}), there exists a distribution, $\sketchDist$, on matrices in $\R^{s \times d}$ for $s = O(\epsilon^{-2} \log(n/\delta)$ that have the following properties:
	\begin{itemize}
		\item $Q \sim \sketchDist$ can be computed in $O(d \log (n/\delta))$ time and each column of $Q$ has at most $O(\log (n/\delta)))$ non-zero entries.
		\item There is a procedure $\decode_Q$ that given any input $Q x$ and $Q y$ for $Q \sim \sketchDist$ drawn independently of $x,y \in \R^d$ outputs $\alpha =\decode_Q(Q x, Qy)$ with $|\alpha - \inprod{x}{y}| \leq \epsilon \norm{x}_2 \norm{y}_2$ in $O(\epsilon^{-2} \log (n/\delta))$ time with probability at least $1 - (\delta/n)$.
	\end{itemize}
	To implement $\dsInit$ our data structure draws $Q \sim \sketchDist$ and then computes $y^i = Q A_{i:}^\top$ for all $i \in [n]$. To implement $\dsQuery(x)$ our data structure then outputs $v \in \R^n$ with each $v_i = \decode_Q(y_i, Q x)$. 
	
	To see that our data structure is a $\MVEtwo$ note that $|v_i - \inprod{A_{i:}^\top}{x}| \leq \epsilon \norm{A_{i:}^\top}_2 \norm{x}_2$ with probability at least $1 - (\delta/n)$ by the properties of $Q$. Since $ \norm{A_{i:}^\top}_2 \leq \norm{A}_{2 \rightarrow \infty}$ for all $i \in [n]$ by applying union bound for this event for all $i \in [n]$ we have the desired bound that $\norm{v - Ax}_\infty \leq \epsilon  \norm{A}_{2 \rightarrow \infty} \norm{x}_2$ with probability at least $1 - \delta$.
	
	To bound the algorithm's runtime, first note that computing $Q x$ for any vector $x$ can be implemented in $O(\nnz(x) \log(n/\delta))$ just by considering the $O(\log(n/\delta))$-sparse column of $Q$ for each non-zero entry of $x$. The runtime for $\dsInit$ follows immediately from this and the time to compute $Q$. The runtime for $\dsQuery(\cdot)$ then follows by first computing $Q x$ in time  $O(\nnz(x) \log(n/\delta))$ and then considering the cost of $n$-invocations of $\decode_Q(\cdot,\cdot)$.
\end{proof}

\begin{theorem}[$\ell_1$-matrix-vector estimation]
	\label{thm:mat-vec-ell1}
	For $p = 1$ there is a $\MVEone$ data structure (\Cref{def:ds_mv_est}) that implements $\dsInit(A, \epsilon, \delta)$ for $A \in \R^{n \times d}$ in time $O(\nnz(A))$ and $\dsQuery(x)$ in time $O(\nnz(x) + n \epsilon^{-2} \log (n/\delta))$.\footnote{Each $\dsInit$ can actually be implement in time $O(0)$, i.e., no initialization is required, provided that looking up entries $A$ and the values of $\epsilon$ and $R$ can all be performed in $O(1)$ during $\dsQuery(\cdot)$.}
\end{theorem}

\begin{proof}
	Our data structure is a straightforward application random sampling and a Chernoff bound.
	For any $a,x \in\R^d$ we let $\samp(a,x)$ be a procedure that outputs independent, random $X \in \R$ by picking $i \in [d]$ with probability proportional to $|x_i|$ and then outputting $\norm{x}_1 a_j \sign(x_j)$, i.e., for any $j\in[n]$ 
	\[
	\Pr*(X = \norm{x}_1 a_j \sign(x_j)) = \frac{|x_j|}{\lone{x}},
	\text{ where }
	\sign(t) \defeq 
	\begin{cases}
		1 & \text{ if } t > 0 \\
		0 & \text{ if } t = 0 \\
		-1 & \text{ if } t < 0 
	\end{cases}
	\,.
	\]
	By design, $\E[X] = \inner{a}{x}$ and by a Chernoff bound \cite[see, e.g.,][]{chernoff1952measure}
	we have that for sufficiently large $T = O(\epsilon^{-2} \log(n/\delta))$ and $\alpha = \frac{1}{T} \sum_{t \in [T]} \samp(a,x)$ it is the case that $|\alpha - \inner{a}{x}| \leq \epsilon \norm{a}_\infty \norm{x}_1$ with probability at least $1 - (\delta/n)$. To implement $\dsInit$ our data structure simply saves $A$, $\epsilon$, and $R$. To implement $\dsQuery(x)$ the data structure then outputs $v \in \R^n$ with each $v_i =  \frac{1}{T} \sum_{t \in [T]} \samp(A_{i:}^\top,x)$. 
	
	To see that our data structure is a $\MVEone$ note that $|v_i - \inprod{a}{x}| \leq \epsilon \norm{A_{i:}^\top}_\infty \norm{x}_1$ with probability at least $1 - (\delta/n)$ by the properties of $Q$. Since $ \norm{A_{i:}^\top}_\infty \leq \norm{A}_{1 \rightarrow \infty}$ for all $i \in [n]$ by applying union bound for this event for all $i \in [n]$ we have the desired bound that $\norm{v - Ax}_\infty \leq \epsilon  \norm{A}_{1 \rightarrow \infty} \norm{x}_1$ with probability at least $1 - \delta$.
	
	To bound the algorithm's runtime, first note that, as discussed in \Cref{sec:prelim}, we assumed that we are in a computation model where can process the vector $|x|$ in $O(\nnz(x))$ time to support sampling $i \propto |x_i|$ in time $O(1)$. Leveraging this, we have that with $O(\nnz(x)))$ time spent all subsequent $\samp(\cdot)$ operations can be performed in $O(1)$. Since there are $n T = O(n \epsilon^{-2} \log(n/\delta))$ sample operations the data structure has the desired running time. 
\end{proof}

\subsection{From estimation to maintenance}
\label{sec:sub:est_to_maint}

Now that we have established efficient $\MVEp$'s for $p = 2$ (\Cref{thm:mat-vec-ell2}) and $p = 1$ (\Cref{thm:mat-vec-ell1}), here we use these data structures to prove our main result on $\MVMp$'s (\Cref{thm:matvec_maintenance}). 

We provide a general reduction from $\MVMp$ to $\MVEp$. In particular, we provide a $\MVMp$ in \Cref{alg:matrix_vector_maintenance} that uses $k = O(\log(R/\eps))$ $\MVEp$'s for different accuracy parameters. In \Cref{thm:mat_vec_meta} we prove that for any such implementation and choice of input parameters $\alpha \in \Delta^k$, \Cref{alg:matrix_vector_maintenance} is indeed a $\MVMp$ and we analyze its runtime. We then prove \Cref{thm:matvec_maintenance} by setting $\alpha$, using our $\MVEp$ implementations and applying an additional runtime improvement technique.

\paragraph{Designing and analyzing the data structure.}
Before providing these results and wrapping up the section, here we provide some additional intuition and information regarding \Cref{alg:matrix_vector_maintenance}. In addition to the standard input for a $\MVMp$ and $\delta > 0$, the data structure is specified by  $k$ $\MVEp$'s and parameters $\alpha \in \Delta^k$. In $\dsInit(\cdot)$, the data structure initializes each $\MVEp$ ---denoted $D_1,\ldots,D_k$---and stores $k + 2$ reference vectors $\vref{x}_0,\ldots,\vref{x}_{k +1} \in \R^d$ all initialized to $x_0$ as well as $\vref{y}_0,\ldots,\vref{y}_{k +1} \in \R^n$ all initialized to $A x_0$. The data structure maintains the invariant that $\vref{x}_0 = x_t$ and $\norm{\vref{x}_i - \vref{x_{i - 1}}}_p \leq \epsilon \cdot 2^{i - 2}$ for all $1\le i\le k+1$. It uses this invariant to efficiently maintain that $\vref{y}_i \approx A \vref{x}_i$. It then holds that, at any given time, $\vref{y}_1$ is a valid response to $\dsQuery(\cdot)$.
	
The challenge in designing and analyzing \Cref{alg:matrix_vector_maintenance} is then to maintain these invariants, bound the error in setting $\vref{y}_1$ to be the response to $\dsQuery(\cdot)$, and analyzing the runtime. Maintaining that $\vref{x}_0 = x_t$ and $\norm{\vref{x}_i - \vref{x_{i - 1}}}_p \leq \epsilon \cdot 2^{i - 2}$ is straightforward; after each $\dsQuery(\cdot)$ we simply set $\vref{x}_0 = x_t$ and then update $\vref{x}_i = x_t$ all $i \in [j]$ for the smallest $j$ for which this suffices to preserve the invariant. Due to the choice of $2^{i - 2}$ and the bound on how much the $x_t$ can change, it is straightforward to show that $\vref{x}_i$ for $i \geq 1$ changes at most $O((R/\epsilon)2^{-i})$ times via this procedure. Furthermore, to update $\vref{y}_i$ for all such $i \in [j]$ we simply estimate $A(x_i - x_{i +1})$ using $D_i.\dsQuery(\cdot)$ and add this estimate to $\vref{y}_{i + 1}$. For appropriate choice of accuracies in the $D_i$ (adjusted by the $\alpha_i$) we show this algorithm works as desired. Further, by choosing $\alpha$ and the accuracies, we get a tradeoff between the cost of each $D_i.\dsQuery(\cdot)$ and the number of times it is invoked. Putting these pieces together and carefully reasoning about computational costs then yields our result.

\begin{algorithm2e}[t!]
	\DontPrintSemicolon
	\caption{$\ell_p$ matrix-vector maintenance meta-data structure}
	\label{alg:matrix_vector_maintenance}
	\SetKwProg{Fn}{function}{}{}
	\KwInput{Parameter $p \geq 1$, $\delta > 0$, and $\alpha \in \Delta^{k}$}
	\textbf{State}: $A \in \R^{n \times d}$, $x_0 \in \R^d$, $R \in \R_{> 0}$, $\epsilon \in \R_{> 0}$, $\delta > 0$ \;
	\textbf{State}: Current vector $x_t \in \R^d$, count $t \in \R_{\geq 0}$, and parameter $k \in \Z_{ > 0}$ \;
	\textbf{State}: $\ell_p$ matrix-vector estimation data structures, $D_1, \ldots, D_k$ \tcp*{see \Cref{def:ds_mv_est}}
	\textbf{State}: Reference vectors $\vref{x}_0, \ldots , \vref{x}_{k + 1} \in \R^d$ and $\vref{y}_0, \ldots , \vref{y}_{k + 1} \in \R^n$
	\tcp*{$\norm{\vref{x}_i - \vref{x}_{i-1}}_p \leq \epsilon \cdot 2^{i - 2}$}
	\BlankLine
	\Fn{$\dsInit(A \in \R^{n \times d}, x_0 \in \R^d, R \in \R_{> 0}, \epsilon \in \R_{> 0})$}{
		Save $A$, $x_0$, $R$, and $\epsilon$ as part of data structure's state \;
		$t \gets 0$ and $k \gets \lceil \log_2(\lceil R/\epsilon \rceil) \rceil +  1$\;
		$\vref{x}_i \gets x_0$ and $\vref{y}_i \gets A x_0$ for all $i \in \{0\} \cup [k + 1]$\;
		Set $\epsilon_i \gets \alpha_i 2^{-i}$
		 and call $D_i.\dsInit(A, \epsilon_i, \bar{\delta})$ for $\bar{\delta} \gets \delta \epsilon / R$ and all $i \in [k]$\;
	}
	\BlankLine
	\Fn{$\dsQuery(\Delta_t \in \R^d)$}{
		$x_{t + 1} \gets x_t + \Delta_t$ and $\vref{x}_0 \gets x_{t + 1}$ and then $t \gets t + 1$\;
		Let $j $ be the minimum $i \in [k + 1]$ such that $\norm{x_t - \vref{x}_i}_p \leq \epsilon \cdot 2^{i -2}$ \label{line:j-def}\;
		\lFor{$i \in \{j - 1, \ldots, 1\}$} {
			$\vref{x}_i \gets \vref{x}_0$ and then
			$\vref{y}_i \gets D_i.\dsQuery(\vref{x}_i - \vref{x}_{i + 1}) + \vref{y}_{i +1}$  
			\label{line:query_update}
		}
		\textbf{return} $\vref{y}_1$\; 
	}
\end{algorithm2e}

\begin{theorem}[Reducing Matrix Vector Maintenance to Estimation]
	\label{thm:mat_vec_meta}
	\Cref{alg:matrix_vector_maintenance} is an $\ell_p$-matrix-vector maintenance data structure (\Cref{def:mv_maintenance}). If the runtime for each $D_i.\dsInit(A,\epsilon_i,\delta_i)$ is $\timeInit(i)$ and the runtime for each subsequent $D_i.\dsQuery(\cdot)$ is $\timeQuery(i)$ then \Cref{alg:matrix_vector_maintenance} can implement $\dsInit$ and $T$ $\dsQuery$ operations in total time
	\[
	O\left(
	\nnz(A) + d \cdot \frac{R}{\epsilon} +
	\sum_{t \in [T]} \nnz(\Delta_t)
	+
	\sum_{i \in [k]} \left( \timeInit(i)
	+ 
	\frac{R}{\epsilon \cdot 2^i}
	\cdot [\timeQuery(i)]
	\right)
	\right)
	\,.
	\]
\end{theorem}

\begin{proof}
	We begin by showing that $\norm{x_t - \vref{x_{k + 1}}}_p \leq \epsilon \cdot 2^{k - 1}$ in each execution of \Cref{line:j-def} and therefore the $j$  on \Cref{line:j-def} is well-defined. 
	To see this, note that in each execution of  \Cref{line:j-def} we have 
		\begin{equation}
		\norm{x_t - x_0}_p 
		= 
		\normbigg{\sum_{i \in [t]} (x_i - x_{i - 1}) }_p
		\leq \sum_{i \in [t]} \norm{x_i - x_{i - 1}}_p
		=  \sum_{i \in [t]} \norm{\Delta_i}_p
		\leq R\,.
	\end{equation}
	Since $\epsilon \cdot 2^{k - 1}  \geq \epsilon \cdot 2^{\log_2(R/\epsilon)} = R$ so long as $\vref{x}_{k + 1} = x_0$ then  $\norm{x_t - \vref{x_{k + 1}}}_p \leq \epsilon \cdot 2^{k - 1}$. However, $\vref{x}_{k + 1} = x_0$ is set in $\dsInit$ and then never updated (since $i \leq k$ is on \Cref{line:query_update}) and the claim follows.
	
	Leveraging that $j$ is well-defined on \Cref{line:j-def}, we show that before and after each call to $\dsQuery(\cdot)$, $\norm{\vref{x}_i - \vref{x}_{i-1}}_p \leq \epsilon \cdot 2^{i - 2}$ for all $i \in [k]$. This invariant holds after $\dsInit$ as each $\vref{x}_i$ is initially set to $x_0$. Next, suppose the invariant holds before a call $\dsQuery(\cdot)$. $\vref{x}_i$ are only changed on \Cref{line:query_update} and for $i \leq j - 1$, in which case they are set to $\vref{x}_0$. However, $\norm{\vref{x}_0 - \vref{x}_j}_p \leq \epsilon \cdot 2^{j - 2}$ by the definition of $j$ (\Cref{line:query_update}) and that $j$ is well-defined. Therefore, after the call to $\dsQuery(\cdot)$ the invariant holds since $\norm{\vref{x}_{j - 1} - \vref{x}_{j}} \leq \epsilon \cdot 2^{j - 2}$ and $\norm{\vref{x}_{i} - \vref{x}_{i + 1}} = 0 \leq \epsilon \cdot 2^{i - 2}$ for all $i \in [j - 2]$.
	
	Next, we show that for all $i \in [k]$, throughout the use of \Cref{alg:matrix_vector_maintenance} as an $\MVMp$, 
	$D_i.\dsQuery(\cdot)$ is called on \Cref{line:query_update} at most $R \epsilon^{-1} 2^{-(i - 2)}$ times. Whenever $D_i.\dsQuery(\cdot)$ is called on \Cref{line:query_update} it must be the case that $\norm{x_t - \vref{x}_{i}}_p > \epsilon \cdot 2^{i - 2}$ (as otherwise $j \leq i$ by the definition of $j$ on \Cref{line:j-def}). Let $v_0,...,v_{L}$ denote the sequence of different $\vref{x}_{i}$ vectors set on \Cref{line:query_update} (where $v_0 = \vref{x}_0$); we have just argued that $\norm{v_\ell - v_{\ell - 1}}_p > \epsilon \cdot 2^{i - 2}$ for $\ell \in [L]$. Further, since the $v_\ell$ are a subsequence of the $x_t$, triangle inequality implies that 
	\[
	R \geq \sum_{t \in [T]} \norm{x_{t}- x_{t -1}}_p \geq \sum_{\ell \in [L]} \norm{v_{\ell}- v_{\ell -1}}_p > L \cdot \epsilon \cdot 2^{i - 2} \, .
	\]
	Since $D_i.\dsQuery(\cdot)$  is invoked $L $ times, the claim follows. 
	
	Leveraging the previous properties, we next establish that with probability at least $1 - \delta$ before and after each call to $\dsQuery(\cdot)$, we have that $\norm{\vref{y}_i - A \vref{x}_i}_\infty \leq \sum_{j = i}^{k - 1} \frac{\alpha_i \epsilon}{2} \leq \frac{\epsilon}{2}$ for all $i \in [k + 1]$. By the preceding paragraph, we know that the total umber of matrix-vector estimation queries on \Cref{line:query_update} is at most
	\[
	\sum_{i \in [k]} \frac{R}{\epsilon \cdot 2^{i - 2}} 
	\leq \frac{R}{2 \epsilon} \sum_{i = 0}^{\infty} \frac{1}{2^i}
	= \frac{R}{2 \epsilon}\,. 
	\]
	Further, by the definition of a $\MVEp$ (\Cref{def:ds_mv_est}) 
	and by the union bound with probability at least $1 - (\bar{\delta} R/ (2 \epsilon)) \geq 1 - \delta$ every call to $D_i.\dsQuery(\vref{x}_0 - \vref{x}_{i + 1})$ on \Cref{line:query_update} outputs a vector $z_i$ where 
	\begin{equation*}
		\norm{z_i - A(\vref{x}_i - \vref{x}_{i + 1})}_\infty \leq 
		\epsilon_i \norm{A}_{p \rightarrow \infty} \norm{\vref{x}_i - \vref{x}_{i + 1}}_\infty
		\leq \frac{\epsilon}{2} \cdot  \alpha_i
	\end{equation*}
	where we used the definition of $\epsilon_i$, that $\norm{A}_{p \rightarrow \infty} \leq 1$ by assumption and that $\norm{\vref{x}_i - \vref{x}_{i + 1}} \leq \epsilon \cdot 2^{i - 1}$ in the last inequality. Consequently, with probability $1 - \delta$, before and after each call to $\dsQuery(\cdot)$  we have that for all $i \in [k - 1]$,
	\begin{align*}
		\norm{\vref{y}_i - A \vref{x}_i}_p
		\leq \norm{z_i - A(\vref{x}_i - \vref{x}_{i + 1})}_p + \norm{\vref{y}_{i + 1} -  A \vref{x}_{i+1}}_p 
		\leq \frac{\epsilon \cdot \alpha_i}{2} +  \norm{\vref{y}_{i + 1} -  A \vref{x}_{i+1}}_p\,.
	\end{align*}
	The claim then follows by induction and the facts that $\norm{\vref{y}_{k + 1} -  A \vref{x}_{k + 1}} = 0$ (they are never changed after initialization) and $R \epsilon^{-1} 2^{-(k-2)} < 1$. 
	
	We now have everything necessary to prove that \Cref{alg:matrix_vector_maintenance} is a $\MVMp$ (\Cref{def:mv_maintenance}). Note that with probability $1 - \delta$ after each call to $\dsQuery(\cdot)$ we have argued that $\norm{\vref{y}_i - A \vref{x}_i}_p \leq \frac{\epsilon}{2}$ and that $\norm{\vref{y}_0 - \vref{y}_1}_p \leq \epsilon \cdot 2^{-1}$. Consequently, 
	\[
	\norm{\vref{y}_1 - A x_t}_p = 
	\norm{\vref{y}_1 - A \vref{x}_0}_p
	\leq \norm{\vref{y}_1 - A \vref{x_1}}_p + \norm{A(\vref{x}_1 - \vref{x}_0)}_p
	\leq \frac{\epsilon}{2} + \norm{A}_{p \rightarrow \infty} \norm{\vref{x}_1 - \vref{x}_0}_p
	\leq \epsilon\,.
	\]
	
	To complete the proof, we need to bound the data structure's runtime. Note that $\dsInit$ can be implemented in time $O(
	\nnz(A)
	+
	\sum_{i \in [k]}  \timeInit(i))$
	by simply performing the operations (and saving multiple copies of vectors and matrices with pointers as needed). Next, note that changes to $x_{t}$, $\vref{x}_0$, and $\vref{x}_0 - \vref{x}_1$ due to $\vref{x}_0$ changing can be computed in $O(\sum_{t \in [T]} \nnz(\Delta_t))$ time. With this, it is possible to keep track of the changes to $\norm{\vref{x}_0 -\vref{x}_1}_p$ due to $\vref{x}_0$ changing in  $O(\sum_{t \in [T]} \nnz(\Delta_t))$ time as well. Whenever $j > 1$ in \Cref{line:j-def}, if we spend $O(d j)$ time to implement \Cref{line:j-def} and $O(d)$ plus the $D_i.\dsQuery(\cdot)$ costs in each iteration of \Cref{line:query_update} then the total additional cost of $\dsQuery(\cdot)$ over all invocations is 
	\[
	O\left(
	d + k +
	\sum_{i \in [k]} 
	\left[
	\frac{R}{\epsilon \cdot 2^i}
	\cdot [\timeQuery(i)]
	+ d
	\right]
	\right)
	= 
	O\left(
	d \cdot \left\lceil \frac{R}{\epsilon} \right\rceil 
	+k + 
	\sum_{i \in [k]} \left(\frac{R}{\epsilon \cdot 2^i}
	\cdot [\timeQuery(i)]
	\right)
	\right)
	\,.
	\]
	provided that only changes to the output of $\dsQuery(\cdot)$ are reported.
\end{proof}

We conclude the section by proving \Cref{thm:matvec_maintenance}, the main result that we use in other sections.

\begin{proof}[Proof of \Cref{thm:matvec_maintenance}]
	Apply \Cref{thm:mat_vec_meta} using \Cref{thm:mat-vec-ell2} and \Cref{thm:mat-vec-ell1} respectively. Using these algorithms for all $i \in [k]$
	\[
	\timeInit(i) = O\left((\nnz(A) + d) \log^{p -1} 
	\left( \frac{n R}{\epsilon \delta} \right)
	\right)
	\text{and } 
	\timeQuery(i) = O\left( d \log^{p - 1} \left(\frac{nR}{\epsilon \delta} \right) + n \epsilon_{i}^{-2} \log \left(\frac{nR}{\epsilon \delta}\right)
	\right) \,.
	\]
	Next, to optimize the contribution of the $\epsilon_i$ terms to to the final runtime, pick $
	\alpha_i \propto 2^{i/3}$, i.e. $\alpha_i = 2^{i/3} / (\sum_{j \in [k]} 2^{j/3})$. Using that $\epsilon_i =  2^{-i} \alpha_i$ this yields that 
	\[
	\sum_{i \in [k]} \frac{1}{2^i} \cdot \frac{1}{\epsilon_i^2}
	= \sum_{i \in [k]} \frac{2^i}{\alpha_i^2}
	= \left(\sum_{i \in [k]} 2^{i/3}  \right)^3
	= \left(
	\frac{2^{(k +1)/3} - 1}
	{2^{1/3} - 1}
	\right)^3
	= O(2^k) = O\prn*{\frac{R}{\epsilon}}
	\]
	where in the last step we used the definition of $k$. Combining with the facts that $\sum_{i \in [k]} \frac{R}{\epsilon 2^i} = O(\frac{R}{\epsilon})$ and $k = O(\log(R/\epsilon))$ yields that
	\begin{align*}
		\sum_{i \in [k]} \timeInit(i)
		& 
		= O\left( \left(\nnz(A) + d\right) \log^{p -1} \left( \frac{n R}{\epsilon \delta}\right) \log\left(\frac{R}{\epsilon}\right) \right) \text{ and }
		\\
		\sum_{i \in [k]} \frac{R}{\epsilon \cdot 2^i} \cdot [\timeQuery(i)]
		&
		= O\left( d \cdot \frac{R}{\epsilon} \log^{p -1} \left(\frac{n R}{\epsilon \delta}\right)
		+ n \left(\frac{R}{\epsilon}\right)^2 \log \left(\frac{n R}{\epsilon \delta}\right) \right)\text{\,.}
	\end{align*}
	The result for $p = 2$ then follows via \Cref{thm:mat_vec_meta} and the $\log(R/\epsilon) = O(R/\epsilon)$. 
	
	To obtain the result for $p = 1$ we proceed identically and add one further improvement on the algorithm's implementation. In the case of $p =1$ that rather spending $O(\nnz(A) \log(R/\epsilon))$ time in each $D_i.\dsInit(\cdot)$ we can simply save the matrix once and use it for each $D_i$. This removes the logarithmic factors on the $\nnz(A)$ terms in the runtime for $p = 1$ and yields the desire result.
\end{proof}

Even though we only use \Cref{alg:matrix_vector_maintenance} to prove \Cref{thm:matvec_maintenance} and in turn only apply \Cref{thm:matvec_maintenance} in a restricted set of settings, we provide the more general algorithm and analysis as it may be useful in additional settings. In particular, we allowed $\alpha$ to be a parameter because if we were in a setting where the runtime of each $D_i.\dsQuery(\cdot)$ had a different dependence on $\epsilon$, e.g., $\epsilon^{-1}$ rather than $\epsilon^{-2}$, then other configurations of $\alpha$ might be preferable, e.g., uniform with $\alpha_i = \frac{1}{k}$.
The particular choice of $\alpha_i \propto 2^{i/3}$ in the proof of \Cref{thm:matvec_maintenance} improves over $\alpha_i =\frac{1}{k}$ by logarithmic factors.

Note that, with more careful analysis, it may be possible to improve the dependence on $d$ in \Cref{thm:matvec_maintenance}, potentially at the cost of additional logarithmic factors. The current dependence arises by accounting for at least $d$ time whenever $j  > 1$ on \Cref{line:j-def}. However, in the case that $\Delta_t$ are sparse one could instead maintain the difference from $x_t$ to each $\vref{x}_i$ and seek faster implementations of $D_i.\dsQuery(\cdot)$ provided that the input changes sparsely. We do not pursue such an improvement for simplicity and since the term proportional to $d$  does not affect our final runtimes. 

\section{Efficient gradient estimation via matrix-vector maintenance}\label{sec:grad-est}

In this section, build upon the data structures developed in the previous section to provide an efficient stochastic gradient oracle for the ``softmax'' approximation of the original objective.
Recall that the ``softmax'' of functions $f_1,\ldots,f_n$ is
\begin{align*}
	\fsm(x) & =\eps'\log\left(\sum_{i\in[n]}\exp\left(\frac{f_i(x)}{\eps'}\right)\right),\\
\text{and}~\nabla f_{\mathrm{smax}}(x) & = \sum_{i\in[n]}p_i(x)\nabla f_i(x)~~\text{where}~~p_i(x) = \frac{\exp(f_i(x)/\eps')}{\sum_{i\in[n]}\exp(f_i(x)/\eps')}.
\end{align*}
Throughout this section we assume that each $f_i$ is 
$\Lg$-smooth and $\Lf$-Lipschitz. 

\Cref{alg:grad-est} provides an unbiased estimator of $\nabla \fsm(x)$ by leveraging a matrix-vector maintenance data structure $\mc{M}$. The algorithm takes as input a sequence of query points $x_1,\ldots,x_T$ that satisfies $\norm{x_t-x_0}\le r$ and  $\sum_{t \le T} \norm{x_t - x_{t-1}} \le r'$ for $r,r'>0$ such that $\half \Lg r^2 \le \epsilon'$. It outputs a sequence of vectors $\gest(x_1), \ldots, \gest(x_T)$ such that (informally) $\Ex*{\gest(x_t) | \mc{M},x_1,\ldots,x_t} = \grad \fsm(x_t)$ for all $t\le T$ with high probability. To compute these estimates the algorithm requires, with high probability, $\Otil{n+T}$ individual function value and gradient calculations, as well as $\Otilb{(n + T)d + d(\Lf r'/\eps') + n(\Lf r'/\eps')^2 }$ additional runtime. We state this guarantee in full detail in the following.

\begin{theorem}[Softmax gradient estimator]\label{thm:grad-est}
	Let $p\in\{1,2\}$ and let $\{f_i\}_{i\in[n]}$ be $\Lg$-smooth and $\Lf$-Lipschitz with respect to
	$\norm{\cdot}_p$.
	For all $t\in [T]$ assume that input $x_t$ to \Cref{alg:grad-est} is a (deterministic)  function of the previous outputs $\gest(x_1),\ldots, \gest(x_{t-1})$, and that $\norm{x_t-x_0}_p\le r$ and $\sum_{t \le T} \norm{x_t - x_{t-1}} \le r'$ hold for parameters $r,r'>0$ such that $\half \Lg r^2 \le \eps'$ and $\eps' \le \Lf r' / 2$. Let $\filt$ be the filtration induced by all the random bits \Cref{alg:grad-est} draws up to iteration $t$ and all those that may be used by $\mathcal{M}$.
	Then for any error tolerance $\delta \in (0,1)$ there exists event $\mc{E}$ such that the following hold:
	\begin{itemize}
	\item We have $\Pr(\mc{E}) \ge 1-\delta$.
	\item When $\mc{E}$ holds we have $\Ex*{\gest(x_t)|\filt[t-1]} = \grad \fsm(x_t)$ for all $t\in[T]$.
	\item When $\mc{E}$ holds, \Cref{alg:grad-est} makes $O(n+T\log(1/\delta))$ queries of the form $\{f_i(x), \nabla f_i(x)\}$, and requires additional runtime
	\[O\left(T\bigg(d+\log\bigg(\frac{1}{\delta}\bigg)\bigg)+\left(nd\log^{p-1}\bigg(\frac{\Lf r'}{\eps'}\bigg)+d\bigg(\frac{\Lf r'}{\eps'}\bigg)\right)\log^{p-1}\bigg(\frac{n\Lf r'}{\eps'\delta}\bigg)+n\bigg(\frac{\Lf r'}{\eps'}\bigg)^2\log\frac{n\Lf r'}{\eps'\delta}\right).\]
	\item With probability 1 we have $\norm{\gest(x_t)}_{p^\star} \le \Lf$, where $p^*$ is such that $\frac{1}{p} + \frac{1}{p^*}=1$.
	\end{itemize}
\end{theorem}

\begin{algorithm2e}[t]
	\DontPrintSemicolon
	\caption{Softmax gradient estimator}
	\label{alg:grad-est}
	\KwInput{$\{f_i\}_{i\in[n]}$, query sequence $\{x_t\}_{t\le T}$ such that $x_t$ is a function of the previous outputs $\gest(x_1), \ldots, \gest(x_{t-1})$ (i.e., $x_0$ and $x_1$ do not depend on any outputs).}
	\KwParameters{Softmax tolerance $\epsilon'$, movement bound $r'$, Lipschitz constant $\Lf$, error tolerate $\delta \in (0,1)$, $\ell_p$-matrix-vector maintenance data structure $\mc{M}$.}
	Call 
	$\mathcal{M}.\dsInit(A, 0, r', \frac{\eps'}{\Lf},\frac{\delta}{2})$ where $A = [\frac{1}{\Lf}\nabla f_i(x_0)^\top ]_{i\in[n]}$\;
	\For{$t=1,2,\cdots, T$}{
	$y_t\gets \Lf \cdot \mathcal{M}.\dsQuery(x_t-x_{t-1})$ \Comment{maintain vector $y_t\approx \Lf A (x_t-x_0) = \brk*{\inner{\grad f_i(x_0)}{x_t-x_0}}_{i\in[n]}$} 
	\textsf{accepted} $\gets \False$\;
	\While{\textup{not \textsf{accepted}}}{
	Draw $i\sim \exp\prn*{\frac{f_i(x_0)+[y_t]_i}{\epsilon'}}$\;
	\Block{With probability $\min\crl*{\exp\prn*{\frac{f_i(x_t)-f_i(x_0)-[y_t]_i}{\epsilon'}-2}, 1}$\label{line:rejection-step}}{\textbf{yield} $i_t = i$ and $\gest(x_t) = \nabla f_{i_t}(x_t)$\;
	\textsf{accepted} $\gets \True$\;
	}
	}
	}
\end{algorithm2e}

\begin{proof}
	We prove the theorem by coupling \Cref{alg:grad-est} to an ``alternative'' algorithm that uses $\mc{M}$ in a strictly oblivious manner, and produces a potentially different sequence of indices $i_1',\ldots,i_T'$, queries $x_2', \ldots, x_T'$ and matrix-vector estimates $y_t' = \Lf \cdot \mc{M}.\dsQuery(x_t' - x_{t-1}')$. The alternative algorithm proceeds exactly like \Cref{alg:grad-est}, except at every iteration it tests whether
	\begin{equation}
	\label{eq:good-approx-cond-t}
	\max_{i\in[n]} \abs[\Big]{[y_t']_i - \inner{\grad f_i(x_0)}{x_t'-x_0}} = 
	\Lf \norm*{\tfrac{1}{\Lf}y_t' - A (x_t'-x_0)}_\infty \le \eps'
	\end{equation}
	holds. As long as this condition holds, the algorithm produces $i_t'$ using rejection sampling as in \Cref{alg:grad-est} (and with the same random bits). If at any $t\le T$ the condition fails, the algorithm proceeds to directly draw $i_t' \sim e^{f_i(x_t')/\eps'}$ at all subsequent iterations, ignoring the values of $y_t'$. Thus, both algorithms produce identical outputs $\gest(x_t')=\gest(x_t)$ leading to identical queries $x_t'=x_t$ whenever \Cref{eq:good-approx-cond-t} holds for all $t\in[T]$.

	For the alternative algorithm we have $i_t' \sim e^{f_i(x_t')/\eps'}$ for all $t\in[T]$, regardless of randomness in $\mc{M}$. To see this, note that by smoothness of the $f_i$ we have
	\[
		\abs*{\frac{f_i(x_t')-f_i(x_0)- \inner{\grad f_i(x_0)}{x_t'-x_0}}{\epsilon'}}	\le \frac{\half \Lg r^2 }{\epsilon'} \le 1
	\]
	for all $i\in[n]$, 
	by our assumptions that $\norm{x_t'-x_0}_p \le r$ and $\half \Lg r^2 \le \epsilon'$. Consequently, when \Cref{eq:good-approx-cond-t} holds we have
	\[
		\abs*{\frac{f_i(x_t')-f_i(x_0)- [y_t']_i}{\epsilon'}}	\le 
		1 + 
		\abs*{\frac{\inner{\grad f_i(x_0)}{x_t'-x_0}- [y_t']_i}{\epsilon'}}
		\le 2
	\]
	for all $i\in [n]$. Therefore, $\exp\prn*{\frac{f_i(x_t')-f_i(x_0)-[y_t']_i}{\epsilon'}-2}\le 1$ and (by standard analysis of rejection sampling) we have $\Pr(i_t' = i) \propto \exp\prn*{\frac{f_i(x_0)+[y_t']_i}{\epsilon'}} \cdot \exp\prn*{\frac{f_i(x_t')-f_i(x_0)-[y_t']_i}{\epsilon'}} = e^{f_i(x_t')/\eps'}$. As a consequence, the alternative algorithm's outputs satisfy $\Ex*{\gest(x_t')|\filt[t-1]'} = \grad \fsm(x_t')$ for all $t\le T$ by definition of the softmax function.

	Since the alternative algorithm's queries are oblivious to the data structure's randomness, we may apply\footnote{
	The matrix $A$ satisfies $\norm{A}_{p\to \infty} \le 1$ since $\norm{\grad f_i(x_0)}_{p^*} \le \Lf$ for all $i\in[n]$ by the Lipschitz continuity assumption.
	}
	\Cref{thm:matvec_maintenance} to conclude that, with probability at least $1-\frac{\delta}{2}$ we have $\norm*{\frac{1}{\Lf}y_t' - A(x_t'-x_0)}_\infty \le \frac{\epsilon'}{\Lf}$ for all $t\le T$, implying that the condition \eqref{eq:good-approx-cond-t} holds for all $t\le T$ and therefore both algorithms produce identical outputs.   This defines a probability $\ge 1-\frac{\delta}{2}$ event under which $\Ex*{\gest(x_t)|\filt[t-1]} = \grad \fsm(x_t)$ for all $t\le T$, giving the first part of the theorem. 

	For the second part of the theorem, we note that smoothness and the condition \eqref{eq:good-approx-cond-t} also imply that the rejection probability $\exp\prn*{\frac{f_i(x_t')-f_i(x_0)-[y_t']_i}{\epsilon'}-2}\ge e^{-4}$. Therefore, the expected number of rejection sampling steps in the alternative algorithm is $O(1)$. By standard Chernoff bounds~\cite[see, e.g.,][]{chernoff1952measure}, with probability at least $1-\frac{\delta}{2}$ the alternative algorithms makes $O(T\log\frac{1}{\delta})$ rejection sampling steps throughout. Thus, by a union bound we have that with probability $1-\delta$, \Cref{alg:grad-est} and the alternative algorithm are identical, with each making $O(T\log\frac{1}{\delta})$ rejection sampling steps. Each rejection sampling step costs $O(1)$ function and gradient evaluations, and to construct the matrix $A$ we require $n$ additional evaluations. Additionally, given the computational model of this paper, with $O(n)$ preprocessing we can implement each random sampling of $i$ in $O(1)$ time as discussed in \Cref{sec:prelim}. Altogether, this brings the overall cost to $O(n + T\log\frac{1}{\delta})$ and the bound on additional runtime follows immediately from \Cref{thm:matvec_maintenance}. 

	Finally, the third part of the theorem is immediate from noting that $\gest(x_t) = \grad f_{i_t}(x_t)$ for some $i_t\in[n]$ and therefore $\norm{\gest(x_t)}_{p^*}\le \Lf$ by the Lipschitz continuity of the $f_i$. 	
 \end{proof}

\section{Runtime bounds}\label{sec:runtimes}

We now put together the pieces constructed in the previous sections to obtain runtime bounds for minimizing the maximum of convex functions. In \Cref{subsec:runtimes-general} we study general convex functions, in \Cref{subsec:runtimes-linear} we specialize our results to linear functions, and in \Cref{subsec:runtimes-minEB} we specialize them further to the problem of finding a minimum enclosing ball. 

\subsection{General convex functions}\label{subsec:runtimes-general}
Recall the problem
\begin{equation}\label{def:mtm}
	\minimize_{x\in\xset}\crl*{\fmax(x)\defeq\max_{y\in\Delta^n}\sum_{i\in[n]}y_i f_i(x)}.
\end{equation}
We consider the problem in two different settings, which we call the ball setup or the simplex setup, formally defined as follows.

\begin{definition}[Ball setup]\label{def:ball-setup} In the \emph{ball setup}, the norm $\norm{\cdot}$ is the Euclidean norm $\norm{\cdot}_2$, the domain $\xset$ is a closed and convex subset of the unit Euclidean ball $\ball^d = \{x\in\R^d \mid \norm{x}_2\le 1\}$, and the Bregman divergence is $V_x(y) = \half \norm{y-x}_2^2$. Furthermore, we let $\xset_\nu \defeq \xset$ for all $\nu \ge 0$.
\end{definition}

\begin{definition}[Simplex setup]\label{def:simplex-setup} In the \emph{simplex setup}, the norm $\norm{\cdot}$ is the 1-norm $\norm{\cdot}_1$, the domain $\xset$ is a closed and convex subset of the probability simplex $\Delta^d = \{x\in\R^d_{\ge 0} \mid \sum_{i\in[d]} x_i =1 \}$, and the Bregman divergence is $V_x(y) = \sum_{i\in[d]} y_i \log \frac{y_i}{x_i}$. Furthermore, we let $\xset_\nu \defeq \{x\in\xset \mid x_i \ge \nu, ~ \forall i\in[d]\}$ for all $\nu\ge0$.
\end{definition}
\noindent
We introduce the set $\xset_\nu$ in the definitions above in order to satisfy the $\tau$-triangle in the simplex setup; see~\Cref{def:tau-triangle} and~\Cref{example:truncated-entropy-inequality}.

The following is our main result concerning the complexity of solving the problem~\eqref{def:mtm}.

\newcommand{\Teval}{\mathcal{T}_{\textup{eval}}}
\newcommand{\Tmd}{\mathcal{T}_{\textup{md}}}

\begin{theorem}\label{thm:mtm}
	Consider the problem~\eqref{def:mtm} in either the ball or simplex setups (\Cref{def:ball-setup,def:simplex-setup}, respectively), where each function $f_i$ is convex, $\Lf$-Lipschitz, and $\Lg$-smooth with respect to $\norm{\cdot}$. Let $\epsilon>0$, let $\nu = \frac{\epsilon}{4d\Lf}$, and for initial point $x_0 \in \xset_\nu$ let $\max_{x\in\xset_\nu} V_{x_0}(x) \le \half R^2$. 
	Then, \Cref{alg:framework} with parameters $r \le \sqrt{\frac{\epsilon}{\Lg \log n}}$, $R$, $\mathcal{E}_0=L_f R$, accuracy $\frac{\eps}{8}$, ball oracle implementation~\Cref{alg:ball-oralce} and gradient oracle implementation in~\Cref{alg:grad-est}, return a point $x$ such that \[\E \fmax(x) - \min_{\xopt \in \xset} \fmax(\xopt) \le \epsilon.\]

	Let $\Teval$ be the time to compute $f_i(x),\grad f_i(x)$ for any $x\in\xset$ and $i\in[n]$, and let $\Tmd$ be the time to compute a mirror descent step of the form $\argmin_{z\in\xset_\nu} \crl*{\inner{g}{z} + \lambda V_{y}(z)+ V_{x}(z)}$ for any $g\in\R^d$ and $x,y\in\xset$. For $\epsilon \le \min\crl*{\Lg R^2, \Lf^2/\Lg}$ and $r = \min\crl*{\sqrt{\frac{\epsilon}{\Lg \log n}}, \frac{\epsilon \sqrt{\Teval +d}}{\Lf}}$ with probability at least $\frac{9}{10}$, the algorithm has runtime
	\begin{equation}\label{eq:computation-general}
		\Otilb{
			n (\Teval +d) \prn*{\frac{\Lg R^2}{\epsilon}}^{1/3} + n \prn*{\frac{(\Teval + d)\Lf R}{\epsilon}}^{2/3} + (\Teval + \Tmd + d)\prn*{\frac{\Lf R}{\epsilon}}^{2} 
			}.
	\end{equation}
\end{theorem}

\begin{proof} We establish the theorem in four steps: reducing the objective to softmax on a truncated domain, describing the gradient oracle implementation, arguing the correctness of our methods, and finally bounding the runtime.

\paragraph{Reduction.}
We claim it suffices to solve to $\eps/4$ additive error the problem
\begin{equation}\label{def:mtm-redx}
\begin{aligned}
	\minimize_{x\in\xset_{\nu}}~\fsm(x)~~\text{where}~~\fsm(x) = \eps'\log\left(\sum_{i\in[n]}\exp\left(\frac{f_i(x)}{\eps'}\right)\right),~\eps' = \frac{\eps}{2\log n},~\nu = \frac{\eps}{4d\Lf}.
\end{aligned}
\end{equation}
To see the claim is true, note we have $|\fsm(x)-\fmax(x)|\le \eps/2$ for all $x\in\xset$ and  
\[
	\min_{x\in\xset}\fmax(x)\le \min_{x\in\xset_\nu}\fmax(x)\le \min_{x\in\xset}\fmax(x)+\eps/4
\] due to $\Lf$-Lipschitz continuity of $f$. Consequently, for any $\tilde{x}$ that is an $\eps/4$ approximate optimizer of the true minimizer of~\eqref{def:mtm-redx}, we have 
\[\fsm(\tilde{x})\le \min_{x\in\xset_\nu}\fsm(x)+\eps/4\le\min_{x\in\xset_\nu}\fmax(x)+\eps/2+\eps/4\le \min_{x\in\xset}\fmax(x)+\eps/4+\eps/2+\eps/4.\]
This proves such $\tilde{x}$ is also an $\eps$-additive minimizer of the original problem~\eqref{def:mtm-redx}. We therefore focus on solving ~\eqref{def:mtm-redx} to $\eps/4$ additive error. 

\paragraph{Stochastic gradient oracle.} 
At the $t$'th ball oracle call in the outer loop (\Cref{alg:framework}) we instantiate a gradient estimator $\gest$ for $\grad\fsm$ using \Cref{alg:grad-est} with initial point $\Phi_t(v_t)$, parameters $r$ and $r' = \Otil{r}$, and failure probability $\delta = \frac{\epsilon}{\Lf R} \cdot \frac{1}{100\Tmax}$ with $\Tmax = \Otil{(R/r)^{2/3}}$ such that \Cref{thm:outerloop} guarantees (via Markov's inequality) that \Cref{alg:framework} requires at most $\Tmax$ iterations with probability at least $\frac{99}{100}$. Since \Cref{thm:grad-est} only guarantees that this gradient estimator is unbiased for $\grad \fsm$ with high probability, we repeat the coupling argument from the proof of \Cref{thm:grad-est}. Namely, we consider ``alternative'' completely unbiased gradient estimators that with probability at least $1-\delta$ produce identical outputs to \Cref{alg:grad-est}. We then analyze an alternative algorithm with the alternative estimators, and use the fact that with probability at least $1-\delta\Tmax$ it produces the same output as our algorithm. By our choice of $\delta$ and $\Tmax$, we have that with probability at least $1-\frac{\epsilon}{50\Lf R}$ the actual and alternative gradient estimators produce identical outputs for the entire duration of the algorithm.

\paragraph{Correctness.} 
For the ball and simplex setups, our chosen Bregman divergence $V_x(y)$ is $1$-strongly convex with respect to the $\ell_2$ or $\ell_1$ norm, respectively. The corresponding divergence also satisfies a $\tef$-triangle inequality (\Cref{def:tau-triangle}) with $\tau = \widetilde{\Theta}(1)$. For the $k$'th out loop iteration, let us argue that the stochastic gradient queries made the inner loop of \Cref{alg:ball-oralce} satisfy the conditions of \Cref{thm:grad-est}. Let $\gest$ denote the estimator for $\grad \fsm$ defined above, and let $\gest_k(x) = a_{k+1} \gest(\Phi_k(x))$ be the gradient estimator for $h_k$ defined in \Cref{alg:framework} and let $y=\Phi(v_k)$. Let $x_1,\ldots,x_T$ denote the sequence of queries to $\gest_k$ made by \Cref{alg:ball-oralce}. Then \Cref{thm:oracle-impl} guarantees that $\norm{x_t - v_k} \le \rho$ for all $t\in[T]$ and that $\sum_{t\in[T]} \norm{x_t - x_{t-1}} = \Otil{\rho}$. By design of \Cref{alg:framework} we have that $\Phi_k(z)-\Phi_k(z') = \frac{r}{\rho}(z-z')$ and therefore the queries to $\gest$ satisfy $\norm{\Phi_k(x_t) - y} \le r$ for all $t\in[T]$ and $\sum_{t\in[T]} \norm{\Phi_k(x_t) - \Phi_k(x_{t-1})} = \Otil{r} = r'$ as required by \Cref{thm:grad-est}.

With the conditions of \Cref{thm:grad-est} satisfied, we have a nearly unbiased gradient estimator for $\fsm$, that with probability at least $1-\frac{\epsilon}{50\Lf R}$ produce identical outputs to a completely unbiased gradient for the entire duration of the algorithm, as discussed above. \Cref{thm:oracle-impl} then guarantees that (using the alternative gradient estimator) \Cref{alg:ball-oralce} implements a valid $(\rho,\gamma,\cmax)$ restricted proximal oracle for $\rho = \Thtil{R^{2/3}r^{1/3}}$, $\gamma = \Otil{1}$ and $\cmax<\infty$. We may therefore apply \Cref{thm:outerloop} (with $\epsilon \to \epsilon/8$) to conclude that with alternative gradient estimator we output $x'$ such that
\[
	\E \fsm(x') \le \min_{\xopt \in \xset_\nu} \fsm(\xopt) + \frac{\eps}{8}.
\]
Letting $x$ be the output of the algorithm using the actual gradient estimator, we have
\begin{flalign*}
	\E \fsm(x) &= \E \fsm(x')\indic{x=x'} + \E \fsm(x)\indic{x\ne x'}
	\\ & \overle{(i)}
	\E \fsm(x') + \E \Lf \norm{x-x_0} \indic{x\ne x'}
	\\ & \overle{(ii)} 
	\E \fsm(x') + \Lf R \cdot \Pr(x\ne x')
	\\ & \overle{(iii)} 
	\min_{\xopt \in \xset_\nu} \fsm(\xopt) + \frac{\eps}{8} + \Lf R \cdot \frac{1}{50\Lf R} \le \min_{\xopt \in \xset_\nu} \fsm(\xopt)  + \frac{\epsilon}{4},
\end{flalign*}
due to the $(i)$ Lipschitz continuity of $\fsm$, $(ii)$ the definition of $R$, and $(iii)$ the bounds on $\E f(x')$ and the probability of $x=x'$ discussed above. This proves the correctness of our algorithm.

\paragraph{Complexity.} 
By~\Cref{thm:outerloop} and the discussion above, the outer loop (\Cref{alg:framework}) terminates in~$\Tmax = \widetilde{O}(R^{2/3}r^{-2/3})$ iterations with probability at least $\frac{99}{100}$. Each iteration of the outer loop performs $O(1)$ operations on $d$-dimensional vectors  and makes one call to a ball restricted proximal oracle.

By \Cref{thm:oracle-impl}, each restricted ball oracle call makes $\Otilb{\Gamma^2/\rho^2}$ calls to the gradient estimator, mirror descent step computations, and $d$-dimensional vector arithmetic operations.\footnote{Logarithmic factors in \Cref{thm:oracle-impl} depend on a bound for $\max_{x,y\in\xset_\nu}V_x(y)$ whereas we only assumed $\max_{y\in\xset_\nu}V_{x_0}(y) \le R^2 /2$. However, a $\tef$-triangle inequality with $\tef=\Otil{1}$ implies that $\max_{x,y\in\xset_\nu}V_x(y) = \Otil{\max_{y\in\xset_\nu}V_{x_0}(y)}$.}
Recalling that $r'=\Otil{r}$ and $\eps'=\Otil{\eps}$, \Cref{thm:grad-est} gives that, with probability at least $1-\frac{1}{100\Tmax}$ the runtime of a restricted oracle call is at most
\[
	\Tinner = \widetilde{O}\left(n\frac{\Lf r^2}{\eps^2}+d\frac{\Lf r}{\eps}+n(\Teval+d)+(\Teval+\Tmd+d)\frac{\GG^2}{\rho^2}\right).
\]

By \Cref{thm:outerloop,thm:grad-est} we have that $\rho = \widetilde{\Theta}(R^{2/3}r^{1/3})$ and $\GG=\Otilb{\frac{\Lf r^{2/3} R^{4/3}}{\epsilon}}$. Moreover, the number of ball oracle calls is bounded by $\Tmax = \widetilde{O}(R^{2/3}r^{-2/3})$ with probability at least $\frac{99}{100}$. Substituting and applying a union bound, we get that the total runtime of the algorithm is bounded by
\[
	\Tmax\cdot  \Tinner = \widetilde{O}\left(n\frac{\Lf^2R^{2/3}r^{4/3}}{\eps^2}+n(\Teval+d)\frac{R^{2/3}}{r^{2/3}}+(\Teval+\Tmd+d)\frac{\Lf^2R^2}{\eps^{2}}\right)
\]
with probability at least $\frac{9}{10}$. Substituting $r = \min\crl*{\sqrt{\frac{\epsilon}{\Lg \log n}}, \frac{\epsilon \sqrt{\Teval +d}}{\Lf}}$ yields the claimed bound~\eqref{eq:computation-general} and completes the proof.
\end{proof}

\subsection{Matrix games}\label{subsec:runtimes-linear}
In the special case where $f_i(x) = [A^\top x]_i$ are linear functions, the ball and simplex setups reduce to $\ell_p$-$\ell_1$ matrix games with $p\in\{2,1\}$, respectively. Formally, the problem definition is
\begin{equation}\label{def:matrix-games}
	\minimize_{x\in\xset}\left[\max_{y\in\Delta^n}x^\top Ay\right],~~\text{where}~\xset = \Delta^d~\text{for}~\ell_1\text{-}\ell_1~\text{and}~\xset = \ball^d~\text{for}~\ell_2\text{-}\ell_1.
\end{equation}
To simplify expressions, we assume that each $f_i(x)$ is $1$-Lipschitz in $\norm{\cdot}_p$, which is equivalent to assuming that
\begin{equation}\label{def:matrix-games-assump}
	\norm{A}_{p\rightarrow \infty} = \begin{cases}
 	\max_{j,i}|A_{ji}|~~&~~\text{for}~\ell_1\text{-}\ell_1~\text{games}\\
 	\max_{i\in[n]}\|A_{:i}\|_2~~&~~\text{for}~\ell_2\text{-}\ell_1~\text{games} \le 1
 \end{cases}
\end{equation}
Our runtime guarantees are as follows.

\begin{corollary}[Matrix games]\label{coro:matrix-games}
	For $p\in\{1,2\}$, consider the problem of $\ell_p$-$\ell_1$ matrix games~\eqref{def:matrix-games} under the assumption \eqref{def:matrix-games-assump}. For $\eps\in(0,1)$ and $\nu=\eps/(4d)$, \Cref{alg:framework} with parameters $r = \min(1,\sqrt{d}\eps)$, $R=\widetilde{O}(1)$, $\mathcal{E}_0=R$, accuracy $\eps/4$, ball oracle implementation in~\Cref{alg:ball-oralce} and gradient oracle implementation in~\Cref{alg:grad-est}, return a point $x$ such that
	\[\E \min_{x\in\xset}\left[\max_{y\in\Delta^n}x^\top Ay\right]-\min_{x_\star\in\xset}\left[\max_{y_\star\in\Delta^n}x_\star^\top Ay_\star\right]\le \eps.\] 
	With probability at least $\tfrac{9}{10}$ the runtime of the algorithm is \[\widetilde{O}\left(nd+nd^{2/3}\frac{1}{\eps^{2/3}}+d\frac{1}{\eps^{2}}\right).\]
\end{corollary}

\begin{proof}

We invoke~\Cref{thm:mtm} with $\Lf=1$ (by assumption) and $\Lg = 0$ (since each function is linear). For matrix games we have $\Teval=O(d)$. Let us also argue that $\Tmd=\Otil{d}$, recalling that $\Tmd$ is the time to find $w=\argmin_{w\in\xset}\crl*{\langle g, w\rangle+\lambda V_y(w)+V_{z}(w)}$ for some $y,z\in\xset_\nu$ and $g\in\R^d$. In the ball setup we simply have
\[
	w = \Pi_{\ball^d}\prn*{\frac{z + \lambda y -  g}{1+\lambda}}~~\mbox{where}~~\Pi_{\ball^d}(x) = \frac{x}{\max\{1, \norm{x}\}}
\]
is the Euclidean projection onto $\ball^d$. Therefore, $\Tmd=O(d)$ in the ball setup. 

In the (truncated) simplex setup $\xset = \Delta_\nu^d$ with some $\nu\in(0,1/2d]$, we can implement the mirror descent step as follows. Let $\xi = z^{\tfrac{1}{1+\lambda}}\circ  y^{\tfrac{\lambda}{1+\lambda}}\circ \exp(-\frac{1}{1+\lambda}g)$, where we use $\circ$ to represent element-wise product. Let $\sigma$ be a permutation of $(1,\ldots, d)$ such that $\xi_{\sigma_i}$ is the $i$-th largest entry of $\xi$ (breaking ties arbitrarily). Now define $\alpha_i = \frac{\nu \sum_{j\le i} \xi_{\sigma_j}}{1-\nu(d-i)}$ (so that $\frac{\alpha_i}{\sum_{j\le i} \xi_{\sigma_j}+\alpha_i(d-i)} = \nu$), and the cutoff index $i'\in[d]$ to be the largest $i\in [d]$ such that $\frac{\xi_{\sigma_{i}}}{\sum_{j\le {i}}\xi_{\sigma_j}} \ge\frac{\nu}{1-\nu(d-{i})}$. Such $i'\in[d]$ must be well-defined as the inequality is satisfied when $i=1$. It is then straightforward to verify that $w \in \R^d$ such that for all $i \in [d]$
\[
	w_i = \begin{cases}
	\frac{\nu}{\alpha_{i'}}\cdot \xi_{\sigma_i}~&~\text{if}~i\le i',\\
	\nu~&~\text{if}~i>i'\\
	\end{cases}
\] 
is the solution to the problem defining the mirror descent step. Computing $\xi$ takes $O(d)$ time, sorting it takes $O(d\log d)$ time, and finding $i'$ and calculating $w$ each take additional $O(d)$ time, so overall $\Tmd = \Otil{d}$ in the simplex setup.

Plugging $\Lf=1$, $\Lg=0$, and $\Teval,\Tmax=\Otil{d}$ into \Cref{eq:computation-general} yields the claimed runtime bound.
\end{proof}

\subsection{Minimum Enclosing Ball}\label{subsec:runtimes-minEB}

In this section, we apply our method to solving the minimum enclosing ball problem, defined as follows. Given data points $a_1,\ldots,a_n\in\R^d$ such that $a_1=0$ and $\max_{i\in[n]}\|a_i\|_2=1$, the goal is to find the minimum radius $\Ropt$ ball containing all data points. 
That is, 
\begin{equation}\label{eqn:minEB}
\half \Ropt^2 = \min_{x\in\R^d}\max_{y\in\Delta^n} f_i(x)~~\mbox{where}~~f_i(x) = \half \norm{x-a_i}^2_2.
\end{equation}
The problem is also equivalent to an $\ell_2$-$\ell_1$ matrix game with a quadratic regularization term, but for our purpose the natural formulation above is more convenient. Letting $\xopt \defeq \argmin_{x\in\R^d}\max_{y\in\Delta^n} f_i(x)$, it holds without loss of generality that $\|x_\star\|_2\le 1$ and $\Ropt \in [\half, 1]$  (see~\citet{allen2016optimization} for detailed explanation). Under these assumptions, we obtain the following runtime guarantee.

\begin{corollary}[minimum enclosing ball]\label{coro:minEB}
	Consider the problem~\eqref{eqn:minEB} with $a_1=0$ and $\max_{i\in[n]}\norm{a_i}_2 \le 1$ (so that $\norm{\xopt}\le 1$ and $\Ropt \ge 1/2$). For any $\eps\in(0,1)$, there is an algorithm that makes $\Otil{1}$ calls to~\Cref{alg:framework} with ball oracle implementation~\Cref{alg:ball-oralce} and gradient oracle implementation in~\Cref{alg:grad-est} and, with probability at least $\frac{9}{10}$ returns a point $x$ such that
	\[
		\frac{1}{2}\|x-x_\star\|_2^2\le \epsilon\cdot\Ropt^2
	\]
	with total runtime 
	\[
		\widetilde{O}\left(nd+nd^{2/3}\eps^{-1/3}+d\eps^{-1}\right).
	\]
\end{corollary}

\begin{proof}
Let $K=\log_2 \frac{4}{\epsilon}$. 
We use \Cref{thm:mtm} with $f_i(x) =\frac{1}{2}\|x-a_i\|_2^2$ defined above, and boost its result to failure probability $\frac{1}{10K}$ by repeatedly calling the algorithm $\Otil{1}$ times, cutting it off whenever it exceeds the runtime bound, and selecting the best result in $\Otil{nd}$ time. We apply this high-probability solver recursively, generating a sequence of solutions $x\pind{0}, \ldots, x\pind{K}$ that satisfies, with probability at least $\frac{9}{10}$,
\[
	\half \norm{x\pind{k}-\xopt}^2_2 \le 2^{-(k+1)} \le 2^{-k} 2\Ropt^2 ~~\mbox{for all}~~k\le K,
\]
so that $x=x\pind{K}$ satisfies $\frac{1}{2}\|x-x_\star\|_2^2\le \epsilon\cdot\Ropt^2$ as required.

To generate $x\pind{0}, \ldots, x\pind{K}$, we start with $x\pind{0}=0$, which satisfies $\frac{1}{2}\|x^{(0)}-x_\star\|_2^2 \le \frac{1}{2}\le 2\Ropt^2$ by assumption. To produce $x\pind{k}$ for $k\ge 1$ we apply our algorithm on with parameters $R_k = 2^{-(k-1)/2}$, $\epsilon_k = 2^{-(k+1)}$, $\Lg=1$ and $\Lf=O(1)$ on the domain $\xset_k = \crl*{x \mid \norm{x-x\pind{k-1}} \le 2^{-(k-1)/2}}$, which contains $\xopt$ by the inductive assumption that $\norm{x\pind{k-1}-\xopt}_2 \le 2^{-(k-1)/2}$. The $1$-strong-convexity of our objective function then guarantees (with the appropriate probability) that $\frac{1}{2}\|x^{(k)}-x_\star\|_2^2 \le \epsilon_k = 2^{-(k+1)}$, completing the induction. The runtime to produce $x\pind{k}$ is
\begin{align*}
& \Otilb{
			n (\Teval +d) \prn*{\frac{R_k^2}{\epsilon_k}}^{1/3} + n \prn*{\frac{(\Teval + d)R_k}{\epsilon_k}}^{2/3} + (\Teval + \Tmd + d)\prn*{\frac{R_k}{\epsilon_k}}^{2} 
			}\\
			& \hspace{3em}= \Otilb{
			nd + n d^{2/3}\cdot 2^{k/3} + d\cdot 2^{k}			},
\end{align*}
where the transition follows from substituting $R_k,\eps_k$, and plugging in $\Teval = \Tmd = O(d)$. Summing this over $k\in [K]$ and recalling that $2^{K} = O(\frac{1}{\epsilon})$ yields the claimed runtime bound.
\end{proof} 

\section*{Acknowledgments}
We thank Kfir Levy for suggesting the work~\cite{cutkosky2019anytime} may be useful for ball oracle acceleration, and the anonymous reviewers for their helpful feedback.

YC was supported in part by the Israeli Science Foundation (ISF) grant no.\ 2486/21 and the Len Blavatnik and the Blavatnik Family foundation. YJ was supported in part by a Stanford Graduate Fellowship and the Dantzig-Lieberman Fellowship. 
AS was supported in part by a Microsoft Research Faculty Fellowship, NSF CAREER Award CCF-1844855, NSF Grant CCF-1955039, a PayPal research award, and a Sloan Research Fellowship.

\arxiv{
	\setlength{\bibsep}{6pt}
}

\bibliographystyle{abbrvnat}

\begin{thebibliography}{66}
\providecommand{\natexlab}[1]{#1}
\providecommand{\url}[1]{\texttt{#1}}
\expandafter\ifx\csname urlstyle\endcsname\relax
  \providecommand{\doi}[1]{doi: #1}\else
  \providecommand{\doi}{doi: \begingroup \urlstyle{rm}\Url}\fi

\bibitem[Adler(2013)]{adler2013equivalence}
I.~Adler.
\newblock The equivalence of linear programs and zero-sum games.
\newblock \emph{International Journal of Game Theory}, 42:\penalty0 165--177,
  2013.

\bibitem[Allen-Zhu et~al.(2016)Allen-Zhu, Liao, and
  Yuan]{allen2016optimization}
Z.~Allen-Zhu, Z.~Liao, and Y.~Yuan.
\newblock Optimization algorithms for faster computational geometry.
\newblock In \emph{International Colloquium of Automata, Languages and
  Programming}, 2016.

\bibitem[Asi et~al.(2021)Asi, Carmon, Jambulapati, Jin, and
  Sidford]{asi2021stochastic}
H.~Asi, Y.~Carmon, A.~Jambulapati, Y.~Jin, and A.~Sidford.
\newblock Stochastic bias-reduced gradient methods.
\newblock \emph{Advances in Neural Information Processing Systems}, 34, 2021.

\bibitem[Axelrod et~al.(2020)Axelrod, Liu, and Sidford]{axelrod2020near}
B.~Axelrod, Y.~P. Liu, and A.~Sidford.
\newblock Near-optimal approximate discrete and continuous submodular function
  minimization.
\newblock In \emph{Symposium on Discrete Algorithms, (SODA)}, 2020.

\bibitem[Balamurugan and Bach(2016)]{balamurugan2016stochastic}
P.~Balamurugan and F.~Bach.
\newblock Stochastic variance reduction methods for saddle-point problems.
\newblock In \emph{Advances in Neural Information Processing Systems
  (NeurIPS)}, 2016.

\bibitem[Blackwell(1997)]{blackwell1997large}
D.~Blackwell.
\newblock Large deviations for martingales.
\newblock In \emph{Festschrift for Lucien Le Cam}, 1997.

\bibitem[Blanchet and Glynn(2015)]{blanchet2015unbiased}
J.~H. Blanchet and P.~W. Glynn.
\newblock Unbiased {M}onte {C}arlo for optimization and functions of
  expectations via multi-level randomization.
\newblock In \emph{2015 Winter Simulation Conference (WSC)}, pages 3656--3667,
  2015.

\bibitem[Bubeck et~al.(2019{\natexlab{a}})Bubeck, Jiang, Lee, Li, and
  Sidford]{bubeck2019complexity}
S.~Bubeck, Q.~Jiang, Y.~T. Lee, Y.~Li, and A.~Sidford.
\newblock Complexity of highly parallel non-smooth convex optimization.
\newblock \emph{arXiv:1906.10655}, 2019{\natexlab{a}}.

\bibitem[Bubeck et~al.(2019{\natexlab{b}})Bubeck, Jiang, Lee, Li, and
  Sidford]{bubeck2019optimal}
S.~Bubeck, Q.~Jiang, Y.~T. Lee, Y.~Li, and A.~Sidford.
\newblock Near-optimal method for highly smooth convex optimization.
\newblock In \emph{Proceedings of the Thirty Second Annual Conference on
  Computational Learning Theory}, pages 492--507, 2019{\natexlab{b}}.

\bibitem[Bullins(2020)]{bullins2020highly}
B.~Bullins.
\newblock Highly smooth minimization of non-smooth problems.
\newblock In \emph{Conference on Learning Theory}, pages 988--1030, 2020.

\bibitem[Carmon and Hausler(2022)]{carmon2022distributionally}
Y.~Carmon and D.~Hausler.
\newblock Distributionally robust optimization via ball oracle acceleration.
\newblock \emph{arXiv:2203.13225}, 2022.

\bibitem[Carmon and Hinder(2022)]{carmon2022making}
Y.~Carmon and O.~Hinder.
\newblock Making {SGD} parameter-free.
\newblock In \emph{Conference on Learning Theory (COLT)}, 2022.

\bibitem[Carmon et~al.(2019)Carmon, Jin, Sidford, and Tian]{carmon2019variance}
Y.~Carmon, Y.~Jin, A.~Sidford, and K.~Tian.
\newblock Variance reduction for matrix games.
\newblock \emph{Advances in Neural Information Processing Systems (NeurIPS)},
  2019.

\bibitem[Carmon et~al.(2020{\natexlab{a}})Carmon, Jambulapati, Jiang, Jin, Lee,
  Sidford, and Tian]{carmon2020acceleration}
Y.~Carmon, A.~Jambulapati, Q.~Jiang, Y.~Jin, Y.~T. Lee, A.~Sidford, and
  K.~Tian.
\newblock Acceleration with a ball optimization oracle.
\newblock In \emph{Advances in Neural Information Processing Systems},
  2020{\natexlab{a}}.

\bibitem[Carmon et~al.(2020{\natexlab{b}})Carmon, Jin, Sidford, and
  Tian]{carmon2020coordinate}
Y.~Carmon, Y.~Jin, A.~Sidford, and K.~Tian.
\newblock Coordinate methods for matrix games.
\newblock In \emph{Symposium on Foundations of Computer Science (FOCS)},
  2020{\natexlab{b}}.

\bibitem[Carmon et~al.(2021)Carmon, Jambulapati, Jin, and
  Sidford]{carmon2021thinking}
Y.~Carmon, A.~Jambulapati, Y.~Jin, and A.~Sidford.
\newblock Thinking inside the ball: Near-optimal minimization of the maximal
  loss.
\newblock In \emph{Conference on Learning Theory}, 2021.

\bibitem[Carmon et~al.(2022{\natexlab{a}})Carmon, Hausler, Jambulapati, Jin,
  and Sidford]{carmon2022optimal}
Y.~Carmon, D.~Hausler, A.~Jambulapati, Y.~Jin, and A.~Sidford.
\newblock Optimal and adaptive {M}onteiro-{S}vaiter acceleration.
\newblock In \emph{Advances in Neural Information Processing Systems
  (NeurIPS)}, 2022{\natexlab{a}}.

\bibitem[Carmon et~al.(2022{\natexlab{b}})Carmon, Jambulapati, Jin, and
  Sidford]{carmon2022recapp}
Y.~Carmon, A.~Jambulapati, Y.~Jin, and A.~Sidford.
\newblock {RECAPP}: Crafting a more efficient catalyst for convex optimization.
\newblock In \emph{International Conference on Machine Learning (ICML)},
  2022{\natexlab{b}}.

\bibitem[Charikar et~al.(2018)Charikar, Chen, and
  Farach{-}Colton]{charikar2002finding}
M.~Charikar, K.~C. Chen, and M.~Farach{-}Colton.
\newblock Finding frequent items in data streams.
\newblock In \emph{International Colloquium of Automata, Languages and
  Programming}, 2018.

\bibitem[Chernoff(1952)]{chernoff1952measure}
H.~Chernoff.
\newblock A measure of asymptotic efficiency for tests of a hypothesis based on
  the sum of observations.
\newblock \emph{The Annals of Mathematical Statistics}, pages 493--507, 1952.

\bibitem[Clarkson et~al.(2012)Clarkson, Hazan, and
  Woodruff]{clarkson2012sublinear}
K.~L. Clarkson, E.~Hazan, and D.~P. Woodruff.
\newblock Sublinear optimization for machine learning.
\newblock \emph{Journal of the ACM (JACM)}, 59\penalty0 (5):\penalty0 1--49,
  2012.

\bibitem[Cohen et~al.(2021)Cohen, Lee, and Song]{cohen2021solving}
M.~B. Cohen, Y.~T. Lee, and Z.~Song.
\newblock Solving linear programs in the current matrix multiplication time.
\newblock \emph{Journal of the ACM (JACM)}, 68\penalty0 (1):\penalty0 1--39,
  2021.

\bibitem[Cutkosky(2019)]{cutkosky2019anytime}
A.~Cutkosky.
\newblock Anytime online-to-batch, optimism and acceleration.
\newblock In \emph{International Conference on Machine Learning (ICML)}, 2019.

\bibitem[Dantzig(1953)]{dantzig1953linear}
G.~B. Dantzig.
\newblock \emph{Linear Programming and Extensions}.
\newblock Princeton University Press, Princeton, NJ, 1953.

\bibitem[Duchi et~al.(2008)Duchi, Shalev-Shwartz, Singer, and
  Chandra]{duchi2008efficient}
J.~Duchi, S.~Shalev-Shwartz, Y.~Singer, and T.~Chandra.
\newblock Efficient projections onto the l 1-ball for learning in high
  dimensions.
\newblock In \emph{International Conference on Machine Learning (ICML)}, 2008.

\bibitem[Frostig et~al.(2015)Frostig, Ge, Kakade, and
  Sidford]{frostig2015unregularizing}
R.~Frostig, R.~Ge, S.~Kakade, and A.~Sidford.
\newblock Un-regularizing: approximate proximal point and faster stochastic
  algorithms for empirical risk minimization.
\newblock In \emph{International Conference on Machine Learning (ICML)}, 2015.

\bibitem[Gasnikov et~al.(2019)Gasnikov, Dvurechensky, Gorbunov, Vorontsova,
  Selikhanovych, and Uribe]{gasnikov2019optimal}
A.~V. Gasnikov, P.~E. Dvurechensky, E.~Gorbunov, E.~A. Vorontsova,
  D.~Selikhanovych, and C.~A. Uribe.
\newblock Optimal tensor methods in smooth convex and uniformly convex
  optimization.
\newblock In \emph{Proceedings of the Thirty Second Annual Conference on
  Computational Learning Theory}, pages 1374--1391, 2019.

\bibitem[Grigoriadis and Khachiyan(1995)]{grigoriadis1995sublinear}
M.~D. Grigoriadis and L.~G. Khachiyan.
\newblock A sublinear-time randomized approximation algorithm for matrix games.
\newblock \emph{Operations Research Letters}, 18\penalty0 (2):\penalty0 53--58,
  1995.

\bibitem[G{\"u}ler(1992)]{guler1992new}
O.~G{\"u}ler.
\newblock New proximal point algorithms for convex minimization.
\newblock \emph{SIAM Journal on Optimization}, 2\penalty0 (4):\penalty0
  649--664, 1992.

\bibitem[Ivgi et~al.(2023)Ivgi, Hinder, and Carmon]{ivgi2023dog}
M.~Ivgi, O.~Hinder, and Y.~Carmon.
\newblock {D}o{G} is {SGD}'s best friend: A parameter-free dynamic step size
  schedule.
\newblock In \emph{International Conference on Machine Learning (ICML)}, 2023.

\bibitem[Jiang et~al.(2019)Jiang, Wang, and Zhang]{jiang2019optimal}
B.~Jiang, H.~Wang, and S.~Zhang.
\newblock An optimal high-order tensor method for convex optimization.
\newblock In \emph{Proceedings of the Thirty Second Annual Conference on
  Computational Learning Theory}, pages 1799--1801, 2019.

\bibitem[Jiang et~al.(2021)Jiang, Song, Weinstein, and Zhang]{Jiang0WZ21}
S.~Jiang, Z.~Song, O.~Weinstein, and H.~Zhang.
\newblock A faster algorithm for solving general lps.
\newblock In S.~Khuller and V.~V. Williams, editors, \emph{Proceedings of the
  Fifty-Third Annual ACM Symposium on the Theory of Computing}, pages 823--832.
  {ACM}, 2021.

\bibitem[Karmarkar(1984)]{karmarkar84new}
N.~Karmarkar.
\newblock A new polynomial-time algorithm for linear programming.
\newblock In \emph{Symposium on Theory of Computing (STOC)}, pages 302--311,
  1984.

\bibitem[Kovalev and Gasnikov(2022)]{kovalev2022first}
D.~Kovalev and A.~Gasnikov.
\newblock The first optimal acceleration of high-order methods in smooth convex
  optimization.
\newblock In \emph{Advances in Neural Information Processing Systems
  (NeurIPS)}, 2022.

\bibitem[Lan(2016)]{lan2016gradient}
G.~Lan.
\newblock Gradient sliding for composite optimization.
\newblock \emph{Mathematical Programming}, 159\penalty0 (1):\penalty0 201--235,
  2016.

\bibitem[Lan and Ouyang(2022)]{lan22accelerated}
G.~Lan and Y.~Ouyang.
\newblock Accelerated gradient sliding for structured convex optimization.
\newblock \emph{Computational Optimization and Applications}, 82\penalty0
  (2):\penalty0 361--394, 2022.

\bibitem[Larsen et~al.(2021)Larsen, Pagh, and Tetek]{larsen2021count}
K.~G. Larsen, R.~Pagh, and J.~Tetek.
\newblock Count{S}ketches, feature hashing and the median of three.
\newblock In \emph{International Conference on Machine Learning (ICML)}, 2021.

\bibitem[Lee and Sidford(2013)]{lee2013efficient}
Y.~T. Lee and A.~Sidford.
\newblock Efficient accelerated coordinate descent methods and faster
  algorithms for solving linear systems.
\newblock In \emph{Symposium on Foundations of Computer Science (FOCS)}, pages
  147--156, 2013.

\bibitem[Lee and Sidford(2015)]{lee2015efficient}
Y.~T. Lee and A.~Sidford.
\newblock Efficient inverse maintenance and faster algorithms for linear
  programming.
\newblock In \emph{Symposium on Foundations of Computer Science (FOCS)}, 2015.

\bibitem[Lin et~al.(2015)Lin, Mairal, and Harchaoui]{lin2015universal}
H.~Lin, J.~Mairal, and Z.~Harchaoui.
\newblock A universal catalyst for first-order optimization.
\newblock In \emph{Advances in Neural Information Processing Systems
  (NeurIPS)}, 2015.

\bibitem[Minsky and Papert(1988)]{minsky1988perceptrons}
M.~Minsky and S.~Papert.
\newblock \emph{Perceptrons: An introduction to computational geometry}.
\newblock MIT Press, 1988.

\bibitem[Monteiro and Svaiter(2012)]{monteiro2012iteration}
R.~D. Monteiro and B.~F. Svaiter.
\newblock Iteration-complexity of a {N}ewton proximal extragradient method for
  monotone variational inequalities and inclusion problems.
\newblock \emph{SIAM Journal on Optimization}, 22\penalty0 (3):\penalty0
  914--935, 2012.

\bibitem[Monteiro and Svaiter(2013)]{monteiro2013accelerated}
R.~D.~C. Monteiro and B.~F. Svaiter.
\newblock An accelerated hybrid proximal extragradient method for convex
  optimization and its implications to second-order methods.
\newblock \emph{SIAM Journal on Optimization}, 23\penalty0 (2):\penalty0
  1092--1125, 2013.

\bibitem[Namkoong and Duchi(2016)]{namkoong2016stochastic}
H.~Namkoong and J.~C. Duchi.
\newblock Stochastic gradient methods for distributionally robust optimization
  with $f$-divergences.
\newblock In \emph{Advances in Neural Information Processing Systems
  (NeurIPS)}, 2016.

\bibitem[Nemirovski(2004)]{nemirovski2004prox}
A.~Nemirovski.
\newblock Prox-method with rate of convergence $o(1/t)$ for variational
  inequalities with {L}ipschitz continuous monotone operators and smooth
  convex-concave saddle point problems.
\newblock \emph{SIAM Journal on Optimization}, 15\penalty0 (1):\penalty0
  229--251, 2004.

\bibitem[Nemirovski et~al.(2009)Nemirovski, Juditsky, Lan, and
  Shapiro]{nemirovski2009robust}
A.~Nemirovski, A.~Juditsky, G.~Lan, and A.~Shapiro.
\newblock Robust stochastic approximation approach to stochastic programming.
\newblock \emph{SIAM Journal on optimization}, 19\penalty0 (4):\penalty0
  1574--1609, 2009.

\bibitem[Nesterov(2005)]{nesterov2005smooth}
Y.~Nesterov.
\newblock Smooth minimization of non-smooth functions.
\newblock \emph{Mathematical programming}, 103:\penalty0 127--152, 2005.

\bibitem[Nesterov(2007)]{nesterov2007dual}
Y.~Nesterov.
\newblock Dual extrapolation and its applications to solving variational
  inequalities and related problems.
\newblock \emph{Mathematical Programming}, 109\penalty0 (2-3):\penalty0
  319--344, 2007.

\bibitem[Nesterov(2018)]{nesterov2018lectures}
Y.~Nesterov.
\newblock \emph{Lectures on convex optimization}.
\newblock Springer, 2018.

\bibitem[Reiss(2012)]{reiss2012approximate}
R.-D. Reiss.
\newblock \emph{Approximate distributions of order statistics: with
  applications to nonparametric statistics}.
\newblock Springer science \& business media, 2012.

\bibitem[Renegar(1988)]{renegar1988polynomial}
J.~Renegar.
\newblock A polynomial-time algorithm, based on {N}ewton's method, for linear
  programming.
\newblock \emph{Mathematical programming}, 40\penalty0 (1-3):\penalty0 59--93,
  1988.

\bibitem[Rockafellar(1997)]{rockafellar1997convex}
R.~T. Rockafellar.
\newblock \emph{Convex analysis}.
\newblock Princeton university press, 1997.

\bibitem[Salzo and Villa(2012)]{salzo2012inexact}
S.~Salzo and S.~Villa.
\newblock Inexact and accelerated proximal point algorithms.
\newblock \emph{Journal of Convex analysis}, 19\penalty0 (4):\penalty0
  1167--1192, 2012.

\bibitem[Shalev-Shwartz and Wexler(2016)]{shalev2016minimizing}
S.~Shalev-Shwartz and Y.~Wexler.
\newblock Minimizing the maximal loss: How and why?
\newblock In \emph{International Conference on Machine Learning (ICML)}, 2016.

\bibitem[Sidford and Tian(2018)]{sidford2018coordinate}
A.~Sidford and K.~Tian.
\newblock Coordinate methods for accelerating $\ell_\infty$ regression and
  faster approximate maximum flow.
\newblock In \emph{Symposium on Foundations of Computer Science (FOCS)}, 2018.

\bibitem[Song et~al.(2021)Song, Wright, and Diakonikolas]{song2021variance}
C.~Song, S.~J. Wright, and J.~Diakonikolas.
\newblock Variance reduction via primal-dual accelerated dual averaging for
  nonsmooth convex finite-sums.
\newblock In \emph{International Conference on Machine Learning}, 2021.

\bibitem[Song et~al.(2022)Song, Lin, Wright, and
  Diakonikolas]{song2022coordinate}
C.~Song, C.~Y. Lin, S.~Wright, and J.~Diakonikolas.
\newblock Coordinate linear variance reduction for generalized linear
  programming.
\newblock In \emph{Advances in Neural Information Processing Systems
  (NeurIPS)}, 2022.

\bibitem[Sylvester(1857)]{sylvester1857question}
J.~J. Sylvester.
\newblock A question in the geometry of situation.
\newblock \emph{Quarterly Journal of Pure and Applied Mathematics}, 1\penalty0
  (1):\penalty0 79--80, 1857.

\bibitem[Thekumparampil et~al.(2020)Thekumparampil, Jain, Netrapalli, and
  Oh]{thekumparampil2020projection}
K.~K. Thekumparampil, P.~Jain, P.~Netrapalli, and S.~Oh.
\newblock Projection efficient subgradient method and optimal nonsmooth
  frank-wolfe method.
\newblock In \emph{Advances in Neural Information Processing Systems}, 2020.

\bibitem[van~den Brand(2020)]{Brand20}
J.~van~den Brand.
\newblock A deterministic linear program solver in current matrix
  multiplication time.
\newblock pages 259--278. {SIAM}, 2020.

\bibitem[van~den Brand(2021)]{Brand21}
J.~van~den Brand.
\newblock Unifying matrix data structures: Simplifying and speeding up
  iterative algorithms.
\newblock In H.~V. Le and V.~King, editors, \emph{4th Symposium on Simplicity
  in Algorithms (SOSA)}, pages 1--13. {SIAM}, 2021.

\bibitem[van~den Brand et~al.(2020{\natexlab{a}})van~den Brand, Lee, Nanongkai,
  Peng, Saranurak, Sidford, Song, and Wang]{brand2020bipartite}
J.~van~den Brand, Y.~T. Lee, D.~Nanongkai, R.~Peng, T.~Saranurak, A.~Sidford,
  Z.~Song, and D.~Wang.
\newblock Bipartite matching in nearly-linear time on moderately dense graphs.
\newblock In \emph{Symposium on Foundations of Computer Science (FOCS)},
  2020{\natexlab{a}}.

\bibitem[van~den Brand et~al.(2020{\natexlab{b}})van~den Brand, Lee, Sidford,
  and Song]{brand2020solving}
J.~van~den Brand, Y.~T. Lee, A.~Sidford, and Z.~Song.
\newblock Solving tall dense linear programs in nearly linear time.
\newblock In K.~Makarychev, Y.~Makarychev, M.~Tulsiani, G.~Kamath, and
  J.~Chuzhoy, editors, \emph{Proceedings of the Fifty-Second Annual ACM
  Symposium on the Theory of Computing}, 2020{\natexlab{b}}.

\bibitem[Van Den~Brand et~al.(2021)Van Den~Brand, Lee, Liu, Saranurak, Sidford,
  Song, and Wang]{brand2021minimum}
J.~Van Den~Brand, Y.~T. Lee, Y.~P. Liu, T.~Saranurak, A.~Sidford, Z.~Song, and
  D.~Wang.
\newblock Minimum cost flows, {MDP}s, and $\ell_1$-regression in nearly linear
  time for dense instances.
\newblock In \emph{Symposium on Theory of Computing (STOC)}, 2021.

\bibitem[Vose(1991)]{vose91sample}
M.~Vose.
\newblock A linear algorithm for generating random numbers with a given
  distribution.
\newblock \emph{IEEE Transactions on Software Engineering}, 17\penalty0
  (9):\penalty0 972--975, 1991.

\bibitem[Wang(2020)]{wang2020randomized}
M.~Wang.
\newblock Randomized linear programming solves the {M}arkov decision problem in
  nearly linear (sometimes sublinear) time.
\newblock \emph{Mathematics of Operations Research}, 45\penalty0 (2):\penalty0
  517--546, 2020.

\end{thebibliography}

\newpage
\appendix

\end{document}